\documentclass[opre,nonblindrev]{informs3NoHeading} 

\OneAndAHalfSpacedXI 

\usepackage[dvipsnames]{xcolor}

\usepackage{natbib}
 \bibpunct[, ]{(}{)}{,}{a}{}{,}%
 \def\bibfont{\small}%
 \def\bibsep{\smallskipamount}%
 \def\bibhang{24pt}%

\usepackage[margin=1in]{geometry}
\usepackage{bm}
\usepackage{algorithm,algpseudocode}
\usepackage{diagbox}
\usepackage{multirow}
\usepackage{enumerate}
\usepackage{mathtools}
\usepackage{mathrsfs}
\usepackage{enumitem}
\usepackage{dsfont}
\usepackage{float}
\usepackage{amssymb}
\usepackage{lscape}
\usepackage{graphicx}
\usepackage{amsmath}
\usepackage{bbm}
\usepackage{comment}
\usepackage{tikz}
\usepackage{booktabs}
\usepackage{enumitem}
\usepackage{natbib}
\usepackage{caption}
\usepackage[colorlinks,linkcolor=blue,citecolor=blue,hypertexnames=false]{hyperref}
\usepackage[labelformat=simple]{subcaption}

\usetikzlibrary{calc,positioning}

\tikzset{round/.style = { rounded corners=2mm }}

\usepackage{cleveref}
\usepackage[english]{babel}
\usepackage[autostyle, english = american]{csquotes}
\MakeOuterQuote{"}

\newcommand{\sgk}{\gamma^*_k}

\newcommand{\p}{p}

\newcommand{\price}{\boldsymbol{p}}

\newcommand{\beq}{\begin{eqnarray}}
\newcommand{\eeq}{\end{eqnarray}}
\newcommand{\beqn}{\begin{eqnarray*}}
\newcommand{\eeqn}{\end{eqnarray*}}

\newcommand{\LCkunit}{\mathsf{LP}^{\mathsf{OCRS}}_{k}}

\newcommand{\DCkunit}{\mathsf{Dual}^{\mathsf{OCRS}}_{k}}

\newcommand{\bI}{\mathbbm{1}}

\theoremstyle{definition}
\newtheorem{theorem}{Theorem}

\newtheorem{lemma}{Lemma}
\newtheorem{proposition}{Proposition}
\newtheorem{claim}{Claim}

\newtheorem{defn}{Definition}

\bibpunct{(}{)}{;}{a}{,}{,}

\allowdisplaybreaks

\EquationsNumberedThrough    

\begin{document}


\RUNAUTHOR{Jiang, Ma and Zhang}

\RUNTITLE{Multi-unit Prophet Inequalities and Online Knapsack}

\TITLE{Tight Guarantees for Multi-unit Prophet Inequalities and Online Stochastic Knapsack}

\ARTICLEAUTHORS{%
\AUTHOR{$\text{Jiashuo Jiang}^\dagger$ \quad $\text{Will Ma}^\ddagger$ \quad $\text{Jiawei Zhang}^\S$}

\AFF{\  \\
$\dagger~$Department of Industrial Engineering \& Decision Analytics, Hong Kong University of Science and Technology\\
$\ddagger~$Graduate School of Business and Data Science Institute, Columbia University\\
$\S~$Department of Technology, Operations \& Statistics, Stern School of Business, New York University\\
}
}

\ABSTRACT{

\textbf{Abstract:}
Prophet inequalities are a useful tool for designing online allocation procedures and comparing their performance to the optimal offline allocation. In the basic setting of $k$-unit prophet inequalities, the well-known procedure of \citet{alaei2011bayesian} with its celebrated performance guarantee of $1-\frac{1}{\sqrt{k+3}}$ has found widespread adoption in mechanism design and general online allocation problems in online advertising, healthcare scheduling, and revenue management. Despite being commonly used to derive approximately-optimal algorithms for multi-resource allocation problems, the tightness of Alaei's guarantee has remained unknown. In this paper characterize the tight guarantee in Alaei's setting, which we show is in fact strictly greater than $1-\frac{1}{\sqrt{k+3}}$ for all $k>1$.

We also consider the more general online stochastic knapsack problem where each individual allocation can consume an arbitrary fraction of the initial capacity. Here we introduce a new ``best-fit'' procedure
with a performance guarantee of $\frac{1}{3+e^{-2}}\approx0.319$, which we also show is tight with respect to the standard LP relaxation.
This improves the previously best-known guarantee of 0.2 for online knapsack.
Our analysis differs from existing ones by eschewing the need to split items into ``large'' or ``small'' based on capacity consumption, using instead an invariant for the overall utilization on different sample paths.
Finally, we refine our technique for the unit-density special case of knapsack, and improve the guarantee from 0.321 to 0.3557 in the multi-resource appointment scheduling application of \citet{stein2020advance}.
}

\HISTORY{This version from Oct 7th, 2023. A preliminary version appeared at SODA 2022.}

\maketitle
\section{Introduction}

Online resource allocation problems arise in many domains, such as posted-price mechanism design, transportation logistics, e-commerce fulfillment, online advertising, healthcare scheduling, and revenue management.
These problems can be characterized by a decision-maker facing a sequence of stochastically-generated queries, which must be irrevocably assigned to be served by a resource or rejected as they arrive online.
The resources have limited capacities, and the objective is to maximize the cumulative reward collected from serving queries over a finite time horizon.
We provide some concrete formulations of online resource allocation problems below.

\textbf{$k$-unit prophet inequalities.}
Prophet inequalities date back to \citet{krengel1978semiamarts}, and the $k$-unit version of it was pioneered by \citet{hajiaghayi2007automated,alaei2011bayesian}, with applications in posted-price mechanism design.
In this problem, there are $k$ copies of a single item (resource) and more than $k$ agents who want one.
Each agent has a valuation that is drawn independently from a known distribution.
The agents arrive sequentially and an agent's valuation is revealed upon arrival, at which point the agent must be either immediately given an item or irrevocably rejected.
Agents cannot be served once no items remain.
The objective is to maximize expected welfare, i.e.\ the sum of valuations of agents who receive an item, and compare to the expected welfare obtainable by a prophet who sees all the realized valuations in advance.
$k$-unit prophet inequalities can also be used to design posted-price mechanisms when the objective is to maximize revenue \citep{hajiaghayi2007automated,chawla2010multi}.

\textbf{Online knapsack.} Online knapsack is a classical problem in Operations Research dating back to \citet{papastavrou1996dynamic,kleywegt1998dynamic}, who called it the dynamic and stochastic knapsack problem, with applications in freight transportation, scheduling, and pricing.  Online knapsack can be viewed as a generalization of $k$-unit prophet inequalities in which arriving queries reveal both a valuation and a size.  The valuation and size of each query are drawn from a known joint distribution that is independent
(but could be heterogeneous)
across queries.  A query can be served as long as its size does not exceed the remaining resource capacity, and if served, its size is subtracted from the resource capacity and its valuation is collected as reward.  The objective is to maximize the total reward collected in expectation.
Again this can be compared to the expected reward obtainable by a prophet who sees all valuation/size realizations in advance.

\textbf{Online matching/assignment.} Online matching is the generalization of $k$-unit prophet inequalities to multiple resources, each starting with some number of units.  Queries have a separate valuation for each resource, drawn from a known distribution that could be correlated across resources but is independent across queries.  These valuations are revealed upon arrival, at which point the query must be irrevocably assigned ("matched") to a resource with units remaining, or rejected.  If the query is matched to a resource, then its valuation for that resource is collected as reward, noting that zero valuations can be used to indicate incompatibility with a resource.  The objective is to maximize the total reward collected from matching finite resources over a finite time horizon, which has applications in e-commerce fulfillment \citep{jasin2015lp} and matching impressions with bidders in online advertising \citep{alaei2012online}.

Online assignment is the further generalization of online knapsack to multiple resources, in which queries could take  a different size for each resource, and can only be assigned to a resource for which its size does does not exceed the remaining capacity.  This has applications in healthcare scheduling, where patients may take different amounts of time if assigned to doctors with different specialties, as described in \citet{stein2020advance}.

In these multi-resource problems, the comparison is against a prophet who sees all valuation/size realizations in advance and can make the optimal matching/assignment decisions in hindsight.

\subsection{Scope of this Paper}

We study the aforementioned problems, all of which fall under the most general problem of online assignment.
We always assume that valuation/size distributions are \textit{known} and \textit{independent across queries}, but otherwise place no restrictions on them.
We label the queries $t=1,\ldots,T$ and assume they arrive in that order\footnote{This is only for simplicity.  Our algorithmic results hold even if the order of queries is chosen by the adaptive adversary described in \citet{kleinberg2019matroid}; see the remarks after \Cref{constructprophet}.  (This is important in the mechanism design applications where agents may be strategic about their order of arrival.)}.
We allow the queries to have heterogeneous distributions, capturing valuation/size distributions that vary over time (but noting that realizations are still independent).

Our goal is to derive polynomial-time algorithms with \textit{guarantees} on how their expected total reward compares to the expected total reward of a prophet who sees all valuation/size realizations in advance.
This serves two purposes.
\begin{itemize}
\item In the single-resource problems, optimal or near-optimal online algorithms can be found via dynamic programming, which can be implemented in place of our algorithm.  However, our result still provides a guarantee on how much welfare the agency is extracting (using their optimal algorithm, whose reward is no less than our algorithm) compared to the alternative of waiting for all agents to arrive before committing to any allocations to achieve the prophet's reward. The magnitude of this guarantee provides insights for the higher-order decision of "Should the agency make agents wait until the end, in order to achieve higher social welfare?"
\item In the multi-resource problems, dynamic programming is intractable due to the curse of dimensionality (the state space is exponential in the number of resources).  For these problems, our polynomial-time algorithms come with guarantees on how well they approximate the optimal dynamic program (in fact, how well they approximate the stronger prophet benchmark).
\end{itemize}
All of our comparisons and guarantees are in terms of \textit{ratios}.

Finally, our paper follows a well-known framework of reducing multi-resource matching/assignment problems to single-resource accept/reject problems via a Linear Programming (LP)
relaxation \citep{alaei2012online, alaei2013online, wang2018online, stein2020advance}.  Therefore, we focus on describing our results for the single-resource accept/reject problems, $k$-unit prophet inequalities and online knapsack, although we do provide a self-contained explanation of the reduction in \Cref{sec:multiToSingle}.  We note that this reduction framework also extends to actions more general than assignment, e.g.\ joint assortment and pricing, so our results also apply to these problems in revenue management.  However, we do not formalize this connection in the present paper, instead deferring to the expansive literature \citep{gallego2015online,goyal2020asymptotically,ma2021dynamic,feng2022near,chen2023assortment}.

\subsection{Contribution of this Paper}


This paper characterizes the \textit{tight} guarantee relative to the LP relaxation for both single-resource problems: $k$-unit prophet inequalities and online knapsack.  By tight guarantee, we mean the best-possible ratio that an online algorithm can obtain (in terms of its expected reward, divided by the value of the LP relaxation), on a worst-case instance chosen by an adversary.  This tight ratio depends on whether the problem is $k$-unit prophet inequalities (and the specific value of $k$) or online knapsack, with an adversary choosing the number of queries $T$ and the valuation/size distribution of each query.

\textbf{$k$-unit prophet inequalities.}
We characterize the tight LP-relative guarantee for all positive integers $k$, improving the previously best-known lower bound of $1-1/\sqrt{k+3}$ from \citet{alaei2011bayesian} for $k>1$.  (When $k=1$, Alaei's bound equals 1/2 and is tight, which is the original "prophet inequality".)  There is no closed-form for our tight ratios when $k>1$, but we display some values in \Cref{RatioTable} and \Cref{ComparisonRatioFigure}.  We note that Alaei's lower bound was recently improved for small values of $k$ by \citet{chawla2020static}, which is the existing lower bound displayed in \Cref{RatioTable}.  The best-known existing upper bound relative to the LP is inherited from the correlation gap in the IID setting \citep[see][]{yan2011mechanism}, and our tight result improves both the upper and lower bounds.

The LP relaxation enables us to directly generalize our results to a multi-resource setting where there are $m$ resources and each resource $j$ can serve up to $k_j$ queries. To elaborate, denote by $\gamma^*_{k}$ the tight LP-relative guarantee for the single-resource problem. Our results then imply a tight LP-relative guarantee of $\gamma^*_{k_{\min}}$ for multi-resource online matching, where $k_{\min}=\min_{j=1,\dots,m}k_j$.

\begin{table}
\centering
\caption{
Our tight ratios compared to existing lower/upper bounds, displayed up to $k=8$.
}\label{RatioTable}
\begin{tabular}{|c|cccccccc|}
\hline
value of  $k$& 1&2&3&4&5&6&7&8\\
\hline
Existing lower bound &0.5000 & 0.5859 & 0.6309 & 0.6605 & 0.6821 & 0.6989 & 0.7125 & 0.7240   \\
\hline
Our tight ratios &0.5000 & 0.6148 & 0.6741 & 0.7120  & 0.7389 & 0.7593 & 0.7754 & 0.7887  \\
\hline
Existing upper bound& 0.5000 & 0.7293 & 0.7760  & 0.8046 & 0.8245 & 0.8394 & 0.8510  & 0.8604 \\
\hline
\end{tabular}
\end{table}

\begin{figure}
\centering
\includegraphics[width=0.6\textwidth]{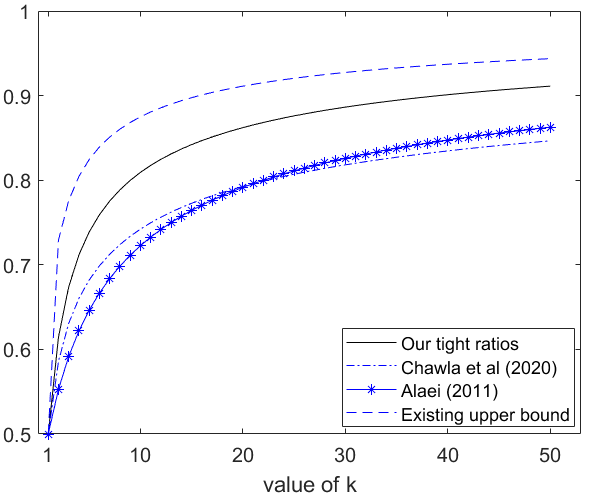}
\caption{
Our tight ratios compared to existing lower/upper bounds, plotted up to $k=50$.
}\label{ComparisonRatioFigure}
\end{figure}

Alaei's result was derived through a "Magician's problem".  We instead analyze $k$-unit prophet inequality through the lens of \textit{Online Contention Resolution Schemes (OCRS)}, which is equivalent to the Magician's problem, but allows for an LP formulation whose optimal solution is nicely structured.  Given this, we then solve an optimization from the adversary's perspective to implicitly characterize the tight guarantee.  We elaborate further on these techniques in \Cref{sec:newTechProph}.

\textbf{Online knapsack.}
We show the tight LP-relative guarantee to be $1/(3+e^{-2})\approx0.319$, improving the previously best-known guarantee of 0.2 from \citet{dutting2020prophet}.  In the \textit{unit-density} special case, where the realized size always equals the realized valuation, we establish an improved guarantee of 0.355, better than the approximation ratio of 0.321 from \citet{stein2020advance}.

Like in the $k$-unit prophet inequality case, both of these results go through the LP relaxation and hence extend to the online assignment problem with multiple resources.

Previous results \citep{dutting2020prophet,stein2020advance,feldman2021online} analyzed knapsack algorithms that accepted either only "large-sized" items or only "small-sized" items.  We instead consider a "best-fit" algorithm that can pack large-sized items alongside small-sized items, which we are able to analyze by establishing an \textit{invariant} on its distribution of capacity consumption at any point in time.  Our invariant technique is flexible, and we show how it can be modified under the unit-density assumption to yield an improved guarantee in this special case.  We elaborate further in \Cref{sec:newTechKnap}.  We note that even in the easier setting where the order of queries is uniformly random (instead of fixed by an adversary), no guarantee better than $1/(3+e^{-2})$ is known.

\textbf{Significance of LP-relative tightness.}
Our ratios are best-possible if the denominator is the relaxed LP value;
it is however plausible that better guarantees are possible if the denominator is directly the prophet's expected reward (not by much, as we show in \Cref{Upperprophetproposition} for $k=2$).
Nonetheless, such guarantees are currently unknown, so our guarantees are state-of-the-art even if the denominator is the prophet's expected reward.
We remark that the multi-resource to single-resource reductions do require comparing against the LP, so our guarantees indicate the limits of approximation ratios obtained through the LP relaxation.
Finally, even for a single resource, tight guarantees relative to the LP are a fundamental quantity of interest, with connections to correlation gaps and online contention resolution (see \citet{chekuri2014submodular} and \citet{feldman2021online} for further details).
%

\subsection{Roadmap}
\Cref{sec:PFormulation} formalizes the problems described at the beginning of the Introduction.  \Cref{sec:multiToSingle} reduces these problems to single-resource OCRS problems.  \Cref{sec:newTechProph,sec:newTechKnap} explain our new techniques for the $k$-unit and knapsack OCRS problems, respectively.  \Cref{sec:relatedWork} discusses further related work.  \Cref{1overkspecialcase} provides detailed results for $k$-unit OCRS, while \Cref{randomsizesection} provides detailed results for knapsack OCRS.  \Cref{sec:extensions} presents two extensions related to knapsack OCRS.  \Cref{sec:conc} concludes.

\section{Problem Formulations, New Techniques for OCRS}

We first formalize the most general problem studied in this paper, online assignment, which can capture as special cases the other three problems, $k$-unit prophet inequalities, online knapsack, and online matching.
We then explain our techniques on the Online Contention Resolution Scheme (OCRS) versions of the single-resource problems, $k$-unit OCRS, and knapsack OCRS.

\subsection{Problem Formulation for Online Assignment}\label{sec:PFormulation}

There are $m$ resources and the initial capacity of each resource is scaled to $1$. At each period $t$, query $t$ arrives and is associated with a non-negative reward $\mathbf{\tilde{r}}_t=(\tilde{r}_{t1},\dots,\tilde{r}_{tm})$ and a size $\mathbf{\tilde{d}}_t=(\tilde{d}_{t1},\dots,\tilde{d}_{tm})\in(0,1]^m$, where vector $(\mathbf{\tilde{r}}_t, \mathbf{\tilde{d}}_t)$ is assumed to be stochastic and drawn from a known distribution $F_t(\cdot)$. After the value of $(\mathbf{\tilde{r}}_t, \mathbf{\tilde{d}}_t)$ is revealed, the decision maker has to decide to serve or reject this query irrevocably. If served, the decision maker also needs to assign a resource $j\in\{1,\ldots,m\}$ to serve query $t$. Then, query $t$ will take up $\tilde{d}_{tj}$ capacity of resource $j$ and a reward $\tilde{r}_{tj}$ will be collected. The goal is to maximize the total collected reward without violating the capacity constraint of any resource.

Any online policy $\pi$ for the decision maker can be specified by a set of decision variables $\{x^{\pi}_{tj}\}_{j=1,\dots,m, t=1,\dots,T}$, where $x^{\pi}_{tj}$ is a binary variable and denotes whether query $t$ is served by resource $j$. Note that $x^\pi_{tj}$ is a binary random variable that can also depend on randomness in the policy. 
$\pi$ needs to satisfy the following capacity constraint:
\begin{equation}\label{capacityconstraint}
\sum_{t=1}^{T}\tilde{d}_{tj}\cdot x^{\pi}_{tj}\leq 1,~~~\forall j=1,\dots,m.
\end{equation}
and the constraint $\sum_{j=1}^{m}x^\pi_{tj}\leq1$ for all $t=1,\dots,T$.
The total reward collected by policy $\pi$ is denoted by $V^{\pi}(\bm{I})=\sum_{t=1}^{T}\sum_{j=1}^{m}\tilde{r}_{tj}\cdot x^{\pi}_{tj}$, where $\bm{I}=((\mathbf{\tilde{r}}_1, \mathbf{\tilde{d}}_1), (\mathbf{\tilde{r}}_2, \mathbf{\tilde{d}}_2), \dots, (\mathbf{\tilde{r}}_T, \mathbf{\tilde{d}}_T))$ denotes the realized rewards/sizes.  We let $\mathcal{H}=(F_1,F_2,\dots,F_T)$ denote the problem instance and it is understood that $\bm{I}\sim\mathcal{H}$ means $(\mathbf{\tilde{r}}_t, \mathbf{\tilde{d}}_t)$ is drawn independently from $F_t$ for all $t$, so that $\mathbb{E}_{\pi, \bm{I}\sim\mathcal{H}}[V^{\pi}(\bm{I})]$ denotes the expected total reward collected by policy $\pi$.  

The expected reward collected by the online algorithm is compared to that of the prophet, who can make decisions based on the knowledge of the realizations of all the queries. The prophet's expected reward is denoted by $\mathbb{E}_{\bm{I}\sim\mathcal{H}}[V^{\text{off}}(\bm{I})]$. For any feasible online policy $\pi$, its guarantee $\gamma$ is defined as
\begin{equation}\label{competitiveratio}
\gamma=\inf_{\mathcal{H}}\frac{\mathbb{E}_{\pi, \bm{I}\sim\bm{F}}[V^\pi(\bm{I})]}{\mathbb{E}_{\bm{I}\sim\bm{F}}[V^{\text{off}}(\bm{I})]}
\end{equation}

Typically, the best-known guarantees come from comparing the online algorithm to a \textit{Linear Programming (LP) relaxation} of the prophet, which only has to satisfy the capacity constraint in expectation. The LP relaxation can be formulated as follows.
\begin{subequations}\label{LPuppermultidimension}
\begin{align}
\text{UP}(\mathcal{H})=&~\max~~\sum_{t=1}^{T}\mathbb{E}_{(\tilde{\mathbf{r}}_t, \tilde{\mathbf{d}}_t)\sim F_t}[\sum_{j=1}^m \tilde{r}_{tj}\cdot x_{tj}(\tilde{\mathbf{r}}_t, \tilde{\mathbf{d}}_t)]\\
&~~\mbox{s.t.}~~~~\sum_{t=1}^{T}\mathbb{E}_{(\tilde{\mathbf{r}}_t, \tilde{\mathbf{d}}_t)\sim F_t}[\tilde{d}_{tj}\cdot x_{tj}(\tilde{\mathbf{r}}_t, \tilde{\mathbf{d}}_t)]\leq1,~~~\forall j=1,\dots,m  \label{con:LPcapacity} \\
&~~~~~~~~~~~\sum_{j=1}^{m}x_{tj}(\mathbf{r}_t,\mathbf{d}_t)\leq 1,~~~\forall t, \forall (\mathbf{r}_t, \mathbf{d}_t) \label{LP:probvec} \\
&~~~~~~~~~~~ x_{tj}(\mathbf{r}_t, \mathbf{d}_t)\geq0,~~~\forall t, \forall (\mathbf{r}_t, \mathbf{d}_t), \forall j=1,\dots,m. \label{LP:nonneg}
\end{align}
\end{subequations}
Here $\mathcal{H}=(F_1,\dots, F_T)$ denotes the problem instance, while $x_{tj}(\mathbf{r}_t, \mathbf{d}_t)$ denotes the probability of assigning resource $j$ to serve query $t$ conditional on its reward and size vectors realizing to $\mathbf{r}_t=(r_{t1},\dots,r_{tm})$ and $\mathbf{d}_t=(d_{t1},\dots,d_{tm})$, respectively.

\textbf{Special cases.} \textit{Online matching} is captured by having size vector $\mathbf{\tilde{d}}_t$ deterministically equal $(1/k_1,\ldots,1/k_m)$ for all queries $t$.  That is, a query when served by resource $j$ always consumes $1/k_j$ of its initial capacity, which can be interpreted as consuming one of $k_j$ initial "units".  Orthogonally, the special case of \textit{online knapsack} is captured when $m=1$.  The intersection of both special cases is the \textit{$k$-unit prophet inequalities} problem, where we have omitted subscript $j$ when denoting the starting number of units $k$ of the single resource.

\textbf{Note about distributions.} We always assume that the random rewards and sizes are input as discrete distributions, so $\text{UP}(\mathcal{H})$ is finite and polynomial-sized. We re-iterate that for convenience, we assume sizes are always positive. When sums are indexed by $d_t$, it is understood that this is summing over the finite, positive support of the distribution of sizes that query $t$ can take.

\subsection{Reduction to Single-resource OCRS Problems} \label{sec:multiToSingle}

The online assignment problem described in \Cref{sec:PFormulation} can be solved as follows.
Let $x^*_{tj}(\mathbf{r}_t, \mathbf{d}_t)$ denote an optimal solution to $\text{UP}(\mathcal{H})$.
For each query $t$, observe its realization $(\tilde{\mathbf{r}}_t, \tilde{\mathbf{d}}_t)$ and "route" it to at most one resource, such that the probability of routing to each resource $j$ is $x^*_{tj}(\tilde{\mathbf{r}}_t, \tilde{\mathbf{d}}_t)$, which satisfies $\sum_{j=1}^m x^*_{tj}(\tilde{\mathbf{r}}_t, \tilde{\mathbf{d}}_t)\le 1$ by LP constraints~\eqref{LP:probvec}--\eqref{LP:nonneg}.  An OCRS (specified later) for that resource $j$ will then determine whether to accept query $t$: if so, the algorithm serves query $t$ using resource $j$; otherwise, the algorithm does not serve query $t$ at all.

We remark that it is possible for a query $t$ to not get routed to any resource $j$, which would ensure its rejection.
Intuitively, query $t$ is only routed to resources $j$ for which the realized reward $\tilde{r}_{tj}$ is high compared to the realized consumption $\tilde{d}_{tj}$; it may not get routed at all on realizations where the entries of $\tilde{\mathbf{r}}_t$ are low (even when the number of resources is $m=1$).
Also, once routed, the exact value of $\tilde{r}_{tj}$ is ignored, with the presumption that it is "high enough" relative to $\tilde{d}_{tj}$.
We finally remark that the algorithm does not check resource state when randomly routing---a query $t$ can get routed to a resource $j$ with insufficient (less than $\tilde{d}_{tj}$) remaining capacity.  In this case, the query does not get re-routed to another resource, and again its rejection is ensured.

We now explain how an OCRS works and the guarantee it provides, first for a single \textit{$k$-unit} resource.  We should interpret a query as "active" if it was routed to the resource.
\begin{itemize}
\item \textbf{Input}: for each query $t=1,\ldots,T$, the probability $p_t$ with which it is independently active.
\item \textbf{Output}: for each query $t$, the (randomized) decision of whether to accept it when it is active, depending on how many of the $k$ units of capacity have already been consumed.
\item \textbf{Guarantee}: as long as $\sum_{t=1}^T p_t\le k$, every query $t$ will be accepted w.p.\ at least $\gamma^*_k$ conditional on it being active, where $\gamma^*_k$ is the tight guarantee for $k$-unit OCRS (to be specified later).
\end{itemize}
This $k$-unit OCRS is used for the online matching and $k$-unit prophet inequalities problems.  To elaborate, fix a resource $j$.  We set $p_t=\mathbb{E}_{(\tilde{\mathbf{r}}_t, \tilde{\mathbf{d}}_t)\sim F_t}[x^*_{tj}(\tilde{\mathbf{r}}_t, \tilde{\mathbf{d}}_t)]$, the probability that each query $t=1,\ldots,T$ is active (routed to resource $j$).  We obtain $\sum_{t=1}^T p_t \le k_j$ through LP constraints~\eqref{con:LPcapacity} after noting that $\tilde{d}_{tj}=1/k_j$ w.p.~1.
Therefore, the OCRS guarantees to accept every query $t$ with probability at least $\gamma^*_{k_j}$ whenever $t$ is active.  The expected reward collected from resource $j$ is $\sum_{t=1}^T \gamma^*_{k_j}\mathbb{E}_{(\tilde{\mathbf{r}}_t, \tilde{\mathbf{d}}_t)\sim F_t}[\tilde{r}_{tj} x^*_{tj}(\tilde{\mathbf{r}}_t, \tilde{\mathbf{d}}_t)]$, where we note that conditional on query $t$ being active for resource $j$, whether it is actually served (occuring w.p.~$\gamma^*_{k_j}$) is independent of $\tilde{r}_{tj}$ (because the OCRS's decisions do not depend on the exact values of $\tilde{r}_{tj}$).  Summing over resources and noting that $\min_j\gamma^*_{k_j}=\gamma^*_{\min_j k_j}=\gamma^*_{k_{\min}}$ ($\gamma^*_k$ is increasing in $k$), the algorithm's expected reward is at least
\begin{align*}
\sum_{j=1}^m \gamma^*_{k_j}\cdot\sum_{t=1}^T \mathbb{E}_{(\tilde{\mathbf{r}}_t, \tilde{\mathbf{d}}_t)\sim F_t}[\tilde{r}_{tj} x^*_{tj}(\tilde{\mathbf{r}}_t, \tilde{\mathbf{d}}_t)]
\ge(\min_{j=1,\ldots,m}\gamma^*_{k_j})\cdot\sum_{j=1}^m\sum_{t=1}^T\mathbb{E}_{(\tilde{\mathbf{r}}_t, \tilde{\mathbf{d}}_t)\sim F_t}[\tilde{r}_{tj} x^*_{tj}(\tilde{\mathbf{r}}_t, \tilde{\mathbf{d}}_t)]
=\gamma^*_{k_{\min}}\cdot\text{UP}(\mathcal{H}).
\end{align*}

We now explain OCRS for a single \textit{knapsack} resource.  Here, a query $t$, when active, also takes one of various non-zero sizes $d_t$.
\begin{itemize}
\item \textbf{Input}: for each query $t=1,\ldots,T$, the independent distribution of sizes $d_t$ it can take, given by probabilities $p_t(d_t)$ satisfying $\sum_{d_t} p_t(d_t)\le 1$; the query is inactive w.p.~$1-\sum_{d_t}p_t(d_t)$.
\item \textbf{Output}: for each query $t$, the (randomized) decision of whether to accept it under any realized size $d_t$, depending on how much of the resource's capacity has already been consumed.
\item \textbf{Guarantee}: as long as the expected sum of sizes does not exceed the initial capacity 1, every query $t$ will be accepted w.p.\ at least $\gamma=1/(3+e^{-2})$, conditional on any size $d_t$ taken.
\end{itemize}
This knapsack OCRS is used for the online assignment and online knapsack problems.  To elaborate, again fix a resource $j$.  We set $p_t(d_t)=\mathbb{E}_{(\tilde{\mathbf{r}}_t, \tilde{\mathbf{d}}_t)\sim F_t}[\bI(\tilde{d}_{tj}=d_t)x^*_{tj}(\tilde{\mathbf{r}}_t, \tilde{\mathbf{d}}_t)]$ as the probability that each query $t$ takes each size $d_t$ (for resource $j$).  It is easy to see that $\sum_{d_t} p_t(d_t)\le 1$ (because $x^*_{tj}(\tilde{\mathbf{r}}_t, \tilde{\mathbf{d}}_t)\le1$), and the expected sum of sizes satisfies
\begin{align*}
\sum_t \sum_{d_t} d_t \cdot p_t(d_t)
&=\sum_t\sum_{d_t} \mathbb{E}_{(\tilde{\mathbf{r}}_t, \tilde{\mathbf{d}}_t)\sim F_t}[d_t\cdot\bI(\tilde{d}_{tj}=d_t)\cdot x^*_{tj}(\tilde{\mathbf{r}}_t, \tilde{\mathbf{d}}_t)]
\\ &=\sum_t\mathbb{E}_{(\tilde{\mathbf{r}}_t, \tilde{\mathbf{d}}_t)\sim F_t}[\tilde{d}_{tj}\cdot x^*_{tj}(\tilde{\mathbf{r}}_t, \tilde{\mathbf{d}}_t)]
\\ &\le1
\end{align*}
by LP constraints~\eqref{con:LPcapacity}.  Therefore, the OCRS guarantees to accept every query $t$ with probability at least $\gamma=1/(3+e^{-2})$ conditional on any non-zero size $d_t$ taken.
The expected reward collected from resource $j$ is $\sum_{t=1}^T \gamma\cdot \mathbb{E}_{(\tilde{\mathbf{r}}_t, \tilde{\mathbf{d}}_t)\sim F_t}[\tilde{r}_{tj}\cdot x^*_{tj}(\tilde{\mathbf{r}}_t, \tilde{\mathbf{d}}_t)]$, again noting that conditional on query $t$ being routed to resource $j$, it is always served w.p.~$1/(3+e^{-2})$ independent of the value of $\tilde{r}_{tj}$.  Summing over resources, the algorithm's expected reward is at least
\begin{align*}
\sum_{j=1}^m \gamma\sum_{t=1}^T \mathbb{E}_{(\tilde{\mathbf{r}}_t, \tilde{\mathbf{d}}_t)\sim F_t}[\tilde{r}_{tj}\cdot x^*_{tj}(\tilde{\mathbf{r}}_t, \tilde{\mathbf{d}}_t)]
=\frac{1}{3+e^{-2}}\cdot\text{UP}(\mathcal{H}).
\end{align*}

Our $k$-unit OCRS is described in \Cref{constructprophet} and the remarks afterward.  Our knapsack OCRS is described in \Cref{definedistributionrandomsize} and the remarks afterward.  Formal specifications of the multi-resource algorithms that use these OCRS's as subroutines are deferred to \Cref{sec:FinalAlgorithm}.

In proving optimality, we show that even for a single resource, the guarantees in the $k$-unit and knapsack OCRS problems cannot exceed $\gamma^*_k$ and $1/(3+e^{-2})$ respectively.
Now, it may seem like the OCRS problem is unnecessarily stringent---it requires a \textit{query-wise} acceptance guarantee, instead of only a guarantee on the algorithm's \textit{total} reward compared to $\text{UP}(\mathcal{H})$.
However, \citet{lee2018optimal} use a simple LP duality argument to show that under adversarially-chosen reward values, guarantees for the original reward collection problem are no better than guarantees for the corresponding OCRS problem.
Therefore, \textbf{from this point on in the paper}, we focus solely on the OCRS problems, having established that they suffice for providing guarantees on the single- or multi- resource reward collection problems and that the OCRS guarantees are best-possible when comparing to $\text{UP}(\mathcal{H})$.

\subsection{New Techniques for $k$-unit OCRS} \label{sec:newTechProph}
In \Cref{sec:multiToSingle} we explained why the single-resource $k$-unit OCRS problem is useful as a subroutine for the online matching and $k$-unit prophet inequalities problems.  We now formalize the $k$-unit OCRS problem in \Cref{def:kUnitOnlineRounding}, and explain our new techniques for solving it optimally.

\begin{defn}[$k$-unit OCRS Problem] \label{def:kUnitOnlineRounding}
There is a sequence of queries $t=1,\ldots, T$, each of which is active independently according to a known probability $p_t$.
Whether a query is active is sequentially observed, and active queries can be immediately served or rejected, while inactive queries must be rejected.
At most $k$ queries can be served in total, and it is promised that $\sum_tp_t\le k$.
The goal of an online algorithm is to serve every query $t$ with probability at least $\gamma$ conditional on it being active, for a constant $\gamma\in[0,1]$ as large as possible, potentially with the aid of randomization.
\end{defn}

It is easy to see\footnote{
For example, suppose that $k=1$, $T=2$, and $p_1=p_2=1/2$.  If we attempt to set $\gamma=1$, then the first query would be served ex-ante w.p.~1/2, i.e.,  whenever it is active.  This means that half the time no capacity would remain for query 2, i.e., half the time query 2 is active it does not get served.  For this example, the optimal value of $\gamma$ can be calculated to be 2/3.
} that despite $\price$ being fractionally feasible, a guarantee of $\gamma=1$ in \Cref{def:kUnitOnlineRounding} is generally impossible.
The work of \citet{alaei2011bayesian} implies a solution to \Cref{def:kUnitOnlineRounding} with $\gamma=1-\frac{1}{\sqrt{k+3}}$.
Presented in the slightly different context of a ``$\gamma$-Conservative Magician,''
Alaei's procedure has the further appealing property that it does not need to know vector $\price$ in advance, as long as each $p_t$ is revealed when query $t$ is observed, and it is promised that $\sum_tp_t\le k$.
However, it has remained unknown whether Alaei's $\price$-agnostic procedure or its analyzed bound of $\gamma=1-\frac{1}{\sqrt{k+3}}$ is tight for an arbitrary positive integer $k$.
In this paper, we resolve this question, in the following steps.
\begin{enumerate}
\item Under the assumption that $\price$ is known, we formulate the optimal $k$-unit OCRS problem using a new LP.
This LP tracks the probability distribution of the capacity utilization, which must lie in $\{0,\frac{1}{k},\ldots,1\}$, over time $t=1,\ldots, T$.
The decision variables correspond to subdividing and selecting sample paths
at each time $t$, with total measure exactly $\gamma$, on which the algorithm will serve query $t$ whenever it is active.
This selection is constrained to sample paths with at least $\frac{1}{k}$ capacity remaining, which is enforced in the LP through tracking the capacity utilization.
Finally, $\gamma$ is also a decision variable, with the objective being to maximize $\gamma$.
\item For an arbitrary $\price$, we characterize an optimal solution to this LP based on the structure of its dual.  The optimal selection prioritizes sample paths with the \textit{least} capacity utilized, at every time $t$, \textit{irrespective} of the values $p_{t+1},p_{t+2},\ldots$ in the future.
Such a solution corresponds to the $\gamma$-Conservative Magician from \citet{alaei2011bayesian}, except that $\gamma$, instead of being fixed to $1-\frac{1}{\sqrt{k+3}}$, is set to an optimal value that depends on the vector $\price$.
\item We derive a closed-form expression for this optimal value of $\gamma$ as a function of $\price$.  We show that $\gamma$ is minimized when $p_t=k/T$ for each $t$ and $T\to\infty$, corresponding to a Poisson distribution of rate $k$.
We characterize this infimum value of $\gamma$ using an ODE and provide an efficient procedure for computing it numerically.
\end{enumerate}

For any $k$, let $\sgk$ denote the infimum value of $\gamma$ described in Step~3 above.
The conclusion is that setting $\gamma=\sgk$ is a feasible solution to \Cref{def:kUnitOnlineRounding}, with $\sgk>1-\frac{1}{\sqrt{k+3}}$ for all $k>1$, achieved using the $\price$-agnostic procedure described above.
Moreover, the guarantee of $\sgk$ is the best possible, since even a procedure that knows $\price$ in advance cannot do better than a $\gamma$-Conservative Magician with an optimized value of $\gamma$, which in the Poisson worst case can be as low as $\sgk$.

\textbf{Comparison to \citet{alaei2012online}.}
$k$-unit prophet inequalities have been analyzed using LPs before in \citet{alaei2012online}, who formulate a primal LP encoding the adversary's problem of minimizing an online algorithm's optimal dynamic programming value. They then use an auxiliary ``Magician's problem,'' analyzed through a ``sand/barrier'' process, to construct a feasible dual solution with $\gamma=1-\frac{1}{\sqrt{k+3}}$. By contrast, we directly formulate the $k$-unit OCRS problem using an LP under the assumption that the vector $\price$ is known.
Our LP dual along with complementary slackness allows us to establish the structure of the optimal $k$-unit OCRS, showing that it indeed corresponds to a $\gamma$-Conservative Magician.  However, in our case $\gamma$ is set to a value dependent on $\price$, which we show is always at least $\sgk$, and strictly greater than $1-\frac{1}{\sqrt{k+3}}$ for all $k>1$.

\textbf{Comparison to \citet{wang2018online}.}
The values of $\sgk$ we derive have previously appeared in \citet{wang2018online} through the stochastic analysis of a ``reflecting'' Poisson process.  Our work differs by establishing \textit{optimality} for these values $\sgk$, as the solutions to a sequence of optimization problems from our framework.
Moreover, their paper assumes Poisson arrivals to begin with, while we allow arbitrary probability vectors $\price$ and show the limiting Poisson case to be the worst case.

\textbf{The classical prophet inequality comparison.}
We should note that classically in the $k$-unit prophet inequality problem, the goal is to compute the worst-case performance of an online algorithm, which sequentially observes independent draws from known distributions and can keep $k$ of them, and compare instead to a \textit{prophet}, whose performance is the expected sum of the $k$ highest realizations. The prophet’s performance is upper-bounded by the LP relaxation, so our guarantees that are tight relative to the LP also imply the \textit{best-known prophet inequalities to date} for all $k>1$. We do give an example that demonstrates this guarantee to be “almost” tight even
when compared to the weaker prophet benchmark.
Through our LP's and complementary slackness, we can convert the Poisson worst case for the $k$-unit OCRS problem into an explicit instance of $k$-unit prophet inequalities, on which the reward of any online algorithm relative to the LP relaxation is upper-bounded by $\sgk$.
Moreover, by modifying such an instance, we also provide a new upper bound of $0.6269$ relative to the prophet,  when $k=2$ (\Cref{Upperprophetproposition}). Since $\gamma^*_2\approx0.6148$, this shows that not much improvement beyond $\sgk$ is possible relative to the prophet when $k=2$.

\subsection{New Techniques for the Knapsack Setting} \label{sec:newTechKnap}
In \Cref{sec:multiToSingle} we explained why the single-resource knapsack OCRS problem is useful as a subroutine for the online assignment and online knapsack problems.  We now formalize the knapsack OCRS problem in \Cref{def:knapsackOnlineRounding}, and explain our new techniques for solving it.

\begin{defn}[Knapsack OCRS Problem] \label{def:knapsackOnlineRounding}
There is a sequence of queries $t=1,\ldots, T$, and each query $t$ independently realizes a size, which equals $d_t\in(0,1]$ with a known probability $p_t(d_t)$ satisfying $\sum_{d_t} p_t(d_t)\leq 1$. With probability $1-\sum_{d_t} p_t(d_t)$, the query is "inactive" with size 0 and can be ignored. After the query's size is observed, the query must be immediately served or rejected.
The total size of queries served cannot exceed 1, and it is promised that $\sum_t\sum_{d_t} p_t(d_t)\cdot d_t\le 1$.
The goal of an online algorithm is to serve every query $t$ with probability at least $\gamma$ conditional on the size being realized to $d_t$, for each $d_t\in(0,1]$, and for a constant $\gamma\in[0,1]$ as large as possible.
\end{defn}

Similar to our approach for the multi-unit setting, we design a solution for the knapsack OCRS by tracking the distribution of capacity utilization over time. For each size realization $d_t\in(0,1]$, we select for each query $t$ a $\gamma$-measure of sample paths on which it should be served whenever the size is realized as $d_t$, under the constraint that these paths have a current utilization of at most $1-d_t$.
However, different from the multi-unit setting, in the knapsack setting, we need to always maintain a $\gamma$-measure of sample paths on which utilization is 0, in case an item $T$ with size realization $d_T=1$ and $p_T(d_T)=\varepsilon$ arrives at the end. Accordingly, in stark contrast to the $\gamma$-Conservative Magician, our knapsack procedure selects for each query and each size realization the sample paths with the \textit{most} capacity utilized, on which that query still fits.
We dub this procedure a ``Best-fit Magician.''\footnote{This is because it resembles the ``best-fit'' heuristic for bin packing \citep{garey1972worst}.}
In the more general knapsack setting, capacity utilization can only be tracked in polynomial time after discretizing size realizations by $1/K$ for some large integer $K$; nonetheless, we will show (in \Cref{sec:knapsackPolytimeImpl}) that this loses a negligible additive term of $O(1/K)$ in the guarantee.

To derive the maximum feasible guarantee $\gamma$ for a Best-fit Magician, we note that the expected capacity utilization over the sample paths is $\gamma\cdot\sum_t\sum_{d_t} p_t(d_t)\cdot d_t$, which is always upper-bounded by $\gamma$, since $\sum_t\sum_{d_t} p_t(d_t)\cdot d_t\le1$.
Therefore, to lower-bound the measure of sample paths with 0 utilization, it suffices to upper-bound the measure of sample paths whose utilization is small but non-zero.
To do so, we use the rule of the Best-fit Magician, namely, that an arriving query with a size realization $d_t$ will only be served on a previously empty sample path if there is less than a $\gamma$-measure of sample paths with utilization in $(0,1-d_t]$.
Based on this fact, we derive an \textit{invariant} that holds after each query $t$ and upper-bounds the measure of sample paths with utilization in $(0,b]$ by a decreasing exponential function of the measure with utilization in $(b,1-b]$, for any small size $b\in(0,1/2]$.
This allows us to show that a $\gamma$ as large as $\frac{1}{3+e^{-2}}\approx0.319$ allows for a $\gamma$-measure of sample paths to have 0 utilization at all times, and hence is feasible.
The Best-fit Magician is also agnostic to knowing the probabilities $\{p_t(d_t)\}_{\forall t, \forall d_t}$ in advance, as long as it is promised that $\sum_t\sum_{d_t} p_t(d_t)\cdot d_t\le1$.
Nonetheless, we construct a counterexample showing it to be \textit{optimal}, in that $\gamma=\frac{1}{3+e^{-2}}$ is an upper bound on the guarantee for the knapsack OCRS problem even if the probabilities $\{p_t(d_t)\}_{\forall t, \forall d_t}$ are known in advance.

To our knowledge, our analysis differs from existing ones for knapsack in an online setting \citep{dutting2020prophet,stein2020advance,feldman2021online} by eschewing the need to split queries into ``large'' vs. ``small'' based on their size (usually, whether their size is greater than 1/2).
In fact, we show that any algorithm that \textit{packs large and small queries separately} is limited to $\gamma\le0.25$ in our problem (\Cref{Largesmallproposition}), whereas our tight guarantee is $\gamma=\frac{1}{3+e^{-2}}\approx0.319$.

Our result can be further improved in the case of \textit{unit-density} online knapsack, where the random size and reward of a query are always identical.
Indeed, since it is no longer possible for a small query to have a high reward, we no longer need to guarantee a uniform lower bound $\gamma$ on the probability of serving any query with any size realization.
Instead, we show that our invariant still holds for a decreasing sequence of service probabilities $\gamma_1\ge\cdots\ge\gamma_T$, and devise a particular sequence that guarantees an expected reward that is at least 0.3557 times the optimal LP value. 
This then implies a 0.3557 approximation for the multi-resource appointment scheduling problem of \citet{stein2020advance}, improving upon their 0.321 approximation.

\textbf{Comparison to \cite{alaei2013online}.} Another related setting is the online stochastic generalized assignment problem of \citet{alaei2013online}, for which the authors establish a guarantee of $1-\frac{1}{\sqrt{k}}$ when each query can realize a random size that is at most $1/k$.
They eliminate the possibility of ``large'' queries by imposing $k$ to be at least 2, showing that a constant-factor guarantee is impossible when $k=1$.
Although our problem can be generalized to random sizes, we need to assume that size is observed \textit{before} the algorithm makes a decision, whereas in their problem size is randomly realized \textit{after} the algorithm decides to serve a query.
This distinction allows our problem to have a constant-factor guarantee that holds even when queries can have size 1.  Moreover, our procedure starkly contrasts with theirs in that we prioritize selecting sample paths with the \textit{most} capacity utilized on which a query fits, while they prioritize sample paths with the \textit{least} capacity utilized.

\subsection{Further Related Work}\label{sec:relatedWork}

\textbf{Online knapsack.}
We should point out though that in the unit-density setting with a \textit{single} knapsack, a guarantee of 1/2, better than our guarantee of 0.3557, is possible under any fixed sequence of adversarial arrivals \citep{han2015randomized}.
However, such a guarantee fails\footnote{An additional factor of 1/2 would be lost, resulting in a guarantee of only 1/4; see \citet{ma2019competitive}.
In fact, a guarantee of 1/2 relative to the LP is impossible, due to the upper bound of 0.432 presented in our \Cref{Upperunitdensity}.} to extend to multiple knapsacks, whereas our guarantee of 0.3557, which holds relative to the LP, directly extends there, following {the same reduction argument as in \citet{stein2020advance}.}

\textbf{Prophet inequalities.}
{Prophet inequalities were originally posed in the statistics literature by \citet{krengel1978semiamarts}.
Due to their implications for posted pricing and mechanism design, prophet inequalities have been a surging topic in algorithmic game theory since the seminal works of \citet{hajiaghayi2007automated,chawla2010multi,yan2011mechanism,alaei2011bayesian}.
Of particular interest in these works are bounds for $k$-item prophet inequalities, and in this paper, we improve such bounds for all $k>1$ and show that our bounds are tight relative to the LP relaxation, under an adversarial arrival order.
More recently, prophet inequalities have also been studied under random order \citep{esfandiari2017prophet, correa2021prophet, arnosti2021tight}, free order \citep{correa2021prophet, beyhaghi2021improved}, or IID arrivals \citep{hill1982comparisons, correa2017posted, jiang2022tightness}, with $k$-unit prophet inequalities, in particular, being studied by \citet{arnosti2021tight} under random order, \citet{beyhaghi2021improved} under free order, and \citet{jiang2022tightness} under IID arrivals.
Prophet inequalities have also been studied under the batched setting \citep{alaei2022descending} with applications to descending-price auctions and have also been used as algorithmic subroutine for other revenue management problems (e.g. \citet{cominetti2010optimal, alaei2021revenue}). 
A survey of recent results in prophet inequalities can be found in \citet{correa2019recent}. 
}

\textbf{OCRS.} A guarantee of $\gamma$ for our problem in \Cref{def:kUnitOnlineRounding} (resp.\ \Cref{def:knapsackOnlineRounding}) is identical to a \textit{$\gamma$-selectable OCRS} for the $k$-uniform matroid (resp.\ knapsack polytope) as introduced in \citet{feldman2021online}.
However, we should clarify some assumptions about what is known beforehand and the choice of arrival order.
Our OCRS holds against an \textit{online adversary}, who can adaptively choose the next query to arrive but does not know the realizations of queries yet to arrive.  We show that the guarantee does not improve against the \textit{weakest adversary}, who has to reveal the arrival order in advance.  However, our OCRS do not satisfy the \textit{greedy} property and consequently do not hold against the \textit{almighty adversary}, who knows the realizations of all queries before having to choose the order.
We note that our $k$-unit OCRS does satisfy monotonicity \citep[see][]{chekuri2014submodular} but our knapsack OCRS does not.

In \citet{feldman2021online}, the authors derive a 1/4-selectable\footnote{
\citet{lee2018optimal} have improved this to a 1/2-selectable OCRS for general matroids, against the weakest adversary.  Our guarantees of $\sgk$ are all greater than 1/2 and hold against the online adversary but in the special case of $k$-uniform matroids.
} greedy OCRS for general matroids, and a $0.085$-selectable greedy OCRS for the knapsack polytope, both of which hold against an almighty adversary.
We establish significantly improved selectabilities against the weaker online adversary, and, importantly, show that our guarantees are \textit{tight} for our setting.

\textbf{Magician's problem.}
Our algorithms do enjoy a property not featured in the OCRS setting though: they \textit{need not know the universe} of elements in advance, holding even if the adversary can adaptively ``create'' the $p_t$ (and $d_t$) of the next query $t$, under the promise that $\sum_tp_td_t\le1$.
This property is inherited from the \textit{Magician's problem}, introduced by
\citet{alaei2011bayesian} as a powerful black box for approximately solving combinatorial auctions.
Our work fully resolves\footnote{
The main difference in the Magician's problem is that a query must be selected \textit{before} it is known whether it is active, and, if so, is irrevocably served.  The goal is to select each query with an ex-ante probability at least $\gamma$.  Our problem can be reinterpreted as selecting a $\gamma$-measure of sample paths on which each query should be served whenever it is active, which is completely equivalent.  Therefore, all of our results also hold for Alaei's Magician problem and its applications.
} his $k$-unit Magician problem, showing his $\gamma$-Conservative Magician to be optimal, and, importantly, showing how to find the optimal value $\gamma=\sgk$, which is greater than the value of $\gamma=1-\frac{1}{\sqrt{k+3}}$, for all $k>1$.
This improves all of the guarantees for combinatorial auctions, summarized in \citet{alaei2014bayesian}, that depend on this value of $\gamma$.

\section{$k$-unit Prophet Inequalities}\label{1overkspecialcase}
For each $k$, we derive the
tight guarantee for the $k$-unit prophet inequality problem with respect to the LP upper bound, or equivalently the
optimal solution $\sgk$ to the $k$-unit OCRS problem.  Note that our values $\sgk$ strictly exceed $1-\frac{1}{\sqrt{k+3}}$ for all $k>1$, and hence we also improve the best-known prophet inequalities for all $k>1$.
The structure of our proof follows the three steps outlined in \Cref{sec:newTechProph}. In a preliminary version \citep{jiang2022tight} of this work, we illustrate our approach for a special case $k=2$.

\subsection{LP Formulation of $k$-unit OCRS Problem}\label{sec:LPkOCRS}

We first present a new LP formulation of the $k$-unit OCRS problem, with the vector $\bm{p}$ satisfying $\sum_{t=1}^Tp_t\leq k$.
We name our LP as $\LCkunit(\bm{p})$.
\begin{align}
\LCkunit(\bm{p})=~  &\max ~~~  \theta \label{dualdual} \\
  &~~\text{s.t.}~~~  \theta\leq \frac{\sum_{l=1}^{k}x_{l,t}}{\p_t}~~~\forall t \tag{\ref*{dualdual}a} \label{ddconstraint5} \\
  &~~~~~~~~~ x_{1,t}\leq \p_t\cdot (1-\sum_{\tau<t}x_{1,\tau})~~~\forall t \tag{\ref*{dualdual}b}\label{ddconstraint6}\\
  &~~~~~~~~~ x_{l,t}\leq \p_t\cdot \sum_{\tau<t}(x_{l-1,\tau}-x_{l,\tau})~~~\forall t, \forall l=2,\dots,k \tag{\ref*{dualdual}c}\label{ddconstraint7}\\
  &~~~~~~~~~\theta, x_{1,t}\geq0, x_{2,t}\geq0, \dots, x_{k,t}\geq0 \nonumber.
\end{align}
Here, the variable $\theta$ can be interpreted as guarantee $\gamma$ in the $k$-unit OCRS problem and $x_{l,t}$ can be interpreted as the ex-ante probability of serving query $t$ as the $l$-th one. Then, constraint \eqref{ddconstraint5} guarantees that each query $t$ is served with an ex-ante probability $\theta\cdot\p_t$. Moreover, it is easy to see that the term $\sum_{\tau<t}x_{l-1,\tau}$ can be interpreted as the probability that the number of served queries has ``reached'' $l-1$ during the first $t-1$ periods, while the term $\sum_{\tau<t}x_{l,\tau}$ can be interpreted as the probability that the number of served queries is larger than $l-1$. Then, the term $\sum_{\tau<t}(x_{l-1,\tau}-x_{l,\tau})$ denotes the probability that the number of served queries is $l-1$ at the beginning of period $t$. Similarly, the term $1-\sum_{\tau<t}x_{1,\tau}$ denotes the probability that no query is served at the beginning of period $t$. Further note that each query $t$ can be served only after it becomes active, which happens independently with probability $\p_t$, and hence we get constraint \eqref{ddconstraint6} and \eqref{ddconstraint7}. 

The ``$\gamma$-Conservative Magician'' procedure of \citet{alaei2011bayesian} implies a feasible solution to $\LCkunit(\bm{p})$,
for any $\bm{\p}$ satisfying $\sum_{t=1}^{T}\p_t\leq k$, despite being presented in the different context of the Magician's problem. We now describe this implied solution
in \Cref{constructprophet}, which is based on a predetermined $\theta$. In general, our approach would continuously increase the value of $x_{l,t}$ from $0$ until one of the constraints \eqref{ddconstraint5}, \eqref{ddconstraint6} and \eqref{ddconstraint7} hold with equality, for each $l=1,\dots,k$ and each $t=1,\dots, T$. To be more specific, we define $t_1=0$ and then sequentially for each $l=1,\dots,k$, we remain $x_{l,t}=0$ for $t\leq t_l$ and increase the value of $x_{l,t}$ until the constraint \eqref{ddconstraint5} is binding sequentially for each $t>t_l$, until a time index $t_{l+1}$ such that constraint \eqref{ddconstraint7} is going to be violated. Then, sequentially for each $t>t_{l+1}$, we increase the value of $x_{l,t}$ such that constraint \eqref{ddconstraint7} holds with equality. The final algorithm for the multi-resource setting is presented in \Cref{sec:FinalAlgorithm}.

\begin{algorithm}
\caption{Pre-processed algorithm for the $k$-unit OCRS problem}
\begin{algorithmic}[1]
\State \textbf{Input}: a parameter $\theta$ and the probability sequence $\bm{p}$.
\State For a fixed $\theta\in[0,1]$, we define $x_{1,t}(\theta)=\theta\cdot\p_t$ from $t=1$ up to $t=t_2$, where $t_2$ is defined as the first time among $\{1,\dots,T\}$ such that $\theta>1-\sum_{t=1}^{t_2}\theta\cdot\p_t$ and if such a $t_2$ does not exist, we denote $t_2=T$. Then we define $x_{1,t}(\theta)=\p_t\cdot(1-\sum_{\tau=1}^{t-1}x_{1,\tau}(\theta))$ from $t=t_2+1$ up to $t=T$.
\For{$l=2,3,\dots, k-1$}
    \State Define $x_{l,t}(\theta)=0$ from $t=1$ up to $t=t_l$.
    \State Define $x_{l,t}(\theta)=\theta\cdot\p_t-\sum_{v=1}^{l-1}x_{v,t}(\theta)$ from $t=t_l+1$ up to $t=t_{l+1}$, where $t_{l+1}$ is defined as the first time among $\{1,\dots,T\}$ such that
    \[
    \theta\cdot\p_{t_{l+1}+1}-\sum_{v=1}^{l-1}x_{v,t_{l+1}+1}(\theta)>\p_{t_{l+1}+1}\cdot\sum_{t=1}^{t_{l+1}}(x_{l-1,t}(\theta)-x_{l,t}(\theta))
    \]
    and if such a $t_{l+1}$ does not exist, we denote $t_{l+1}=T$.
    \EndFor
\State Define $x_{k,t}(\theta)=0$ from $t=1$ up to $t=t_k$ and define $x_{k,t}(\theta)=\theta\cdot\p_t-\sum_{v=1}^{k-1}x_{v,t}(\theta)$ from $t=t_k+1$ up to $t=T$.
\State \textbf{Output}: the candidate solution $\{ x_{l,t}(\theta) \}$.
\end{algorithmic}    
\label{constructprophet}
\end{algorithm}

\textbf{Remarks about \Cref{constructprophet}.}
\begin{enumerate}
\item In \Cref{constructprophet} the policy is described as an LP solution (to later aid our proof of optimality).  The policy is actually implemented as follows: when each query $t=1,\ldots,T$ arrives, conditional on query $t$ being active and $l-1$ queries having already been served, serve query $t$ w.p.~$\frac{x_{1,t}(\theta)}{\p_t\cdot (1-\sum_{\tau<t}x_{1,\tau}(\theta))}$ if $l=1$, and w.p.~$\frac{x_{l,t}(\theta)}{\p_t\cdot \sum_{\tau<t}(x_{l-1,\tau}(\theta)-x_{l,\tau}(\theta))}$ if $l=2,\ldots,k$.  It will be preserved that $\Pr[\text{$l-1$ queries having already been served when query $t$ arrives}]$ equals $1-\sum_{\tau<t}x_{1,\tau}(\theta)$ if $l=1$ and $\sum_{\tau<t}(x_{l-1,\tau}(\theta)-x_{l,\tau}(\theta))$ if $l=2,\ldots,k$, and since query $t$ is active independently w.p.~$p_t$, it will become the $l$'th query served with probability exactly $x_{l,t}(\theta)$.
\item This will only lead a valid policy if parameter $\theta$ and the resulting values of $x_{l,t}(\theta)$ from \Cref{constructprophet} describe a feasible solution to $\LCkunit(\bm{p})$.  We will subsequently characterize the maximum feasible $\theta$, i.e.\ optimal $\theta^*$ for a given vector of active probabilities $\bm{\p}$.
\item Finally, we prove that setting $\theta=\gamma^*_k$ (computed in \Cref{sec:worstkunit}) is always feasible.  In this case, we note that the values $x_{1,t}(\gamma^*_k),\ldots,x_{k,t}(\gamma^*_k)$ for each query $t$ can actually be constructed on-the-fly, and the policy only needs to discover the value of each $p_t$ after making decisions for query $t$.  That is, our results hold even if an adaptive adversary chooses at each time $t$ the next query to arrive (see \citet{kleinberg2019matroid} for a precise definition of this "online" adversary).
\end{enumerate}

\subsection{Characterizing the Optimal LP Solution for a Given \textit{p}}\label{sec:optimalsolution}

In what follows, we identify $\theta^*$ for a fixed $\bm{\p}$, prove the optimality of $\{\theta^*, x_{l,t}(\theta^*)\}$, and describe the procedure of computing $\sgk$.
We begin by proving the condition on $\theta$ for $\{\theta, x_{l,t}(\theta)\}$ to be a feasible solution to $\LCkunit(\bm{p})$.

\begin{lemma}\label{feasiproposition}
For any vector $\bm{\p}$,
there exists a unique $\theta^*\in[0,1]$ such that $\sum_{\tau=1}^{T-1}x_{k,\tau}(\theta^*)=1-\theta^*$. Moreover, for any $\theta\in[0,\theta^*]$, $\{\theta, x_{l,t}(\theta)\}$ is a feasible solution to $\LCkunit(\bm{p})$.
\end{lemma}
The proof is relegated to \Cref{prooffeasiproposition}.
We now prove that $\{\theta^*, x_{l,t}(\theta^*)\}$ is an optimal solution to $\LCkunit(\bm{p})$. The dual of $\LCkunit(\bm{p})$ can be formulated as follows:
\begin{align}
 \DCkunit(\bm{\p})~=~ \min \ \ & \sum_{t=1}^{T}\p_t\cdot \beta_{1,t} \label{dual} \\
  \text{s.t.}\ & \beta_{l,t}+\sum_{\tau>t}\p_\tau\cdot (\beta_{l,\tau}-\beta_{l+1,\tau})-\xi_t\geq0,~~\forall t=1,\dots, T, ~\forall l=1,2,\dots,k-1\nonumber\\
  &\beta_{k,t}+\sum_{\tau>t}\p_\tau\cdot\beta_{k,\tau}-\xi_t\geq0,~~\forall t=1,\dots, T\nonumber \\
  & \sum_{t=1}^{T}\p_t\cdot \xi_t=1 \nonumber\\
  & \beta_{l,t}\geq0, \xi_t\geq0,~~\forall t=1,\dots,T, \forall l=1,\dots,k\nonumber.
\end{align}
To prove the optimality of $\{\theta^*, x_{l,t}(\theta^*)\}$, we will construct a feasible dual solution $\{\beta^*_{l,t}, \xi^*_t\}$ to $\DCkunit(\bm{p})$ such that complementary slackness conditions hold for the primal-dual pair $\{\theta^*, x_{l,t}(\theta^*)\}$ and $\{\beta^*_{l,t}, \xi^*_t\}$; then, the well-known primal-dual optimality criterion \citep{dantzig2006linear} establishes that $\{\theta^*, x_{l,t}(\theta^*)\}$ and $\{\beta^*_{l,t}, \xi^*_t\}$ are the optimal primal-dual pair to $\LCkunit(\bm{p})$ and $\DCkunit(\bm{p})$, which completes our proof. The above arguments are formalized in the following \Cref{Prophetoptimaltheorem}. The proof of \Cref{Prophetoptimaltheorem} is completed based on an induction argument, with details presented in \Cref{newProoftheorem3}. In \Cref{Prooftheorem3}, we also give an alternative constructive proof of \Cref{Prophetoptimaltheorem}, with the formulation of $\{\beta^*_{l,t}, \xi^*_t\}$ given explicitly. 
\begin{theorem}\label{Prophetoptimaltheorem}
The solution $\{\theta^*, x_{l,t}(\theta^*)\}$ is optimal for $\LCkunit(\bm{p})$, where $\theta^*$ is the unique solution to $\sum_{\tau=1}^{T-1}x_{k,\tau}(\theta^*)=1-\theta^*$.
\end{theorem}

\Cref{Prophetoptimaltheorem} shows that \Cref{constructprophet} constructs an optimal solution to $\LCkunit(\bm{p})$, \textit{as long as the $\theta$ is set as the optimal $\theta^*$}, as defined in \Cref{feasiproposition}.  This optimal $\theta^*$ is uniquely defined based on $\bm{\p}$.  \Cref{feasiproposition} further shows that any $\theta\le\theta^*$ is feasible, and hence if we can find a $\theta$ that is no greater than the $\theta^*$ arising from any $\bm{\p}$, then \Cref{constructprophet} will correspond to a $\price$-agnostic procedure for the $k$-unit prophet inequality or OCRS problem with a guarantee of $\theta$.

\subsection{Characterizing the Worst-case Distribution}\label{sec:worstkunit}
Our goal is now to find the $\bm{\p}$ such that the optimal objective value $\theta^*$ of $\LCkunit(\bm{p})$ in \eqref{dualdual} reaches its minimum. We would like to characterize the worst-case distribution and then compute the guarantee.

We first characterize the worst-case distribution for which the optimal objective value of $\LCkunit(\bm{p})$ reaches its minimum. Obviously, it is enough for us to consider only the $\bm{p}$ satisfying $\sum_{t=1}^{T}\p_t=k$. We show in the following lemma that splitting one query into two queries can only make the optimal objective value of $\LCkunit(\bm{p})$ become smaller, and thus, in the worst-case distribution, each $\p_t$ should be infinitesimally small.
\begin{lemma}\label{splititemlemma}
For any $\bm{\p}=(\p_1,\dots,\p_T)$ satisfying $\sum_{t=1}^{T}\p_t=k$, and any $\sigma\in[0,1]$, $1\leq q\leq T$, if we define a new sequence of arrival probabilities $\tilde{\bm{p}}=(\tilde{p}_1,\dots,\tilde{p}_{T+1})$ such that
\[\begin{aligned}
&\tilde{p}_t=\p_t~~~\forall t<q,~~~~\tilde{p}_q=\p_q\cdot \sigma,~~~~\tilde{p}_{q+1}=\p_q\cdot (1-\sigma)\text{~~and~~}\tilde{p}_{t+1}=\p_t~~~\forall q+1\leq t\leq T,
\end{aligned}\]
then it holds that $\LCkunit(\bm{\p})\geq\LCkunit(\tilde{\bm{p}})$.
\end{lemma}
The proof is relegated to \Cref{ProofLemma8}.
Now, for each $\bm{\p}=(\p_1,\dots,\p_T)$ satisfying $\sum_{t=1}^{T}\p_t=k$, we assume without loss of generality that $\p_t$ is a rational number for each $t$, i.e., $\p_t=\frac{n_t}{N}$ where $n_t$ is an integer for each $t$ and $N$ is an integer denoting the common denominator. We first split $\p_1$ into $\frac{1}{N}$ and $\frac{n_1-1}{N}$ to form a new sequence of arrival probabilities. By Lemma \ref{splititemlemma}, we know that such an operation can only decrease the optimal objective value of $\LCkunit(\bm{p})$. We then split $\frac{n_1-1}{N}$ into $\frac{1}{N}$ and $\frac{n_1-2}{N}$ and so on. In this way, we split $\p_1$ into $n_1$ copies of $\frac{1}{N}$ to form a new sequence of arrival probabilities and Lemma \ref{splititemlemma} guarantees that the optimal objective value of $\LCkunit(\bm{p})$ can only become smaller. We repeat the above operation for each $t$. Finally, we form a new sequence of arrival probabilities, denoted by $\bm{\p}^{N}=(\frac{1}{N},\dots,\frac{1}{N})\in\mathbb{R}^{Nk}$, and we have $\LCkunit(\bm{\p})\geq\LCkunit(\bm{\p}^N)$. Intuitively, when $N\rightarrow\infty$, then the optimal objective value of $\LCkunit(\bm{p})$ reaches its minimum. Note that when $N\rightarrow\infty$, we always have $\sum_{t=1}^{Nk}\p_t^{N}=k$, and then the Bernoulli arrival process approximates a Poisson process with rate $1$ over the time interval $[0,k]$. The above argument implies that the worst-case arrival process is a Poisson process.

Under the Poisson process, for each fixed ratio $\theta\in[0,1]$, our solution in \Cref{constructprophet} can be interpreted as a solution to an ordinary differential equation (ODE). We further note that for $\bm{p}^N$ and any $\theta\in[0,1]$, our solution in \Cref{constructprophet} can be regarded as the solution obtained from applying Euler's method to solve this ODE by uniformly discretizing the interval $[0, k]$ into $Nk$ discrete points. Then, for any fixed ratio $\theta\in[0,1]$, after showing the Lipschitz continuity of the function defining this ODE, we can apply the global truncation error theorem of Euler's method (Theorem 212A in \citet{butcher2008numerical}) to establish the solution under the Poisson process as the limit of the solution under $\bm{p}^N$ when $N\rightarrow\infty$. Based on this convergence, we can prove that the optimal value under the Poisson process is equivalent to $\lim_{N\rightarrow\infty}\LCkunit(\bm{p}^N)$, which is the optimal ratio we are looking for.

For general $k$, the values of $\sgk$ have previously been shown in \citet{wang2018online} through the analysis of a ``reflecting'' Poisson process. However, we show that these values $\sgk$ are \textit{optimal}, deriving them instead from $\LCkunit(\bm{p})$.
Moreover, \citet{wang2018online} assume Poisson arrivals to begin with, whereas we allow for arbitrary probability vectors $\price$ and show that the limiting Poisson case is the worst case.

Specifically in the case of $k=2$, we construct an example showing that relative to the weaker prophet benchmark $\mathbb{E}_{\bm{I}\sim\bm{F}}[V^{\text{off}}(\bm{I})]$ (the offline optimum itself rather than the LP upper bound $\text{UP}(\mathcal{H})$), it is not possible to do much better than $\gamma^*_2$.
Our construction is based on adapting the tight example relative to the stronger benchmark $\text{UP}(\mathcal{H})$. 
We note that this suggests that there is \textit{some separation} between optimal ex-ante vs. non-ex-ante prophet inequalities when $k>1$, which is not the case when $k=1$ (because they are both 1/2).
The formal proof of \Cref{Upperprophetproposition} below is relegated to \Cref{Proofproposition3}.

\begin{proposition}\label{Upperprophetproposition}
For the $2$-unit prophet inequality problem, it holds that $\inf_{\mathcal{H}}\frac{\mathbb{E}_{\pi, \bm{I}\sim\bm{F}}[V^\pi(\bm{I})]}{\mathbb{E}_{\bm{I}\sim\bm{F}}[V^{\text{off}}(\bm{I})]}\leq0.6269$ for any online algorithm $\pi$, while $\gamma^*_2\approx0.6148$.
\end{proposition}

We now discuss how the construction in \Cref{constructprophet} should be interpreted when the arrival process is a Poisson process. We find it is more convenient to work with the functions $\{\tilde{y}_{l,\theta}(\cdot)\}_{\forall l=1,\dots,k}$ over $[0,k]$, where $\tilde{y}_{l,\theta}(t)$ denotes the ex-ante probability that there is a query served as the $l$-th one during the period $[0,t]$. Note that the variable $x_{l,t}(\theta)$ denotes the ex-ante probability that there is a query accepted as the $l$-th query at time $t$, and hence we have $x_{l,t}(\theta)=d\tilde{y}_{l,\theta}(t)$.  We denote $\tilde{y}_{0,\theta}(t)=1$ for each $t\in[0,k]$. Then the functions $\{\tilde{y}_{l,\theta}(\cdot)\}_{\forall l=1,\dots,k}$ corresponding to the construction in \Cref{constructprophet} under Poisson arrivals can be interpreted as follows.
\begin{defn}\label{ODEdefinition}
Ordinary Differential Equation (ODE) formula under Poisson arrival
\begin{enumerate}
  \item For each fixed $\theta\in[0,1]$, we define $\tilde{y}_{0,\theta}(t)=1$ for each $t\in[0,k]$ and $t_1=0$.
  \item For each $l=1,2,\dots,k-1$, we do the following:
  \begin{enumerate}
    \item $\tilde{y}_{l,\theta}(t)=0$ when $t\leq t_l$.
    \item When $t_l\leq t\leq t_{l+1}$, it holds that
    \begin{equation}\label{ODE1}
     \frac{d\tilde{y}_{l,\theta}(t)}{dt}=\theta-\sum_{v=1}^{l-1}\frac{d\tilde{y}_{v,\theta}(t)}{dt}=\theta-1+\tilde{y}_{l-1,\theta}(t),~~~\forall t_l\leq t\leq t_{l+1},
    \end{equation}
    where $t_{l+1}$ is defined as the first time that $\tilde{y}_{l,\theta}(t_{l+1})=1-\theta$. If such a $t_{l+1}$ does not exist, we denote $t_{l+1}=k$.
    \item When $t_{l+1}\leq t\leq k$, it holds that
    \begin{equation}\label{ODE2}
    \frac{d\tilde{y}_{l,\theta}(t)}{dt}=\tilde{y}_{l-1,\theta}(t)-\tilde{y}_{l,\theta}(t),~~~\forall t_{l+1}\leq t\leq k.
    \end{equation}
  \end{enumerate}
  \item $\tilde{y}_{k,\theta}(t)=0$ if $t\leq t_k$ and $\frac{d\tilde{y}_{k,\theta}(t)}{dt}=\theta-1+\tilde{y}_{k-1,\theta}(t)$ if $t_k\leq t\leq k$.
\end{enumerate}
\end{defn}
Thus, by Theorem \ref{Prophetoptimaltheorem}, the solution to the equation $\tilde{y}_{k,\theta}(k)=1-\theta$ should be the minimum of the optimal objective value of $\LCkunit(\bm{p})$ in \eqref{dualdual}, which is the guarantee $\sgk$ we are looking for. The above arguments are formalized in the following theorem and the proof is relegated to \Cref{Prooftheorem4}. Note that the following \Cref{worstcasetheorem} is our ultimate result for the $k$-unit case, while \Cref{ODEdefinition} characterizes the ODE formula mentioned in \Cref{sec:newTechProph}. In the remaining part of this section, we will describe the computational procedure for $\sgk$.
\begin{theorem}\label{worstcasetheorem}
For each $\theta\in[0,1]$, denote by $\{\tilde{y}_{l,\theta}(\cdot)\}$ the functions defined in Definition \ref{ODEdefinition}. Then there exists a unique $\sgk\in[0,1]$
such that $\tilde{y}_{k,\sgk}(k)=1-\sgk$ and it holds that
\[
\sgk=~~\inf_{\bm{\p}}~\LCkunit(\bm{\p})~~\mbox{s.t.}~\sum_{t=1}^{T}\p_t=k.
\]
\end{theorem}

We now show that the ODE in Definition \ref{ODEdefinition} admits an analytical solution that enables us to compute $\sgk$ for each $k$. For each fixed $\theta$, when $l=1$, it is immediate that
\[
\tilde{y}_{1,\theta}(t)=\theta\cdot t,~~\text{when~}t\leq t_2=\frac{1-\theta}{\theta},\text{~~and~~}\tilde{y}_{1,\theta}(t)=1-\theta\cdot\exp(t_2-t),~~\text{for~}t_2\leq t\leq k.
\]
Now suppose that there exists a fixed $2\leq l\leq k$ such that for each $1\leq v\leq l-1$, it holds that
\[\begin{aligned}
&\tilde{y}_{v,\theta}(t)=\zeta_{v}+\theta\cdot t+\sum_{q=0}^{v-2}\zeta_{v,q}\cdot t^q\cdot\exp(-t),~~&&\text{when~}t_v\leq t\leq t_{v+1}\\
&\tilde{y}_{v,\theta}(t)=1+\sum_{q=0}^{v-1}\psi_{v,q}\cdot t^q\cdot\exp(-t),~~&&\text{when~} t_{v+1}\leq t\leq k
\end{aligned}\]
for some parameters $\{\zeta_{v},\zeta_{v,q}, \psi_{v,q}\}$, which are specified by $\theta$. Then by ODE \eqref{ODE1} and \eqref{ODE2}, it must hold that
\[\begin{aligned}
&\tilde{y}_{l,\theta}(t)=\zeta_{l}+\theta\cdot t+\sum_{q=0}^{l-2}\zeta_{l,q}\cdot t^q\cdot\exp(-t),~~&&\text{when~}t_l\leq t\leq t_{l+1}\\
&\tilde{y}_{l,\theta}(t)=1+\sum_{q=0}^{l-1}\psi_{l,q}\cdot t^q\cdot\exp(-t),~~&&\text{when~} t_{l+1}\leq t\leq k.
\end{aligned}\]
The parameters $\{\zeta_{l},\zeta_{l,q}, \psi_{l,q}\}$ can be computed in the following steps:
\begin{enumerate}
  \item Set $\zeta_{l,l-1}=0$ and compute $\zeta_{l,q}$ iteratively from $q=l-2$ up to $q=0$ by setting
  \[
  \zeta_{l,q}=(q+1)\cdot\zeta_{l,q+1}-\psi_{l-1,q}.
  \]
  \item Set the value of $\zeta_l$ such that $\tilde{y}_{l,\theta}(t_l)=0$. If $l=k$, we set $t_{l+1}=k$; otherwise, we set $t_{l+1}$ to be the solution to the following equation:
      \[
      1-\theta=\tilde{y}_{l,\theta}(t)=\zeta_{l}+\theta\cdot t+\sum_{q=0}^{l-2}\zeta_{l,q}\cdot t^q\cdot\exp(-t).
      \]
      Note that by definition $\tilde{y}_{l,\theta}(t)$ is monotone increasing with $t$, and hence we can do a bisection search on the interval $[t_l,k]$ to obtain the value of $t_{l+1}$.
  \item Set $\psi_{l,q}=\frac{\psi_{l-1,q-1}}{q}$ for each $q=1,\dots,l-1$. If $l<k$, the value of $\psi_{l,0}$ is determined such that
  \[
  1-\theta=1+\sum_{q=0}^{l-1}\psi_{l,q}\cdot t_{l+1}^q\cdot\exp(-t_{l+1}).
  \]
\end{enumerate}
Thus, for each fixed $\theta$, we can follow the above procedure to obtain the value of $\tilde{y}_{k,\theta}(k)$. Note that \Cref{feasilemma3} established in \Cref{prooffeasiproposition} implies that the value of $\tilde{y}_{k,\theta}(k)$ is monotone increasing with $\theta$, and hence we can do a bisection search on $\theta\in[0,1]$ to obtain the value of $\sgk$ as the unique solution of the equation $\tilde{y}_{k,\theta}(k)=1-\theta$. By \Cref{worstcasetheorem}, $\sgk$ is the optimal value for the guarantee. The above procedure describes how we compute numerically the value of $\sgk$ for each $k$ and the value of $\sgk$ is reported previously in \Cref{RatioTable} and \Cref{ComparisonRatioFigure}.

\section{Results for the Knapsack Setting}\label{randomsizesection}
In this section, we derive the tight guarantee of $\frac{1}{3+e^{-2}}$ for the knapsack OCRS problem, following the techniques outlined in \Cref{sec:newTechKnap}.

\subsection{Algorithm and Interpretation}\label{sec:AlgorithmKnapsack}

Our knapsack policy differs from existing ones for knapsack in an online setting \citep{dutting2020prophet, feldman2021online, stein2020advance} by eschewing the need to split queries into ``large'' vs. ``small'' based on whether its size is greater than 1/2.
In fact, we can show that any algorithm which considers large and small queries separately in our problem is limited to $\gamma\le1/4$, and hence could not match the $\frac{1}{3+e^{-2}}$ upper bound provided earlier. The result follows by considering a problem setup $\mathcal{H}$ where there are 4 queries and $(\tilde{r}_t, \tilde{d}_t)$ is realized as $({r}_t, {d}_t)$ with probability $p_t$ and is realized as $(0,0)$ otherwise, for each query $t$, and letting
\[
({r}_1, \p_1, {d}_1)=(r,1,\epsilon),~~({r}_2, \p_2, {d}_2)=({r}_3, \p_3, {d}_3)=(r, \frac{1-2\epsilon}{1+2\epsilon}, \frac{1}{2}+\epsilon),~~({r}_4, \p_4, {d}_4)=(r/\epsilon,\epsilon,1)
\]
for $r>0$ and some small $\epsilon>0$. We formalize the above arguments as follows, where the formal proof is relegated to \Cref{pfprop:largesmall}.

\begin{proposition}\label{Largesmallproposition}
If the policy $\pi$ serves only either ``large'' queries with a size larger than 1/2, or ``small'' queries with a size no larger than 1/2, then it holds that
$\inf_{{\mathcal{H}}}\frac{\mathbb{E}_{\bm{I}\sim\bm{F}}[V^\pi(\bm{I})]}{\text{UP}({\mathcal{H}})}\leq\frac{1}{4}$.
\end{proposition}

We now turn to the OCRS problem, formalizing our policy in \Cref{definedistributionrandomsize}.
Based on $\gamma$, for each $t$, we use $\tilde{X}_{t,\gamma}$ to denote the distribution of the capacity consumption under our policy at the end of period $t$, where $\tilde{X}_{0,\gamma}$ takes value $0$ deterministically. Then, for each size realization $d_t$ of query $t$, we specify a threshold $\eta_{t,\gamma}(d_t)$ such that the probability of $\tilde{X}_{t-1,\gamma}\in(\eta_{t,\gamma}(d_t),1-d_t]$  is smaller than or equal to $\gamma$, and the probability that  $\tilde{X}_{t-1,\gamma}\in[\eta_{t,\gamma}(d_t),1-d_t]$ is larger than or equal to $\gamma$. When the size of query $t$ is realized as $\tilde{d}_t$, we serve query $t$ when the realized capacity consumption is among $(\eta_{t,\gamma}(\tilde{d}_t),1-\tilde{d}_t]$, or we serve query $t$ with a certain probability, when the realized capacity consumption equals $\eta_{t,\gamma}( \tilde{d}_t)$. {If $\gamma$ is feasible such that $\tilde{X}_{t,\gamma}$ is well-defined for every $t=1,\dots,T$,} i.e., $\eta_{t,\gamma}(d_t)$ exists for all possible sizes $d_t$, it is clear to see that our policy guarantees that query $t$ is served with a total probability $\gamma$. We finally update the distribution of capacity consumption in step \ref{Updaterandomdistribution}.
In this section we establish guarantees while ignoring implementation runtime; in \Cref{sec:knapsackPolytimeImpl} we show how through discretization, a runtime polynomial in $K$ is attainable while losing only an additive $O(1/K)$ in the guarantee, for any large integer $K$. The final algorithm for the multi-resource setting is presented in \Cref{sec:FinalAlgorithm}. 
\begin{algorithm}
\caption{Pre-processed algorithm for the knapsack OCRS problem ($\pi_{\gamma}$)}
\begin{algorithmic}[1]
\State \textbf{Input}: a parameter $\gamma$ and the sequence of probabilities $\{p_t(d_t), \forall t, \forall d_t\}$ satisfying $\sum_t \sum_{d_t} p_t(d_t)\cdot d_t\leq 1$ and $\sum_{d_t} p_t(d_t)\le 1$ for each $t$.
\State We initialize $\tilde{X}_{0,\gamma}$ as a random variable that takes the value $0$ deterministically.
\For{$t=1,2,\dots,T$,}
\State For each realization $d_t>0$, we denote a threshold $\eta_{t,\gamma}(d_t)$ satisfying:
        \begin{equation}\label{defineetarandomsize}
        P(\eta_{t,\gamma}(d_t)<\tilde{X}_{t-1,\gamma}\leq1-d_t)\leq\gamma\leq P(\eta_{t,\gamma}(d_t)\leq\tilde{X}_{t-1,\gamma}\leq1-d_t).
        \end{equation}
\State Initialize $\tilde{X}_{t, \gamma}$ as $\tilde{X}_{t-1, \gamma}$. We now update $\tilde{X}_{t, \gamma}$ to reflect the distribution of the capacity utilization at the end of period $t-1$ via the following procedures. We first let the distribution of $\tilde{X}_{t, \gamma}$ be identical to $\tilde{X}_{t-1,\gamma}$. Then, for each size realization $d_t>0$ and each point $x\in (\eta_{t,\gamma}(d_t) ,1-d_t]$, we decrease the probability of $P(\tilde{X}_{t, \gamma}=x)$ by $p_t(d_t)\cdot P(\tilde{X}_{t-1,\gamma}=x)$ and increase the probability of $P(\tilde{X}_{t, \gamma}=x+d_t)$ by $p_t(d_t)\cdot P(\tilde{X}_{t-1,\gamma}=x)$. For each size realization $d_t$ and $x=\eta_{t,\gamma}(d_t)$, we decrease the probability of $P(\tilde{X}_{t, \gamma}=x)$ by $p_t(d_t)\cdot(\gamma-P(\eta_{t,\gamma}(d_t)<\tilde{X}_{t-1,\gamma}\leq1-d_t))$ and increase the probability of $P(\tilde{X}_{t, \gamma}=x+d_t)$ by $p_t(d_t)\cdot(\gamma-P(\eta_{t,\gamma}(d_t)<\tilde{X}_{t-1,\gamma}\leq1-d_t))$.\label{Updaterandomdistribution}
\EndFor
\State \textbf{Output}: the distributions $\{\tilde{X}_{t,\gamma}, \forall t\}$ and the thresholds $\{ \eta_{t, \gamma}(d_t), \forall t, \forall d_t \}$. 
\end{algorithmic}
\label{definedistributionrandomsize}
\end{algorithm}

\textbf{Remarks about \Cref{definedistributionrandomsize}.}
\begin{enumerate}
\item The policy described in \Cref{definedistributionrandomsize} is implemented as follows: when each query $t=1,\ldots,T$ arrives, conditional it taking size $\tilde{d}_t$ and the capacity utilization being $\tilde{X}_{t-1,\gamma}$, serve query $t$ w.p.~1 if $\eta_{t,\gamma}(\tilde{d}_t)< \tilde{X}_{t-1,\gamma}\leq1-\tilde{d}_t$, w.p.~$\frac{\gamma-P(\eta_{t,\gamma}(\tilde{d}_t)<\tilde{X}_{t-1,\gamma}\leq1-\tilde{d}_t)}{P(\eta_{t,\gamma}(\tilde{d}_t)=\tilde{X}_{t-1,\gamma})}$ if $\tilde{X}_{t-1,\gamma}=\eta_{t,\gamma}(\tilde{d}_t)$, and w.p.~0 otherwise.
It will be preserved that $\tilde{X}_{t-1,\gamma}$ is the true distribution of capacity utilization when query $t$ arrives, and hence conditional on any size $\tilde{d}_t$ taken by query $t$, it will be served with probability exactly $$P(\eta_{t,\gamma}(\tilde{d}_t)<\tilde{X}_{t-1,\gamma}\leq1-\tilde{d}_t)+\frac{\gamma-P(\eta_{t,\gamma}(\tilde{d}_t)<\tilde{X}_{t-1,\gamma}\leq1-\tilde{d}_t)}{P(\eta_{t,\gamma}(\tilde{d}_t)=\tilde{X}_{t-1,\gamma})}P(\eta_{t,\gamma}(\tilde{d}_t)=\tilde{X}_{t-1,\gamma})=\gamma.$$
\item This will only describe a valid policy if parameter $\gamma$ leads to threshold values $\eta_{t,\gamma}(d_t)$ that are well-defined (i.e., \eqref{defineetarandomsize} can be satisfied) for all $t$ and $d_t$.  Unlike the $k$-unit case, we do not characterize the optimal value of $\gamma$ for given size distributions; we instead prove that setting $\gamma=1/(3+e^{-2})$ leads to well-defined thresholds for \textit{any} size distributions.  This again implies that the thresholds $\eta_{t,1/(3+e^{-2})}(d_t)$ can be constructed on-the-fly and that our algorithm and analysis hold even if an adaptive adversary chooses the arrival order of queries.
\end{enumerate}


\subsection{Proof of Guarantee and Tightness}
In this section, we analyze the guarantee of our policy in \Cref{definedistributionrandomsize}.  We also show that no policy can do better.
The key point is to find the largest possible $\gamma$ such that the policy $\pi_{\gamma}$ is feasible for all problem setups $\mathcal{H}$, i.e., the random variables $\tilde{X}_{t,\gamma}$ are well defined in that $\eta_{t,\gamma}(d_t)$ exists for all possible sizes $d_t$, for each $t$.

We now find such a $\gamma$. For any $a$ and $b$, denote $\mu_{t,\gamma}(a,b]=P(a<\tilde{X}_{t,\gamma}\leq b)$ assuming $\tilde{X}_{t,\gamma}$ is well defined.
Following the rules of the Best-fit Magician, we can establish an \textit{invariant} that upper-bounds the measure of sample paths with utilization in $(0,b]$ by a decreasing exponential function of the measure with utilization in $(b,1-b]$. Our invariant holds for all $b\in(0,1/2]$, at all times $t$.
\begin{lemma}\label{generalsamplepathlemma}
For any $0<b\leq\frac{1}{2}$ and any $0<\gamma<1$ such that $\tilde{X}_{t,\gamma}$ is well-defined in that $\eta_{t,\gamma}(d_t)$ exists for all possible sizes $d_t$, the following inequality
\begin{equation}\label{generalsamplepathinequality}
\frac{1}{\gamma}\cdot\mu_{t,\gamma}(0,b]\leq\exp(-\frac{1}{\gamma}\cdot\mu_{t,\gamma}(b,1-b])
\end{equation}
holds for all $t=0,1,\dots,T$.
\end{lemma}
We omit proving \Cref{generalsamplepathlemma} since we will prove a more general \Cref{UDsamplepathlemma} in \Cref{sec:unitDensity}.
For a fixed $t$, assume that the random variable $\tilde{X}_{t,\gamma}$ is well defined. Then, given the invariant \eqref{generalsamplepathinequality} established in \Cref{generalsamplepathlemma}, 
we show that a $\gamma$ as large as $\frac{1}{3+e^{-2}}\approx 0.319$ allows for a $\gamma$-measure of sample paths to have zero utilization at time $t$, which implies that the random variable $\tilde{X}_{t+1,\gamma}$ is well defined. We iteratively apply the above arguments for each $t=1$ up to $t=T$, and hence we prove the feasibility of our Best-fit Magician policy. The above arguments are formalized in the following theorem and the proof is relegated to \Cref{Prooftheorem2}.
\begin{theorem}\label{knapsackratiotheorem}
When $\gamma=\frac{1}{3+e^{-2}}$, the Best-fit Magician policy $\pi_{\gamma}$ in \Cref{definedistributionrandomsize} is feasible and has a guarantee at least $\frac{1}{3+e^{-2}}$.
\end{theorem}
Finally, we show that the guarantee $\gamma=\frac{1}{3+e^{-2}}$ is tight. The proof is completed by bounding the largest ratio of the knapsack OCRS problem when the size of query $t$ is realized to be $d_t$ with probability $p_t$. The sequence $\{(p_t, d_t), \forall t\}$ is specified as follows,
\begin{equation}\label{counterexample}
(\p_1, d_1)=(1, \epsilon),~~(\p_t, d_t)=(\frac{1-2\epsilon}{(T-2)(\frac{1}{2}+\epsilon)}, \frac{1}{2}+\epsilon)\text{~~for~all~}2\leq t\leq T-1\text{~~and~~}(\p_T,d_T)=(\epsilon, 1)
\end{equation}
for some $\epsilon>0$. Our result is formally stated in the following theorem, with the proof relegated to \Cref{pf:uppertheorem}.
\begin{theorem}\label{uppertheorem}
For any feasible online policy $\pi$, it holds that $\inf_{{\mathcal{H}}}\frac{\mathbb{E}_{{\bm{I}}\sim {\bm{F}}}[V^\pi({\bm{I}})]}{\text{UP}({\mathcal{H}})}\leq\frac{1}{3+e^{-2}}$.
\end{theorem}

\section{Extensions for the Knapsack Setting}\label{sec:extensions}
In this section, we discuss the polynomial-time implementation of our Best-fit Magician policy (\Cref{sec:knapsackPolytimeImpl}), and present our improvement in the unit-density special case (\Cref{sec:unitDensity}).

\subsection{A Polynomial-time Implementation Scheme} \label{sec:knapsackPolytimeImpl}
In this section, we discuss how our Best-fit Magician policy can be implemented in polynomial time. Note that the key step in our algorithm is to iteratively compute the distribution of $\tilde{X}_{t, \gamma}$.
Our approach is to discretize the possible sizes to always be a multiple of $\frac{1}{KT}$, for some large integer $K$.  Then
the support set of $\tilde{X}_{t,\gamma}$ contains at most $KT$ elements for each $t$, from which it can then be seen that our algorithm can be implemented in $O(KT^2)$ time.

For any problem setup $\mathcal{H}$ restricted to the single-resource setting, we perform this discretization by rounding each potential size $d_t$ \textit{up} to the nearest multiple of $\frac{1}{KT}$.  We denote the rounded sizes using $d'_t$. We then define $F'_t(\cdot)$ as the distribution of $(r_t, d_t')$, which is a discretization of the distribution $F_t(\cdot)$, and we define $\mathcal{H}'=(F'_1, \dots, F'_T)$. The implementation of our algorithm on the problem setup $\mathcal{H}$ is equivalent to the implementation on $\mathcal{H}'$. To be more specific, we compute the distribution of $\tilde{X}_{t,\gamma}$ based on $F'_t(\cdot)$ for each $t$, and, when we face a query $(r_t, d_t)$, we treat this query as a query with reward $r_t$ and a size $d_t'$.

We now discuss the loss of the guarantee of our algorithm due to such discretization. Note that since $d_t\le d'_t$ for each $t$, if the algorithm attempts any solution that is feasible for $\mathcal{H}'$, then it is also feasible for $\mathcal{H}$. Moreover, since $r_t$ is never changed, it is easy to see that the total reward collected by our algorithm in the problem setup $\mathcal{H}$ satisfies
\[
\mathbb{E}_{\pi_\gamma, \bm{I}\sim\mathcal{H}}[V^{\pi_{\gamma}}(\bm{I})]=\mathbb{E}_{\pi_\gamma, \bm{I}\sim\mathcal{H}'}[V^{\pi_{\gamma}}(\bm{I})]
\]
and, when $\gamma\leq\frac{1}{3+e^{-2}}$, it holds that
\[
\mathbb{E}_{\pi_\gamma, \bm{I}\sim\mathcal{H}'}[V^{\pi_{\gamma}}(\bm{I})]\geq\gamma\cdot\text{UP}(\mathcal{H}').
\]

It only remains to compare $\text{UP}(\mathcal{H})$ and $\text{UP}(\mathcal{H}')$. {Note that the sizes in $\mathcal{H}'$ are bigger, but since at most $\frac{1}{KT}$ can be added to the size of each of the $T$ queries, the total size added is at most $1/K$.  We will use this to argue that $\text{UP}(\mathcal{H}')\ge\frac{1}{1+1/K}\cdot\text{UP}(\mathcal{H})=(1-\frac{1}{K+1})\cdot\text{UP}(\mathcal{H})$.}

Denote by $\{x^*_t(r_t, d_t)\}$ the optimal solution to $\text{UP}(\mathcal{H})$, and denote
\[
\hat{x}_t(r_t, d_t')=\frac{K}{K+1}\cdot \mathbb{E}_{(r_t, \tilde{d}_t)\sim F_t} \left[x^*_t(r_t, \tilde{d}_t)|\tilde{d}_t\text{~is~rounded~up~to~}d'_t \right]~~~\forall (r_t, d_t').
\]
Note that if $d_t$ is rounded up to $d_t'$, it must hold that $d_t'-d_t\leq\frac{1}{KT}$.
Then, we have
\[\begin{aligned}
\sum_{t=1}^{T}\mathbb{E}_{(\tilde{r}_t, \tilde{d}_t')\sim F_t'}[\tilde{d}_t'\cdot \hat{x}_t(\tilde{r}_t, \tilde{d}_t')]&\leq\frac{K}{K+1}\cdot\sum_{t=1}^{T}\mathbb{E}_{(\tilde{r}_t, \tilde{d}_t)\sim F_t}[\tilde{d}_t\cdot x^*_t(\tilde{r}_t, \tilde{d}_t)]+\sum_{t=1}^{T}\mathbb{E}_{(\tilde{r}_t, \tilde{d}_t')\sim F_t'}[\frac{1}{KT}\cdot \hat{x}_t(\tilde{r}_t, \tilde{d}_t')]\\
&\leq \frac{K}{K+1}+\sum_{t=1}^{T}\frac{1}{KT}\cdot\frac{K}{K+1}=1,
\end{aligned}\]
where the last inequality holds by the feasibility of $\{x^*_t(r_t, d_t)\}$ and $\hat{x}_t(r_t, d_t')\in[0, \frac{K}{K+1}]$ for each $(r_t, d_t')$. We conclude that $\{\hat{x}_t(r_t, d_t')\}$ is a feasible solution to $\text{UP}(\mathcal{H}')$. Thus, it holds that
\[
\text{UP}(\mathcal{H}')\geq \sum_{t=1}^{T}\mathbb{E}_{(\tilde{r}_t, \tilde{d}_t')\sim F_t'}[\tilde{r}_t\cdot \hat{x}_t(\tilde{r}_t, \tilde{d}_t')]=\frac{K}{K+1}\cdot \sum_{t=1}^{T}\mathbb{E}_{(\tilde{r}_t, \tilde{d}_t)\sim F_t}[\tilde{r}_t\cdot x^*_t(\tilde{r}_t, \tilde{d}_t)]=\frac{K}{K+1}\cdot \text{UP}(\mathcal{H}),
\]
which implies that
\[
\mathbb{E}_{\pi_\gamma, \bm{I}\sim\mathcal{H}}[V^{\pi_{\gamma}}(\bm{I})]\geq\frac{K}{K+1}\cdot\gamma\cdot \text{UP}(\mathcal{H})
\]
when $\gamma\leq\frac{1}{3+e^{-2}}$.
In this way, we show how our algorithm can be implemented in $O(KT^2)$ time to achieve a guarantee of $\frac{K}{K+1}\cdot \frac{1}{3+e^{-2}}$.

\subsection{Improvement in the Unit-density Special Case} \label{sec:unitDensity}

In this section, we consider the unit-density special case of our online knapsack problem where $\tilde{r}_t=\tilde{d}_t$ for each $t$. Then we can suppress the notation $(\tilde{r}_t, \tilde{d}_t)$ and simply use $\tilde{d}_t$ to denote the size and the reward of query $t$. We modify our previous Best-fit Magician policy to obtain an improved guarantee. 
Following the LP reduction described in \Cref{sec:multiToSingle}, we restrict to the single-resource problem and obtain a sequence of probabilities $\{p_t(d), \forall d, \forall t\}$ by solving the LP relaxation \eqref{LPuppermultidimension}.
To maximize the total collected reward, it is enough for us to maximize the expected capacity utilization.

We now motivate our policy. Note that for the general case where the reward can be arbitrarily different from the size, our Best-fit Magician policy guarantees that each query is served with a \textit{common} ex-ante probability $\gamma$ after the LP relaxation. In this way, each query is treated ``equally'' so that no ``extreme'' reward can be assigned to the query with the smallest ex-ante probability, which would worsen the guarantee of our algorithm. However, in the unit-density case where the reward of each query is restricted to equal its size, it is no longer essential for us to treat each query ``equally.'' Instead, we will serve the \textit{later-arriving queries} with a \textit{smaller} probability to maximize capacity utilization. Our idea can be illustrated through the following example.

\textbf{Example $1$.} We focus on the following example with $4$ queries to illustrate how to maximize capacity utilization, where the size of each query $t$ is $d_t$ and query $t$ becomes active with probability $p_t$:
\[
(\p_1,d_1)=(\frac{2}{3}, \frac{1}{2}),~ (\p_2,d_2)=(\frac{2}{3}, \frac{1}{2}),~ (\p_3,d_3)=(1-\epsilon, \frac{1}{3}) \text{~and~} (\p_4,d_4)=(\frac{\epsilon}{3}, 1).
\]
Note that if we apply the Best-fit Magician policy with a ratio $\gamma$, then the feasible condition of our policy is
\begin{equation}\label{newexamplefeasibility}
P(\tilde{X}_{3,\gamma}=0)=1-\frac{8\gamma}{9}-\frac{5\gamma(1-\epsilon)}{9}\geq\gamma,
\end{equation}
which implies that $\gamma\leq\frac{9}{22}$ as $\epsilon\rightarrow0$. Thus, we conclude that the Best-fit Magician policy can guarantee an expected capacity utilization of at most $\frac{9}{22}$ for this example. However, if we serve each query with a different probability, i.e., if we serve query $t$ with probability $\gamma_t$ whenever it arrives, then the feasibility condition \eqref{newexamplefeasibility} becomes $P(\tilde{X}_{3, \gamma_1, \gamma_2, \gamma_3}=0)\geq\gamma_4$. We can set $\gamma_4=0$, which enables us to set a larger value for $\gamma_1, \gamma_2$, and $\gamma_3$, which in turn leads to a larger capacity utilization.

To be more specific, if we denote by $\tilde{X}_t$ the distribution of capacity utilization at the end of period $t=1,\dots,4$ under the serving probabilities $\gamma_1, \gamma_2, \gamma_3, \gamma_4$, then the distribution of $\tilde{X}_1$ is
\[
P(\tilde{X}_1=\frac{1}{2})=\p_1\cdot\gamma_1=\frac{2\gamma_1}{3}\text{~~and~~}  P(\tilde{X}_1=0)=1-\frac{2\gamma_1}{3}.
\]
We still proceed to serve query $2$; however, we serve it with probability $\gamma_2$ whenever it arrives. Then, we obtain the following distribution of $\tilde{X}_2$:
\[\begin{aligned}
&P(\tilde{X}_2=1)=\gamma_1\cdot\p_1\cdot\p_2=\frac{4\gamma_1}{9},~~P(\tilde{X}_2=\frac{1}{2})=\gamma_2\cdot\p_2-\gamma_1\cdot\p_1\cdot\p_2+\p_1\cdot\gamma_1-\gamma_1\cdot\p_1\cdot\p_2=
\frac{2\gamma_2}{3}-\frac{2\gamma_1}{9}\\
&P(\tilde{X}_2=0)=1-P(\tilde{X}_2=1)-P(\tilde{X}_2=\frac{1}{2})=1-\frac{2\gamma_1}{9}-\frac{2\gamma_2}{3}.
\end{aligned}\]
Finally, we serve query $3$ with probability $\gamma_3$ whenever it arrives. Then, the distribution of $\tilde{X}_3$ can be obtained as follows:
\[\begin{aligned}
&P(\tilde{X}_3=1)=\frac{4\gamma_1}{9},~~P(\tilde{X}_3=\frac{5}{6})=(\frac{2\gamma_2}{3}-\frac{2\gamma_1}{9})\cdot(1-\epsilon),~~P(\tilde{X}_3=\frac{5}{6})=(\frac{2\gamma_2}{3}-\frac{2\gamma_1}{9})\cdot\epsilon,\\
&P(\tilde{X}_3=\frac{1}{3})=(\gamma_3-\frac{2\gamma_2}{3}+\frac{2\gamma_1}{9})\cdot(1-\epsilon),~~P(\tilde{X}_3=0)=1-\frac{4\gamma_1}{9}-\gamma_3+(\frac{2\gamma_1}{9}-\frac{2\gamma_2}{3}+\gamma_3)\cdot\epsilon.
\end{aligned}\]
Note that as long as $P(\tilde{X}_t=0)\geq\gamma_{t+1}$ for each $t=1,2,3$, the random variables $\{\tilde{X}_t\}_{t=1,2,3}$ are well defined and the above procedure is feasible. Obviously, the probabilities $P(\tilde{X}_t=0)$ are decreasing in $t$, which corresponds to the fact that the measure of the sample paths with no capacity consumed decreases over time. Thus, it is natural to set $\gamma_t$ to decrease in terms of $t$. Specifically, we can set
\[
\gamma_2=P(\tilde{X}_1=0)=1-\frac{2\gamma_1}{3}\text{~~and~~}\gamma_3=P(\tilde{X}_2=0)=1-\frac{2\gamma_1}{9}-\frac{2\gamma_2}{3}=\frac{1}{3}+\frac{2\gamma_1}{9}
\]
as a function of $\gamma_1$. We can also set $\gamma_4=0$ since the expected capacity utilization of the last query is $0$ as $\epsilon\rightarrow0$. Then, the only feasibility conditions we need to satisfy are
\[
\gamma_1\geq\gamma_2\geq\gamma_3\text{~~and~~}  P(\tilde{X}_3=0)\geq\gamma_4=0,
\]
which implies that $\gamma_1$ can be set as large as $3/4$ when $\epsilon\rightarrow0$. Then $\gamma_2=\frac{1}{2}$ and $\gamma_3=\frac{1}{2}$ as $\epsilon\rightarrow0$. Thus, we can guarantee a capacity utilization of $\frac{7}{12}$, which improves on the previous utilization of $\frac{9}{22}$ under the Best-fit Magician policy.
\Halmos

We now present our policy in \Cref{UDpolicy}, which is based on a sequence of probabilities with which we serve each query based on its LP relaxation value, denoted by $\boldsymbol{\gamma}=(\gamma_1, \dots, \gamma_T)$ and satisfying $1\geq\gamma_1\geq\gamma_2\geq\dots\geq\gamma_T\geq0$. We will specify later how to determine the vector $\boldsymbol{\gamma}$ such that our policy is feasible and achieves the improved guarantee. The final algorithm for the multi-resource setting is presented in \Cref{sec:FinalAlgorithm}.

\begin{algorithm}
\caption{Pre-processed algorithm for unit-density special case}
\begin{algorithmic}[1]
\State Input: a sequence $\boldsymbol{\gamma}$ satisfying $1\geq\gamma_1\geq\gamma_2\geq\dots\geq\gamma_T\geq0$ and a sequence of probabilities $\{p_t(d_t), \forall t,\forall d_t\}$ satisfying $\sum_t \sum_{d_t} p_t(d_t)\cdot d_t\leq 1$ and $\sum_{d_t} p_t(d_t)$ for each $t$
\State We initialize $\tilde{X}_{0,\boldsymbol{\gamma}}$ as a random variable taking value $0$ with probability $1$.
\For{$t=1,2,\dots,T$}
For each realization $d_t>0$, we denote a threshold $\eta_{t,\boldsymbol{\gamma}}(d_t)$ satisfying
        \begin{equation}\label{defineetaUD}
        P(\eta_{t,\boldsymbol{\gamma}}(d_t)<\tilde{X}_{t-1,\boldsymbol{\gamma}}\leq1-d_t)\leq\gamma_t\leq P(\eta_{t,\boldsymbol{\gamma}}(d_t)\leq\tilde{X}_{t-1,\boldsymbol{\gamma}}\leq1-d_t).
        \end{equation}
\State Initialize $\tilde{X}_{t, \boldsymbol{\gamma}}$ as $\tilde{X}_{t-1, \boldsymbol{\gamma}}$. We now update $\tilde{X}_{t, \boldsymbol{\gamma}}$ to reflect the distribution of the capacity utilization at the end of period $t-1$ via the following procedures. We first let the distribution of $\tilde{X}_{t, \boldsymbol{\gamma}}$ be identical to $\tilde{X}_{t-1,\boldsymbol{\gamma}}$. Then, for each size realization $d_t>0$ and each point $x\in (\eta_{t,\boldsymbol{\gamma}}(d_t) ,1-d_t]$, we decrease the probability of $P(\tilde{X}_{t, \boldsymbol{\gamma}}=x)$ by $p_t(d_t)\cdot P(\tilde{X}_{t-1,\boldsymbol{\gamma}}=x)$ and increase the probability of $P(\tilde{X}_{t, \boldsymbol{\gamma}}=x+d_t)$ by $p_t(d_t)\cdot P(\tilde{X}_{t-1,\boldsymbol{\gamma}}=x)$. For each size realization $d_t>0$ and $x=\eta_{t,\boldsymbol{\gamma}}(d_t)$, we decrease the probability of $P(\tilde{X}_{t, \boldsymbol{\gamma}}=x)$ by $p_t(d_t)\cdot(\gamma_t-P(\eta_{t,\boldsymbol{\gamma}}(d_t)<\tilde{X}_{t-1,\boldsymbol{\gamma}}\leq1-d_t))$ and increase the probability of $P(\tilde{X}_{t, \boldsymbol{\gamma}}=x+d_t)$ by $p_t(d_t)\cdot(\gamma_t-P(\eta_{t,\boldsymbol{\gamma}}(d_t)<\tilde{X}_{t-1,\boldsymbol{\gamma}}\leq1-d_t))$.
\EndFor
\State Output: the distributions $\{\tilde{X}_{t,\bm{\gamma}}, \forall t\}$ and the thresholds $\{ \eta_{t, \bm{\gamma}}(d_t), \forall t, \forall d_t \}$. 
\end{algorithmic}
\label{UDpolicy}
\end{algorithm}

\textbf{Remarks about \Cref{UDpolicy}.}
\begin{enumerate}
\item The policy described in \Cref{UDpolicy} is implemented as follows: when each query $t=1,\ldots,T$ arrives, conditional it taking size $\tilde{d}_t$ and the capacity utilization being $\tilde{X}_{t-1,\boldsymbol{\gamma}}$, serve query $t$ w.p.~1 if $\eta_{t,\boldsymbol{\gamma}}(\tilde{d}_t)< \tilde{X}_{t-1,\boldsymbol{\gamma}}\leq1-\tilde{d}_t$, w.p.~$\frac{\gamma_t-P(\eta_{t,\boldsymbol{\gamma}}(\tilde{d}_t)<\tilde{X}_{t-1,\boldsymbol{\gamma}}\leq1-\tilde{d}_t)}{P(\eta_{t,\boldsymbol{\gamma}}(\tilde{d}_t)=\tilde{X}_{t-1,\boldsymbol{\gamma}})}$ if $\tilde{X}_{t-1,\boldsymbol{\gamma}}=\eta_{t,\boldsymbol{\gamma}}(\tilde{d}_t)$, and w.p.~0 otherwise.
It will be preserved that $\tilde{X}_{t-1,\boldsymbol{\gamma}}$ is the true distribution of capacity utilization when query $t$ arrives, and hence conditional on any size $\tilde{d}_t$ taken by query $t$, it will be served with probability exactly $$P(\eta_{t,\boldsymbol{\gamma}}(\tilde{d}_t)<\tilde{X}_{t-1,\boldsymbol{\gamma}}\leq1-\tilde{d}_t)+\frac{\gamma_t-P(\eta_{t,\boldsymbol{\gamma}}(\tilde{d}_t)<\tilde{X}_{t-1,\boldsymbol{\gamma}}\leq1-\tilde{d}_t)}{P(\eta_{t,\boldsymbol{\gamma}}(\tilde{d}_t)=\tilde{X}_{t-1,\boldsymbol{\gamma}})}P(\eta_{t,\boldsymbol{\gamma}}(\tilde{d}_t)=\tilde{X}_{t-1,\boldsymbol{\gamma}})=\gamma_t.$$
\item \Cref{UDpolicy} will only describe a valid policy if the sequence of parameters $\boldsymbol{\gamma}$ leads to threshold values $\eta_{t,\boldsymbol{\gamma}}(d_t)$ that are well-defined (i.e., \eqref{defineetaUD} can be satisfied) for all $t$ and $d_t$. Note that if the sequence $\boldsymbol{\gamma}$ is uniform, i.e., $\gamma_1=\dots=\gamma_T$, then the modified Best-fit Magician policy in \Cref{UDpolicy} is identical to the Best-fit Magician policy in \Cref{definedistributionrandomsize}. The rest of this section is devoted to determining the sequence $\boldsymbol{\gamma}$ such that our policy in \Cref{UDpolicy} is feasible. Once the sequence $\boldsymbol{\gamma}$ is determined, the threshold $\eta_{t,\boldsymbol{\gamma}}(d_t)$ can be constructed on-the-fly and our algorithm and analysis hold even if an adaptive adversary chooses the arrival order of queries.
\end{enumerate}


Our approach to determine the sequence $\boldsymbol{\gamma}$ relies crucially on the distribution of capacity utilization at the end of each period $t$, denoted by $\tilde{X}_{t,\boldsymbol{\gamma}}$, assuming that the sequence $\boldsymbol{\gamma}$ is feasible.
For any $a$ and $b$, we denote $\mu_{t,\boldsymbol{\gamma}}(a,b]=P(a<\tilde{X}_{t,\boldsymbol{\gamma}}\leq b)$, where $\tilde{X}_{t,\boldsymbol{\gamma}}$ denotes the distribution of capacity utilization at the end of period $t$ in \Cref{UDpolicy}. Then, we can still establish the following invariant, which generalizes Lemma \ref{generalsamplepathlemma} from the uniform sequence to any sequence $\boldsymbol{\gamma}$ satisfying $1\geq\gamma_1\geq\dots\geq\gamma_T\geq0$.
\begin{lemma}\label{UDsamplepathlemma}
For any $0<b\leq\frac{1}{2}$ and any sequence $\boldsymbol{\gamma}$ satisfying $1\geq\gamma_1\geq\dots\geq\gamma_T\geq0$ such that $\tilde{X}_{t,\boldsymbol{\gamma}}$ is well defined, the following inequality
\begin{equation}\label{UDsamplepathinequality}
\frac{1}{\gamma_1}\cdot\mu_{t,\boldsymbol{\gamma}}(0,b]\leq\exp(-\frac{1}{\gamma_1}\cdot\mu_{t,\boldsymbol{\gamma}}(b,1-b])
\end{equation}
holds for any $t=1,\dots,T$.
\end{lemma}
Note that the ``difficult'' case for proving the invariant in  \Cref{generalsamplepathlemma} corresponds to when there is a probability mass moved from point $0$ during the definition of $\tilde{X}_{t,\gamma}$. Then, both $\mu_t(0,b]$ and $\mu_t(b,1-b]$ can become larger, for some $b\in(0,1/2)$. However, when $\boldsymbol{\gamma}$ is a non-increasing sequence, the amount of probability mass that is moved from $0$ into either the interval $(0,b]$ or $(b,1-b]$ under $\gamma_t$ is smaller than the one under $\gamma_1$, which makes the invariant easier to hold.
The proof of \Cref{UDsamplepathlemma} is relegated to \Cref{proofUDlemma}. Using \Cref{UDsamplepathlemma}, we can modify the proof of Theorem \ref{knapsackratiotheorem} to obtain the following result, which will finally lead to our choice of the feasible sequence $\boldsymbol{\gamma}$ and the guarantee of our algorithm.
\begin{theorem}\label{UDemptyboundtheorem}
For any $t$, denote $\psi_t=\mathbb{E}_{\tilde{d}_t\sim F_t}[\tilde{d}_t]$. Then for any sequence $\boldsymbol{\gamma}$ satisfying $1\geq\gamma_1\geq\dots\geq\gamma_T\geq0$ such that $\tilde{X}_{t,\boldsymbol{\gamma}}$ is well defined, the following inequality
\begin{equation}\label{UDbound1}
P(\tilde{X}_{t,\boldsymbol{\gamma}}=0)\geq\min\{1-\gamma_1-\sum_{\tau=1}^{t}\gamma_{\tau}\cdot\psi_{\tau},~~1-2\cdot\sum_{\tau=1}^{t}\gamma_{\tau}\cdot \psi_\tau-\gamma_1\cdot \exp(-\frac{2}{\gamma_1}\cdot\sum_{\tau=1}^{t}\gamma_{\tau}\cdot \psi_\tau) \}
\end{equation}
holds for each $t=1,\dots,T$.
\end{theorem}
The proof is relegated to \Cref{proofUDemptytheorem}. Note that for each $t$, if the random variables $\tilde{X}_{\tau,\boldsymbol{\gamma}}$ are well defined for each $\tau\leq t$, and $\gamma_{t+1}$ satisfies
\begin{equation}\label{UDgammain}
0\leq\gamma_{t+1}\leq \min\{1-\gamma_1-\sum_{\tau=1}^{t}\gamma_{\tau}\cdot\psi_{\tau},~~1-2\cdot\sum_{\tau=1}^{t}\gamma_{\tau}\cdot \psi_\tau-\gamma_1\cdot \exp(-\frac{2}{\gamma_1}\cdot\sum_{\tau=1}^{t}\gamma_{\tau}\cdot \psi_\tau) \},
\end{equation}
then \eqref{UDbound1} implies that $P(\tilde{X}_{t,\boldsymbol{\gamma}}=0)\geq\gamma_{t+1}$. Thus, we know that there always exists a threshold $\eta_{t+1,\boldsymbol{\gamma}}(d_t)$ such that \eqref{defineetaUD} holds (since it can be set to $0$), and the random variable $\tilde{X}_{t+1,\boldsymbol{\gamma}}$ is well defined. We apply the above argument iteratively for each $t=1$ up to $t=T$. In this way, we conclude that a sufficient condition for the non-increasing sequence $\boldsymbol{\gamma}$ to be feasible for our policy in \Cref{UDpolicy} is that \eqref{UDgammain} holds for each $t$.

Note that the expected utilization of our policy in \Cref{UDpolicy} is  $\sum_{t=1}^{T}\gamma_t\cdot\psi_t$.
The above analysis implies we can focus on solving the following optimization problem to determine the sequence $\boldsymbol{\gamma}$:
\begin{align}
\text{OP}(\boldsymbol{\psi}):=&~\max~~\sum_{t=1}^{T}\gamma_t\cdot\psi_t \label{UDoptimization}\\
&~~\mbox{s.t.}~~~\gamma_{t+1}\leq 1-\gamma_1-\sum_{\tau=1}^{t}\gamma_{\tau}\cdot\psi_{\tau},~~~\forall t=1,\dots,T-1\nonumber\\
&~~~~~~~~~\gamma_{t+1}\leq 1-2\cdot\sum_{\tau=1}^{t}\gamma_{\tau}\cdot \psi_\tau-\gamma_1\cdot \exp(-\frac{2}{\gamma_1}\cdot\sum_{\tau=1}^{t}\gamma_{\tau}\cdot \psi_\tau),~~~\forall t=1,\dots,T-1\nonumber\\
&~~~~~~~~~1\geq\gamma_1\geq\dots\geq\gamma_T\geq0\nonumber,
\end{align}
where $\boldsymbol{\psi}=(\psi_1,\dots,\psi_T)$ and it holds that $\sum_{t=1}^{T}\psi_t\leq 1$.

Our solution to $\text{OP}(\boldsymbol{\psi})$ can be obtained from the following function over the interval $[0,1]$, the value of which is iteratively computed based on an initial value $\gamma_0\in(0,1)$:
\begin{align}
h_{\gamma_0}(0)=&\gamma_0\label{UDfunctiondefine}\\
h_{\gamma_0}(t)=&\min\left\{ \lim_{\tau\rightarrow t-}h_{\gamma_0}(\tau),~ 1-\gamma_0-\int_{\tau=0}^{t}h_{\gamma_0}(\tau)d\tau,\right.\nonumber\\
 &~~~~~~~~\left.1-2\cdot\int_{\tau=0}^{t}h_{\gamma_0}(\tau)d\tau-\gamma_0\cdot\exp(-\frac{2}{\gamma_0}\cdot\int_{\tau=0}^{t}h_{\gamma_0}(\tau)d\tau  )  \right\}\nonumber.
\end{align}
It is easy to see that the function $h_{\gamma_0}(\cdot)$ is non-increasing and non-negative over $[0,1]$ as long as $0<\gamma_0<1$. Thus, the function $h_{\gamma_0}(\cdot)$ specifies a feasible solution to $\text{OP}(\boldsymbol{\psi})$ when each component of $\boldsymbol{\psi}$ is infinitesimally small and $T\rightarrow\infty$, where $h_{\gamma_0}(t)$ corresponds to $\gamma_t$ for each $t\in[0,1]$.

We now show that for arbitrary $\boldsymbol{\psi}$ satisfying $\sum_{t=1}^{T}\psi_t\leq 1$, we can still construct a feasible solution to $\text{OP}(\boldsymbol{\psi})$ based on the function $h_{\gamma_0}(\cdot)$ for each fixed $0<\gamma_0<1$. Specifically, we define a set of indices $0=k_0\leq k_1\leq\dots\leq k_T\leq 1$ such that $\psi_t=k_t-k_{t-1}$ for each $t=1,\dots,T$. Then, we define
\begin{equation}\label{UDfeasiblegamma}
\hat{\gamma}_t=\frac{\int_{\tau=k_{t-1}}^{k_t}h_{\gamma_0}(\tau)d\tau}{k_t-k_{t-1}}=\frac{\int_{\tau=k_{t-1}}^{k_t}h_{\gamma_0}(\tau)d\tau}{\psi_t}
\end{equation}
for each $t=1,\dots,T$. We show in the following lemma that $\{\hat{\gamma}_t\}_{t=1}^T$ is a feasible solution to $\text{OP}(\boldsymbol{\psi})$.
\begin{lemma}\label{UDoconstgammalemma}
For each fixed $0<\gamma_0<1$, let $h_{\gamma_0}(\cdot)$ be the function defined in \eqref{UDfunctiondefine}. Then, for any $\boldsymbol{\psi}$, the solution $\{\hat{\gamma}_t\}_{t=1}^T$ is a feasible solution to $\text{OP}(\boldsymbol{\psi})$, where $\hat{\gamma}_t$ is as defined in \eqref{UDfeasiblegamma} for each $t=1,\dots,T$.
\end{lemma}
The formal proof is relegated to \Cref{UDproofconstlemma}. Note that for each $\boldsymbol{\psi}$ satisfying $\sum_{t=1}^{T}\psi_t\leq 1$, if $\{\hat{\gamma}_t\}_{t=1}^T$ is constructed according to \eqref{UDfeasiblegamma}, then it is easy to see that
\[
\sum_{t=1}^{T}\hat{\gamma}_t\cdot\psi_t=\int_{t=0}^{k_T}h_{\gamma_0}(t)dt
\]
and thus the guarantee of our policy in \Cref{UDpolicy} based on the sequence $\{\hat{\gamma}_t\}_{t=1}^T$ is $\frac{\int_{t=0}^{k_T}h_{\gamma_0}(t)dt}{k_T}$ for some $k_T\in(0,1]$, where $k_T$ depends on the setup $\boldsymbol{\psi}$. Since the function $h_{\gamma_0}(\cdot)$ is non-increasing and non-negative over $[0,1]$ as long as $0<\gamma_0<1$, we know that the worst-case setup corresponds to $k_T=1$, i.e., $\sum_{t=1}^{T}\psi_t=1$. Then, it is enough to focus on solving the following problem:
\[
\max_{0<\gamma_0<1} \int_{t=0}^{1}h_{\gamma_0}(t)dt
\]
to obtain the guarantee of our policy. Numerically, we can show that when $\gamma_0\approx0.3977$, the above optimization problem reaches its maximum, which is $0.3557$. We conclude that the guarantee of our policy is $0.3557$. Note that the guarantee of our policy is developed with respect to the LP upper bound $\text{UP}(\mathcal{H})$, and so it is straightforward to generalize our results to a multi-knapsack setting \citep{stein2020advance} where the size of each query can be knapsack-dependent, as explained in \Cref{sec:multiToSingle}.

Also, we can show that no online algorithm can achieve a better guarantee than $\frac{1-e^{-2}}{2}\approx0.432$ relative to $\text{UP}(\mathcal{H})$, even in the single-knapsack setting. The counterexample can be constructed from a problem setup with $T$ queries, where each query has a deterministic size $\frac{1}{2}+\frac{1}{T}$ and is active with probability $\frac{2}{T}$. The proof is relegated to \Cref{proofUpperunitdensity}.
\begin{proposition}\label{Upperunitdensity}
In the single-knapsack unit-density case, no online algorithm can achieve a better guarantee than $\frac{1-e^{-2}}{2}\approx0.432$ relative to the LP upper bound $\text{UP}(\mathcal{H})$.
\end{proposition}
Our guarantee of 0.3557 relative to $\text{UP}(\mathcal{H})$ in the unit-density case demonstrates the power of using our invariant-based analysis instead of a large/small analysis; we leave the possibility of tightening the guarantee relative to the upper bound of 0.432 as future work.

\section{Concluding Remarks} \label{sec:conc}
In this paper, we derive guarantees for prophet inequalities. There are two settings considered in our paper. One is the $k$-unit setting where the decision maker can accept up to $k$ queries. The other is the knapsack setting where each query consumes a random fraction of the capacity of resource. For both settings, we use OCRS problems to derive the optimal algorithms with tight guarantees with respect to the LP relaxation. Specifically, for the $k$-unit OCRS, we show that the ``$\gamma$-Conservative Magician'' procedure of \cite{alaei2011bayesian} is in fact optimal with the optimal ratio $\sgk$. We prove the optimality with a LP duality approach and derive an ODE formulation to compute the optimal ratio $\sgk$. As a consequence, we improve the best-known guarantee for $k$-unit prophet inequalities for all $k>1$. On the other hand, for the knapsack OCRS, we introduce a new ``best-fit'' procedure
with a tight performance guarantee of $\frac{1}{3+e^{-2}}\approx0.319$, which improves the previously best-known guarantee of 0.2 for online knapsack. We then modify our algorithm and derive a further improved ratio of 0.3557 for the unit-density special case.

In a nutshell, we derive algorithms for the multi-unit and the knapsack settings of prophet inequalities with the optimal guarantees with respect to the LP relaxation, which also enables us to generalize directly extend our results to the multi-resource online assignment problem that enjoys a wider range of applications. We develop new techniques for obtaining tight ratios and we provide new theoretical understandings of prophet inequalities. There are multiple directions to further extend our results. For example, one may consider deriving the tight ratios with respect to the prophet itself instead of the LP relaxation. One may also consider a data-driven setting where instead of assuming the decision makers know the distributions, only a finite number of samples of the distribution of each query are given. We leave these interesting directions for future research.

\ACKNOWLEDGMENT{The authors thank two anonymous reviewers, an associate editor, and Daniel Kuhn, the area editor, for their constructive comments that greatly improved the exposition of the paper. The authors also thank one of the reviewers for the suggestion that led to a new proof of Theorem 1 and Theorem 4. A preliminary version of the paper appeared in SODA 2022.
}

\ \

\bibliographystyle{ormsv080} 
\bibliography{myreferences} 

\clearpage
\begin{APPENDIX}{Proofs of Lemmas, Propositions and Theorems}

\OneAndAHalfSpacedXI

\section{Final Algorithms}\label{sec:FinalAlgorithm}

In this section, we present our final algorithms, which combine the random routing approach described in \Cref{sec:multiToSingle} and the pre-processing algorithms for the corresponding OCRS problem under the multi-unit and the knapsack settings, as well as the unit density special case of the knapsack setting.

The final algorithm for the multi-unit setting is presented below in \Cref{alg:multiunit}. We adopt the multi-resource formulation, where we have $m$ resources and each resource $j$ can serve up to $k_j$ queries. The final algorithm for the knapsack setting is presented below in \Cref{alg:knapsack}. We again adopt the multi-resource formulation and the size of each query over each resource can be an arbitrary fraction of the initial capacity of that resource. Finally, we present our final algorithm for the unit-density special case of the knapsack setting below in \Cref{UDpolicy_final}. The algorithm is presented under the multi-resource formulation where for each query over each resource, the corresponding reward and size are assumed to be equivalent to each other which can take an arbitrary fraction of the initial capacity of the resource. 
\begin{algorithm}
 \caption{Algorithm for the multi-unit setting}
 \begin{algorithmic}[1]
     \State Solve the LP relaxation \eqref{LPuppermultidimension} and and obtain an optimal solution $\{ x^*_{tj}(\mathbf{r}_t, \mathbf{d}_t), \forall t, \forall j, \forall (\mathbf{r}_t, \mathbf{d}_t) \}$.
     \State Define $p_{tj}=\mathbb{E}_{(\tilde{\mathbf{r}}_t, \tilde{\mathbf{d}}_t)\sim F_t}[x^*_{tj}(\tilde{\mathbf{r}}_t, \tilde{\mathbf{d}}_t)]$ for each $j$ and each $t$.
     \State For each $j$, compute the value of $\gamma^*_{k_j}$ following the procedure described in \Cref{sec:worstkunit}.
     \State For each $j$, construct a solution $\{x_{l,t, j}(\gamma^*_{k_j})\}_{\forall l, \forall t}$ as in \Cref{constructprophet} with the input $\theta=\gamma^*_{k_j}$ and $\bm{p}_j=(p_{1j}, \dots, p_{Tj})$.
     \For{$t=1,\dots, T$}
     \State Observe the reward and size realization $(\tilde{\mathbf{r}}_t, \tilde{\mathbf{d}}_t)$ of query $t$.
\State Randomly route query $t$ to a resource $j$ with probability $x^*_{tj}(\tilde{\mathbf{r}}_t, \tilde{\mathbf{d}}_t)$.
\State Denote by $l_j$ the quantity such that $l_j-1$ queries have already been served by resource $j$.
\State \textbf{If} $l_j=1$, \textbf{then}, 
\State ~~~Serve query $t$ with probability $\frac{x_{1,t, j}(\gamma^*_{k_j})}{p_{tj}\cdot (1-\sum_{\tau<t}x_{1,\tau, j}(\gamma^*_{k_j}))}$.
\State \textbf{Else}, 
\State ~~~Serve query $t$ with probability $\frac{x_{l_j,t,j}(\gamma^*_{k_j})}{p_{tj}\cdot \sum_{\tau<t}(x_{l_j-1,\tau,j}(\gamma^*_{k_j})-x_{l_j,\tau}(\gamma^*_{k_j}))}$.
     \EndFor
 \end{algorithmic}
 \label{alg:multiunit}
 \end{algorithm}
 
\begin{algorithm}
\caption{Algorithm for the knapsack setting}
\begin{algorithmic}[1]
\State Solve the LP relaxation \eqref{LPuppermultidimension} and obtain an optimal solution $\{ x^*_{tj}(\mathbf{r}_t, \mathbf{d}_t), \forall t, \forall j, \forall (\mathbf{r}_t, \mathbf{d}_t) \}$.
\State Define $p_{tj}(d_t)=\mathbb{E}_{(\tilde{\mathbf{r}}_t, \tilde{\mathbf{d}}_t)\sim F_t}[\bI(\tilde{d}_{tj}=d_t)\cdot x^*_{tj}(\tilde{\mathbf{r}}_t, \tilde{\mathbf{d}}_t)]$ for each $j$, each $t$ and each $d_t$.
\State For each $j$, obtain the distributions $\{\tilde{X}_{ t, j, \gamma}, \forall t\}$ and the thresholds $\{ \eta_{t, j, \gamma}(d_t), \forall t, \forall d_t \}$ from \Cref{definedistributionrandomsize} with input $\gamma=1/(3+e^{-2})$ and $\{p_{tj}(d_t), \forall t, \forall d_t\}$.
\For{$t=1,\dots,T,$}
\State Observe the reward and size realization $(\tilde{\mathbf{r}}_t, \tilde{\mathbf{d}}_t)$ of query $t$.
\State Randomly route query $t$ to a resource $j$ with probability $x^*_{tj}(\tilde{\mathbf{r}}_t, \tilde{\mathbf{d}}_t)$.
\State Observe the capacity utilization $X_{t-1, j}$ of the routed resource $j$.
\State \textbf{If} $\eta_{t, j, \gamma}(\tilde{d}_{tj})< X_{t-1,j}\leq1-\tilde{d}_{tj}$, \textbf{then},
\State ~~~Serve query $t$ with probability $1$ using resource $j$.
\State \textbf{Else if} $X_{t-1,j}=\eta_{t, j, \gamma}(\tilde{d}_{tj})$, \textbf{then}
\State ~~~serve query $t$ with probability $\frac{\gamma-P(\eta_{t, j, \gamma}(\tilde{d}_{tj})<X_{t-1, j}\leq1-\tilde{d}_{tj})}{P(\eta_{t, j, \gamma}(\tilde{d}_{tj})=X_{t-1, j})}$ using resource $j$.
\State \textbf{Else},
\State ~~~reject query $t$.
\EndFor
\end{algorithmic}
\label{alg:knapsack}
\end{algorithm}

\begin{algorithm}
\caption{Algorithm for the unit-density special case of the knapsack setting}
\begin{algorithmic}[1]
\State Solve the LP relaxation \eqref{LPuppermultidimension} and obtain an optimal solution $\{ x^*_{tj}(\mathbf{r}_t, \mathbf{d}_t), \forall t, \forall j, \forall (\mathbf{r}_t, \mathbf{d}_t) \}$.
\State Define $p_{tj}(d_t)=\mathbb{E}_{(\tilde{\mathbf{r}}_t, \tilde{\mathbf{d}}_t)\sim F_t}[\bI(\tilde{d}_{tj}=d_t)\cdot x^*_{tj}(\tilde{\mathbf{r}}_t, \tilde{\mathbf{d}}_t)]$ for each $j$, each $t$ and each $d_t$.
\State  For each $j$, obtain the sequence $\bm{\gamma}_j=(\gamma_{1j},\dots, \gamma_{Tj})$ from \eqref{UDfeasiblegamma}.
\State For each $j$, obtain the distributions $\{\tilde{X}_{t, j, \boldsymbol{\gamma}_j}, \forall t\}$ and the thresholds $\{ \eta_{t, j,  \boldsymbol{\gamma}_j}(d_t), \forall t, \forall d_t \}$ from \Cref{UDpolicy} with input $\boldsymbol{\gamma}_j$ and $\{p_{tj}(d_t), \forall t, \forall d_t\}$.
\For{$t=1,2,\dots,T$}
\State Observe the reward and size realization $(\tilde{\mathbf{r}}_t, \tilde{\mathbf{d}}_t)$ of query $t$.
\State Randomly route query $t$ to a resource $j$ with probability $x^*_{tj}(\tilde{\mathbf{r}}_t, \tilde{\mathbf{d}}_t)$.
\State Observe the capacity utilization $X_{t-1, j}$ of the routed resource $j$.
\State \textbf{If} $\eta_{t, j, \boldsymbol{\gamma}_j}(\tilde{d}_{tj})< X_{t-1,j}\leq1-\tilde{d}_{tj}$, \textbf{then},
\State ~~~Serve query $t$ with probability $1$ using resource $j$.
\State \textbf{Else if} $X_{t-1,j}=\eta_{t, j, \boldsymbol{\gamma}_j}(\tilde{d}_{tj})$, \textbf{then}
\State ~~~serve query $t$ with probability $\frac{\gamma_{tj}-P(\eta_{t, j, \boldsymbol{\gamma}_j}(\tilde{d}_{tj})<X_{t-1, j}\leq1-\tilde{d}_{tj})}{P(\eta_{t, j, \boldsymbol{\gamma}_j}(\tilde{d}_{tj})=X_{t-1, j})}$ using resource $j$.
\State \textbf{Else},
\State ~~~reject query $t$.
\EndFor
\end{algorithmic}
\label{UDpolicy_final}
\end{algorithm}

\section{Proofs in \Cref{1overkspecialcase}}

\subsection{Proof of \Cref{feasiproposition}}\label{prooffeasiproposition}
We first present the following lemma, which shows that instead of checking whether all the constraints of $\LCkunit$ are satisfied, it is enough to consider only one constraint.
\begin{lemma}\label{feasilemma2}
For any $\theta\in[0,1]$, $\{\theta, x_{l,t}(\theta)\}$ is a feasible solution to $\LCkunit(\bm{p})$ if and only if $\sum_{\tau=1}^{T-1}x_{k,\tau}(\theta)\leq 1-\theta$.
\end{lemma}
The proof is relegated to \Cref{ProofLemma6}. We now prove the condition on $\theta$ such that $\sum_{\tau=1}^{T-1}x_{k,\tau}(\theta)\leq 1-\theta$. Due to \Cref{feasilemma2}, this condition implies the feasibility condition of $\{\theta, x_{l,t}(\theta)\}$. Specifically, we will first show that the term $\sum_{t=1}^{T-1}x_{k,t}(\theta)$ is continuously monotone increasing with $\theta$ in the next lemma, where the formal proof is in \Cref{ProofLemma7}.
\begin{lemma}\label{feasilemma3}
For any $1\leq l\leq k$ and any $1\leq t\leq T$, define $y_{l,t}(\theta)=\sum_{\tau=1}^{t}x_{l,t}(\theta)$. Then $y_{l,t}(\theta)$ is monotone increasing with $\theta$ and $y_{l,t}(\theta)$ is also Lipschitz continuous with $\theta$.
\end{lemma}
We are now ready to prove \Cref{feasiproposition}.
\begin{proof}{Proof of \Cref{feasiproposition}:}
Note that when $\theta=0$, $y_{k,T-1}(0)=\sum_{\tau=1}^{T-1}x_{k,\tau}(0)=0<1-\theta=1$, and when $\theta=1$, $y_{k,T-1}(1)=\sum_{\tau=1}^{T-1}x_{k,\tau}(1)>0=1-\theta$. Further note that $1-\theta$ is continuously strictly decreasing with $\theta$ while \Cref{feasilemma3} shows that $y_{k,T-1}(\theta)=\sum_{\tau=1}^{T-1}x_{k,\tau}(\theta)$ is continuously increasing with $\theta$, there must exist a unique $\theta^*\in[0,1]$ such that $\sum_{\tau=1}^{T-1}x_{k,\tau}(\theta^*)=1-\theta^*$ and for any $\theta\in[0,\theta^*]$, it holds that $\sum_{\tau=1}^{T-1}x_{k,\tau}(\theta)\leq 1-\theta$. Combining the above arguments with \Cref{feasilemma2}, we complete our proof.
\Halmos
\end{proof}

\subsection{Proof of \Cref{feasilemma2}}\label{ProofLemma6}
We first prove that for any $\theta\in[0,1]$, $\{x_{l,t}(\theta)\}$ are non-negative.
\begin{lemma}\label{feasilemma1}
For any $\theta\in[0,1]$, we have $x_{l,t}(\theta)\geq0$ for any $l=1,\dots,k$ and $t=1,\dots,T$.
\end{lemma}
\begin{proof}{Proof:}
We now use induction on $l$ to show that for any $l$, we have that $x_{l,t}(\theta)\geq0$ and $\sum_{v=1}^{l}x_{v,t}(\theta)\leq\theta\cdot \p_t$ for any $t$. Since we focus on a fixed $\theta$, we abbreviate $\theta$ in the expression $x_{l,t}(\theta)$ and substitute $x_{l,t}$ for $x_{l,t}(\theta)$ in the proof.

For $l=1$, from definition, we have that for $1\leq t\leq t_2$, it holds that $x_{1,t}\geq0$ and $\sum_{v=1}^{1}x_{v,t}\leq\theta\cdot \p_t$. We now use induction on $t$ to show that for $t_2+1\leq t\leq T$, we have that $0\leq x_{1,t}\leq \theta\cdot \p_t$. Note that from definition, we have that
\[
x_{1,t_2+1}=\p_{t_2+1}\cdot(1-\sum_{\tau=1}^{t_2}\theta\cdot \p_\tau)<\p_{t_2+1}\cdot\theta
\]
Also, note that $1-\sum_{\tau=1}^{t_2-1}\theta\cdot \p_\tau\geq\theta$ and $\p_{t_2}\leq1$, we have that
\[
1-\sum_{\tau=1}^{t_2}\theta\cdot \p_\tau\geq 1-\sum_{\tau=1}^{t_2-1}\theta\cdot \p_\tau-\theta\geq0
\]
Thus, it holds that
\[
0\leq x_{1,t_2+1}=\p_{t_2+1}\cdot(1-\sum_{\tau=1}^{t_2}\theta\cdot \p_\tau)<\p_{t_2+1}\cdot\theta
\]
Now, suppose for a $t$ such that $t_2+1\leq t\leq T$, we have that $0\leq x_{1,\tau}\leq \theta\cdot \p_\tau$ for any $t_2+1\leq \tau\leq t$. Then we have that
\[
x_{1,t+1}\leq \p_{t+1}\cdot (1-\sum_{\tau=1}^{t}x_{1,\tau})\leq \p_{t+1}\cdot (1-\sum_{\tau=1}^{t_2}x_{1,\tau})=\p_{t+1}\cdot(1-\sum_{\tau=1}^{t_2}\theta\cdot \p_\tau)<\p_{t+1}\cdot\theta
\]
Also, note that $x_{1,t}\geq0$ implies that $1-\sum_{\tau=1}^{t-1}x_{1,\tau}\geq0$, we have that
\[
x_{1,t+1}/\p_{t+1}=1-\sum_{\tau=1}^{t-1}x_{1,\tau}-x_{1,t}=(1-\p_t)\cdot(1-\sum_{\tau=1}^{t-1}x_{1,\tau})\geq0
\]
It holds that $0\leq x_{1,t+1}\leq \p_{t+1}\cdot\theta$. Thus, from induction, for any $t$, we have proved that $0\leq x_{1,t}\leq \p_{t}\cdot\theta$.

Suppose that for a $l$ such that $1\leq l\leq k$, we have that $x_{l,t}\geq0$ and $\sum_{v=1}^{l}x_{v,t}\leq\theta\cdot \p_t$ for any $t$. We now consider the case for $l+1$. From definition, $x_{l+1,t}=0$ when $1\leq t\leq t_{l+1}$ and when $t_{l+1}+1\leq t\leq t_{l+2}$, $x_{l+1,t}=\theta\cdot \p_t-\sum_{v=1}^{l}x_{v,t}$. Thus, for $1\leq t\leq t_{l+2}$, we have proved that $x_{l+1,t}\geq0$ and $\sum_{v=1}^{l+1}x_{v,t}\leq\theta\cdot \p_t$. We now use induction on $t$ for $t>t_{l+2}$. When $t=t_{l+2}+1$, from definition, we have that
\[
x_{l+1,t_{l+2}+1}=\p_{t_{l+2}+1}\cdot \sum_{\tau=1}^{t_{l+2}}(x_{l,\tau}-x_{l+1,\tau})\leq \theta\cdot \p_{t_{l+2}+1}-\sum_{v=1}^{l}x_{v,t_{l+2}+1}\Rightarrow\sum_{v=1}^{l+1}x_{v,t_{l+2}+1}\leq\theta\cdot \p_{t_{l+2}+1}
\]
Also, note that
\[
0\leq x_{l+1,t_{l+2}}=\theta\cdot \p_{t_{l+2}}-\sum_{v=1}^{l}x_{v,t_{l+2}}\leq \p_{t_{l+2}}\cdot\sum_{\tau=1}^{t_{l+2}-1}(x_{l,\tau}-x_{l+1,\tau})
\]
we get that
\[
x_{l+1,t_{l+2}+1}/\p_{t_{l+2}+1}=\sum_{\tau=1}^{t_{l+2}}(x_{l,\tau}-x_{l+1,\tau})\geq\sum_{\tau=1}^{t_{l+2}-1}(x_{l,\tau}-x_{l+1,\tau})-x_{l+1,t_{l+2}}\geq(1-\p_{t_{l+2}})\cdot \sum_{\tau=1}^{t_{l+2}-1}(x_{l,\tau}-x_{l+1,\tau})
\]
Thus, we proved that $0\leq x_{l+1,t_{l+2}+1}$ and $\sum_{v=1}^{l+1}x_{v,t_{l+2}+1}\leq\theta\cdot \p_{t_{l+2}+1}$. Now suppose that for a $t$ such that $t_{l+2}+1\leq t\leq T$, it holds that $0\leq x_{l+1,t}$ and $\sum_{v=1}^{l+1}x_{v,t}\leq\theta\cdot \p_t$. Then we have that
\[
\sum_{v=1}^{l+1}x_{v,t+1}/\p_{t+1}=1-\sum_{\tau=1}^{t}x_{l+1,\tau}\leq 1-\sum_{\tau=1}^{t-1}x_{l+1,\tau}=\sum_{v=1}^{l+1}x_{v,t}/\p_{t}\leq\theta
\]
Also, note that $0\leq x_{l+1,t}=\p_t\cdot\sum_{\tau=1}^{t-1}(x_{l,\tau}-x_{l+1,\tau})$, we have that
\[
x_{l+1,t+1}/\p_{t+1}=\sum_{\tau=1}^{t}(x_{l,\tau}-x_{l+1,\tau})\geq\sum_{\tau=1}^{t-1}(x_{l,\tau}-x_{l+1,\tau})-x_{l+1,\tau}=(1-\p_t)\cdot\sum_{\tau=1}^{t-1}(x_{l,\tau}-x_{l+1,\tau})\geq0
\]
Thus, we have proved that $0\leq x_{l+1,t+1}$ and $\sum_{v=1}^{l+1}x_{v,t+1}\leq\theta\cdot p_{t+1}$. From the induction on $t$, we can conclude that for any $1\leq t\leq T$, it holds that $0\leq x_{l+1,t}$ and $\sum_{v=1}^{l+1}x_{v,t}\leq\theta\cdot \p_t$. Again, from the induction on $l$, we can conclude that for any $1\leq l\leq k$ and any $1\leq t\leq T$, it holds that $0\leq x_{l,t}$ and $\sum_{v=1}^{l}x_{v,t}\leq\theta\cdot \p_t$, which completes our proof.
\Halmos
\end{proof}
~\\
Now we are ready to prove Lemma \ref{feasilemma2}.
\begin{proof}{Proof of \Cref{feasilemma2}:}
When $\{x_{l,t}(\theta)\}$ is feasible to $\LCkunit(\bm{p})$ in \eqref{dualdual}, we get from constraint \eqref{ddconstraint6} and \eqref{ddconstraint7} that
\[
x_{1,T}(\theta)\leq \p_T\cdot(1-\sum_{t=1}^{T-1}x_{1,t}(\theta))\text{~~and~~}x_{l,T}(\theta)\leq\p_T\cdot\sum_{t=1}^{T-1}(x_{l-1,t}(\theta)-x_{l,t}(\theta))~~\forall l=2,\dots,k
\]
Summing up the above inequalities, we get
\[
\sum_{l=1}^{k}x_{l,T}(\theta)\leq\p_T\cdot(1-\sum_{t=1}^{T-1}x_{k,t}(\theta))
\]
Further note that by definition, we have $\sum_{l=1}^{k}x_{l,T}(\theta)=\theta\cdot\p_T$. Thus, we show that $\{x_{l,t}(\theta)\}$ is feasible implies that $\sum_{t=1}^{T-1}x_{k,t}(\theta)\leq1-\theta$.

Now we prove the reverse direction. Note that from the definition of $\{x_{l,t}(\theta)\}$, we have that $x_{l,t}(\theta)\leq \p_t\cdot\sum_{\tau=1}^{t-1}(x_{l-1,\tau}(\theta)-x_{l,\tau}(\theta))$ holds for any $1\leq l\leq k-1$ and any $1\leq t\leq T$, where we set $\sum_{\tau=1}^{t-1}x_{0,\tau}(\theta)=1$ for any $t$ for simplicity. Also, $\{x_{l,t}(\theta)\}$ are nonnegative as shown by Lemma \ref{feasilemma1}. Thus, we have that
\[
\{x_{l,t}(\theta)\} \text{~is~feasible~}\Leftrightarrow x_{k,t}(\theta)\leq \p_t\cdot\sum_{\tau=1}^{t-1}(x_{k-1,\tau}(\theta)-x_{k,\tau}(\theta))\text{~holds~for~any~}t_k+1\leq t\leq T
\]
Moreover, note that from definition, for $t_k+1\leq t\leq T$, we have that $x_{l,t}(\theta)=\p_t\cdot\sum_{\tau=1}^{t-1}(x_{l-1,\tau}(\theta)-x_{l,\tau}(\theta))$ when $1\leq l\leq k-1$. Thus, for $t_k+1\leq t\leq T$, we have that
\[
x_{k,t}(\theta)\leq \p_t\cdot\sum_{\tau=1}^{t-1}(x_{k-1,\tau}(\theta)-x_{k,\tau}(\theta))\Leftrightarrow \theta=\sum_{v=1}^{k}x_{v,t}/\p_t\leq 1-\sum_{\tau=1}^{t-1}x_{k,\tau}(\theta)
\]
From the nonnegativity of $\{x_{l,t}(\theta)\}$, we know that $\sum_{\tau=1}^{t-1}x_{k,\tau}(\theta)$ is monotone increasing with $t$. Thus, it holds that
\[
\{x_{l,t}(\theta)\} \text{~is~feasible~}\Leftrightarrow \theta\leq 1-\sum_{t=1}^{T-1}x_{k,t}(\theta)
\]
which completes our proof.
\Halmos
\end{proof}
\subsection{Proof of Lemma \ref{feasilemma3}}\label{ProofLemma7}
\begin{proof}{Proof:}
For any fixed $\theta\in[0,1]$ and any fixed $\Delta\geq0$ such that $\theta+\Delta\in[0,1]$, we compare between $\{x_{l,t}(\theta)\}$ and $\{x_{l,t}(\theta+\Delta)\}$. Since we consider for a fixed $\theta$ and $\Delta$, for notation brevity, we will omit $\theta$ and $\Delta$ by substituting $\{x_{l,t}\}$ for $\{x_{l,t}(\theta)\}$ and substituting $\{x'_{l,t}\}$ for $\{x_{l,t}(\theta+\Delta)\}$. Respectively, we denote $y_{l,t}=\sum_{\tau=1}^{t-1}x_{l,\tau}$ and $y'_{l,t}=\sum_{\tau=1}^{t-1}x'_{l,\tau}$. Also, we denote $\{t_l\}$ to be the time indexes associated with $\{x_{l,t}\}$ in the definition of $\{x_{l,t}\}$ and $\{t'_l\}$ to be the time indexes associated with $\{x'_{l,t}\}$. We will use induction to show that for each $l$, we have that $y_{l,t}\leq y'_{l,t}$ and $\sum_{v=1}^{l}y'_{v,t}\leq\sum_{v=1}^{l}y_{v,t}+\Delta\cdot\sum_{\tau=1}^{t}\p_\tau$ hold for each $t$.

For the case $l=1$, obviously we have that $t_2'\leq t_2$. When $1\leq t\leq t_2'$, from definition, it holds that $y_{1,t}\leq y'_{1,t}\leq y_{1,t}+\Delta\cdot\sum_{\tau=1}^{t}\p_\tau$. We now use induction on $t$ for $t'_2+1\leq t\leq t_2$. When $t=t'_2+1$, note that
\[
x'_{1,t'_2+1}=\p_{t'_2+1}\cdot(1-y'_{1,t'_2})\leq(\theta+\Delta)\cdot \p_{t'_2+1}\text{~~and~~}x_{1,t'_2+1}=\theta\cdot \p_{t'_2+1}\leq \p_{t'_2+1}\cdot(1-y_{1,t'_2})
\]
we have that
\[
y_{1,t'_2+1}=y_{1,t'_2}+x_{1,t'_2+1}\leq \p_{t'_2+1}+(1-p_{t'_2+1})\cdot y_{1,t'_2}\leq  \p_{t'_2+1}+(1-\p_{t'_2+1})\cdot y'_{1,t'_2}=y'_{1,t'_2+1}
\]
and
\[
y'_{1,t'_2+1}=y'_{1,t'_2}+x'_{1,t'_2+1}\leq y_{1,t'_2}+\Delta\cdot\sum_{t=1}^{t'_2}\p_t+(\theta+\Delta)\cdot \p_{t'_2+1}=y_{1,t'_2+1}+\Delta\cdot\sum_{t=1}^{t'_2+1}\p_t
\]
Now suppose for a fixed $t$ satisfying $t'_2+1\leq t\leq t_2-1$, it holds $y_{1,t}\leq y'_{1,t}\leq y_{1,t}+\Delta\cdot\sum_{\tau=1}^{t}p_\tau$. From definition, note that
\[
x'_{1,t+1}=\p_{t+1}\cdot(1-y'_{1,t})\leq(\theta+\Delta)\cdot \p_{t+1}\text{~~and~~}x_{1,t+1}=\theta\cdot \p_{t+1}\leq \p_{t+1}\cdot(1-y_{1,t})
\]
we have
\[
y_{1,t+1}=y_{1,t}+x_{1,t+1}\leq \p_{t+1}+(1-\p_{t+1})\cdot y_{1,t}\leq \p_{t+1}+(1-\p_{t+1})\cdot y'_{1,t}=y'_{1,t+1}
\]
and
\[
y'_{1,t+1}=y'_{1,t}+x'_{1,t+1}\leq y_{1,t}+\Delta\cdot\sum_{\tau=1}^{t}\p_\tau+(\theta+\Delta)\cdot \p_{t+1}=y_{1,t+1}+\Delta\cdot\sum_{\tau=1}^{t+1}\p_\tau
\]
Thus, from induction on $t$, we conclude that $y_{1,t}\leq y'_{1,t}\leq y_{1,t}+\Delta\cdot\sum_{\tau=1}^{t}\p_\tau$ holds for any $t'_2+1\leq t\leq t_2$. Finally, when $t\geq t_2+1$, note that
\[
y_{1,t}=\p_{t}+(1-\p_t)\cdot y_{1,t-1}\text{~~and~~}y'_{1,t}=\p_{t}+(1-\p_t)\cdot y'_{1,t-1}
\]
which implies that
\[
y'_{1,t}-y_{1,t}=(1-\p_t)\cdot(y'_{1,t-1}-y_{1,t-1})=\dots=(y'_{1,t_2}-y_{1,t_2})\cdot \prod_{\tau=t_2+1}^{t}(1-\p_\tau)
\]
Thus, we prove that for any $1\leq t\leq T$, it holds that $y_{1,t}\leq y'_{1,t}\leq y_{1,t}+\Delta\cdot\sum_{\tau=1}^{t}\p_\tau$.

Suppose that for a fixed $1\leq l\leq k$, $y_{l,t}\leq y'_{l,t}$ and $\sum_{v=1}^{l}y'_{v,t}\leq\sum_{v=1}^{l}y_{v,t}+\Delta\cdot\sum_{\tau=1}^{t}\p_\tau$ hold for each $t$. We now consider the case for $l+1$. When $t\leq\min\{t_{l+2},t'_{l+2}\}$, from definition, we have that
\[
\sum_{v=1}^{l+1}y'_{v,t}=(\theta+\Delta)\cdot\sum_{\tau=1}^{t}\p_\tau ~~\text{and}~~\sum_{v=1}^{l+1}y_{v,t}=\theta\cdot\sum_{\tau=1}^{t}\p_\tau
\]
which implies that $\sum_{v=1}^{l+1}y'_{v,t}\leq \sum_{v=1}^{l+1}y_{v,t}+\Delta\cdot\sum_{\tau=1}^{t}\p_\tau$. Also, we have
\[
y'_{l+1,t}-y_{l+1,t}=\Delta\cdot\sum_{\tau=1}^{t}\p_\tau-\left(\sum_{v=1}^{l}y'_{v,t}-\sum_{v=1}^{l}y_{v,t} \right)\geq0
\]
where the last inequality holds from induction condition. Thus, we prove that $y_{l+1,t}\leq y'_{l+1,t}$ and $\sum_{v=1}^{l+1}y'_{v,t}\leq\sum_{v=1}^{l+1}y_{v,t}+\Delta\cdot\sum_{\tau=1}^{t}\p_\tau$ hold for each $1\leq t\leq\min\{t_{l+2},t'_{l+2}\}$. Moreover, note that $t_{l+2}$ is defined as the first time that $\theta>1-y_{l+1,t_{l+2}}$ while $t'_{l+2}$ is defined as the first time that $\theta+\Delta>1-y'_{l+1,t'_{l+2}}$. Since $y'_{l+1,t}\geq y_{l+1,t}$ when $t\leq\min\{t_{l+2}, t'_{l+2}\}$, we must have $t'_{l+2}\leq t_{l+2}$. Then we use induction on $t$ for $t'_{l+2}+1\leq t\leq t_{l+2}$. When $t=t'_{l+2}+1\leq t_{l+2}$, from definition, we have
\[
x'_{l+1,t'_{l+2}+1}=\p_{t'_{l+2}+1}\cdot( y'_{l,t'_{l+2}}-y'_{l+1,t'_{l+2}} )\Rightarrow y'_{l+1,t'_{l+2}+1}=\p_{t'_{l+2}+1}\cdot y'_{l,t'_{l+2}}+(1-\p_{t'_{l+2}+1})\cdot y'_{l+1,t'_{l+2}}
\]
and
\[
x_{l+1,t'_{l+2}+1}\leq \p_{t'_{l+2}+1}\cdot( y_{l,t'_{l+2}}-y_{l+1,t'_{l+2}} ) \Rightarrow y_{l+1,t'_{l+2}+1}\leq\p_{t'_{l+2}+1}\cdot y_{l,t'_{l+2}}+(1-\p_{t'_{l+2}+1})\cdot y_{l+1,t'_{l+2}}
\]
Note that $y'_{l,t'_{l+2}}\geq y_{l,t'_{l+2}}$ and $y'_{l+1,t'_{l+2}}\geq y_{l+1,t'_{l+2}}$, we get $y'_{l+1,t'_{l+2}+1}\geq y_{l+1,t'_{l+2}+1}$. Moreover, note that from the definition of $t'_{l+2}$, we have
\[
\sum_{v=1}^{l+1}x'_{v,t'_{l+2}+1}\leq p_{t'_{l+2}+1}\cdot(\theta+\Delta)=\sum_{v=1}^{l+1}x_{v,t'_{l+2}+1}+\Delta\cdot \p_{t'_{l+2}+1}
\]
which implies that
\[\begin{aligned}
\sum_{v=1}^{l+1}y'_{v,t'_{l+2}+1}&=\sum_{v=1}^{l+1}y'_{v,t'_{l+2}}+\sum_{v=1}^{l+1}x'_{v,t'_{l+2}+1}\leq \sum_{v=1}^{l+1}y_{v,t'_{l+2}}+\Delta\cdot\sum_{j=1}^{t'_{l+2}}\p_j+\sum_{v=1}^{l+1}x_{v,t'_{l+2}+1}+\Delta\cdot \p_{t'_{l+2}+1}\\
&=\sum_{v=1}^{l+1}y_{v,t'_{l+2}+1}+\Delta\cdot\sum_{j=1}^{t'_{l+2}+1}\p_j
\end{aligned}\]
Then suppose for a fixed $t$ satisfying $t'_{l+2}+1\leq t\leq t_{l+2}-1$, it holds that $y_{l+1,t}\leq y'_{l+1,t}$ and $\sum_{v=1}^{l+1}y'_{v,t}\leq\sum_{v=1}^{l+1}y_{v,t}+\Delta\cdot\sum_{\tau=1}^{t}\p_\tau$. From definition, we have
\[
x'_{l+1,t+1}=\p_{t+1}\cdot( y'_{l,t}-y'_{l+1,t} )\Rightarrow y'_{l+1,t+1}=\p_{t+1}\cdot y'_{l,t}+(1-\p_{t+1})\cdot y'_{l+1,t}
\]
and
\[
x_{l+1,t+1}\leq \p_{t+1}\cdot( y_{l,t}-y_{l+1,t} )\Rightarrow y_{l+1,t+1}\leq \p_{t+1}\cdot y_{l,t}+(1-\p_{t+1})\cdot y_{l+1,t}
\]
Note that $y'_{l,t}\geq y_{l,t}$ and $y'_{l+1,t}\geq y_{l+1,t}$, we have $y'_{l+1,t+1}\geq y_{l+1,t+1}$. Also, from the definition of $t'_{l+2}$, we have
\[
\sum_{v=1}^{l+1}x'_{v,t+1}\leq \p_{t+1}\cdot(\theta+\Delta)=\sum_{v=1}^{l+1}x_{v,t+1}+\Delta\cdot \p_{t+1}
\]
which implies that
\[\begin{aligned}
\sum_{v=1}^{l+1}y'_{v,t+1}&=\sum_{v=1}^{l+1}y'_{v,t}+\sum_{v=1}^{l+1}x'_{v,t+1}\leq \sum_{v=1}^{l+1}y_{v,t}+\Delta\cdot\sum_{\tau=1}^{t}\p_\tau+\sum_{v=1}^{l+1}x_{v,t+1}+\Delta\cdot \p_{t+1}\\
&=\sum_{v=1}^{l+1}y_{v,t+1}+\Delta\cdot\sum_{\tau=1}^{t+1}\p_\tau
\end{aligned}\]
Thus, from induction on $t$, we prove that $y_{l+1,t}\leq y'_{l+1,t}$ and $\sum_{v=1}^{l+1}y'_{v,t}\leq\sum_{v=1}^{l+1}y_{v,t}+\Delta\cdot\sum_{\tau=1}^{t}\p_\tau$ hold for any $t'_{l+2}+1\leq t\leq t_{l+2}$. Finally, when $t\geq t_{l+2}+1$, note that
\[
y_{l+1,t}=\p_t\cdot y_{l,t-1}+(1-\p_t)\cdot y_{l+1,t-1}\text{~~and~~}y'_{l+1,t}=\p_t\cdot y'_{l,t-1}+(1-\p_t)\cdot y'_{l+1,t-1}
\]
which implies that
\[
y'_{l+1,t}-y_{l+1,t}=\p_t\cdot (y'_{l,t-1}-y_{l,t-1})+(1-\p_t)\cdot (y'_{l+1,t-1}-y_{l+1,t-1})
\]
It is direct to show inductively on $t$ such that $y_{l+1,t}\leq y'_{l+1,t}$ and $\sum_{v=1}^{l+1}y'_{v,t}\leq\sum_{v=1}^{l+1}y_{v,t}+\Delta\cdot\sum_{\tau=1}^{t}\p_\tau$ hold for any $t\geq t_{l+2}+1$.

Thus, we have proved that for any $1\leq t\leq T$, we have $y_{l+1,t}\leq y'_{l+1,t}$ and $\sum_{v=1}^{l+1}y'_{v,t}\leq\sum_{v=1}^{l}y_{v,t}+\Delta\cdot\sum_{\tau=1}^{t}p_\tau$. By the induction on $l$, we finally prove that for any $1\leq l\leq k$, $y_{l,t}\leq y'_{l,t}$ and $\sum_{v=1}^{l}y'_{v,t}\leq\sum_{v=1}^{l}y_{v,t}+\Delta\cdot\sum_{\tau=1}^{t}\p_\tau$ hold for any $1\leq t\leq T$. In this way, we prove that $y_{l,t}(\theta)$ is monotone increasing with $\theta$ for any $l, t$. Moreover, note that since $\sum_{\tau=1}^{T}\p_\tau\leq k$, we have that $y_{l,t}(\theta+\Delta)\leq y_{l,t}(\theta)+k\cdot\Delta$ hold for any $\theta, \Delta$ and any $l, t$. Thus, $y_{l,t}(\theta)$ is a continuous function on $\theta$, which completes our proof.
\Halmos
\end{proof}

\subsection{Proof of \Cref{Prophetoptimaltheorem}}\label{newProoftheorem3}
\begin{proof}{Proof:}
Given \Cref{feasiproposition}, in order to prove Theorem \ref{Prophetoptimaltheorem}, it is enough for us to construct a feasible solution $\{\beta^*_{l,t}, \xi^*_t\}$ to $\DCkunit(\bm{p})$ in \eqref{dual} such that the primal-dual pair $\{\theta^*, x_{l,t}(\theta^*)\}$ and $\{\beta^*_{l,t}, \xi^*_t\}$ satisfies the complementary slackness conditions. Specifically, we will construct a feasible solution $\{\beta^*_{l,t}, \xi^*_t\}$ to $\DCkunit(\bm{p})$ satisfying the following conditions:
\begin{align}
&\beta^*_{1,t}\cdot\left(x_{1,t}(\theta^*)- \p_t\cdot (1-\sum_{\tau<t}x_{1,\tau}(\theta^*))\right)=0, ~~\forall t=1,\dots,T \label{new_complementaryslackness}\\
&\beta^*_{l,t}\cdot \left(x_{l,t}(\theta^*)- \p_t\cdot \sum_{\tau<t}(x_{l-1,\tau}(\theta^*)-x_{l,\tau}(\theta^*))\right)=0,~~\forall t=1,\dots, T, \forall l=2,\dots,k \nonumber\\
&x_{l,t}(\theta^*)\cdot \left(\beta^*_{l,t}+\sum_{\tau>t}\p_\tau\cdot (\beta^*_{l,\tau}-\beta^*_{l+1,\tau})-\xi^*_t\right)=0, ~~\forall t=1,\dots, T, \forall l=2,\dots,k \nonumber\\
&x_{k,t}(\theta^*)\cdot\left( \beta^*_{k,t}+\sum_{\tau>t}\p_\tau\cdot\beta^*_{k,\tau}-\xi^*_t\right)=0, ~~\forall t=1,\dots,T\nonumber
\end{align}
Note that from definitions, $\{x_{l,t}(\theta^*)\}$ satisfies the following conditions:
\[\begin{aligned}
&x_{l,t}(\theta^*)=0\leq \p_t\cdot \sum_{\tau=1}^{t-1}(x_{l-1,\tau}(\theta^*)-x_{l,\tau}(\theta^*)),~~\forall t\leq t_l\\
&x_{l,t}(\theta^*)=\theta^*\cdot \p_t-\sum_{v=1}^{l-1}x_{v,t}(\theta^*)\leq \p_t\cdot \sum_{\tau=1}^{t-1}(x_{l-1,\tau}(\theta^*)-x_{l,\tau}(\theta^*))~~~\text{for~}t_l+1\leq t\leq t_{l+1}
\end{aligned}\]
where $\{t_l\}$ are the time indexes associated with the definition of $\{x_{l,t}(\theta^*)\}$ and we define $t_1=0$, $t_{k+1}=T-1$. Note that we can set $t_{k+1}=T-1$ because we focus on the solution $\{x_{l,t}(\theta^*)\}$. If we consider other solution $\{x_{l,t}(\theta)\}$ with $\theta\neq\theta^*$, then we cannot have $t_{k+1}=T-1$. Having $\theta=\theta^*$ in the solution $\{x_{l,t}(\theta)\}$ is the only way to make $t_{k+1}=T-1$ consistent with \Cref{constructprophet}.

For simplicity, we also denote $\sum_{\tau=1}^{t-1}x_{0,\tau}(\theta^*)=1$ for any $t$. Thus, in order for $\{\beta^*_{l,t}, \xi^*_t\}$ to satisfy the conditions in \eqref{new_complementaryslackness}, it is enough for $\{\beta^*_{l,t}, \xi^*_t\}$ to be feasible to $\DCkunit(\bm{p})$ and satisfy the following conditions:
\begin{align}
&\beta^*_{l,t}=0~~~\text{for~}t\leq t_{l+1}\label{new_condition1}\\
&\beta^*_{l,t}+\sum_{\tau=t+1}^T p_\tau\cdot(\beta^*_{l,\tau}-\beta^*_{l+1,\tau} )=\xi^*_t~~~\text{for~}t\geq t_l+1 \label{new_condition2}
\end{align}
where we denote $\beta^*_{k+1,t}=0$ for notation simplicity. We now show the construction of the solution $\{\beta^*_{l,t}, \xi^*_t\}$ to $\DCkunit(\bm{p})$.

We first define $\xi^*_T=R$, for a constant $R>0$ that will be specified later. We also define $\beta^*_{l,T}=R$ for any $l=1,\dots,k$. Then, inductively for $t=T-1, T-2, \dots, 1$, we follow the two steps below to specify the value of $\xi^*_t$ and $\beta^*_{j,t}$ for any $l=1,\dots,k$.
\begin{enumerate}
    \item We fix $l'$ such that $t_{l'}+1\leq t\leq t_{l'+1}$, and we define
    \begin{equation}\label{eqn:23062501}
    \xi^*_t=\sum_{\tau=t+1}^T p_{\tau}\cdot (\beta^*_{l',\tau}-\beta^*_{l'+1, \tau}).
    \end{equation}
    \item For each $l=1,\dots,k$, we define
    \begin{equation}\label{eqn:23062502}
    \beta^*_{l,t}=\max\left\{ 0, \xi^*_t-\sum_{\tau=t+1}^T p_{\tau}\cdot (\beta^*_{l,\tau}-\beta^*_{l+1, \tau}) \right\}.
    \end{equation}
\end{enumerate}
Finally, the constant $R$ is selected such that $\sum_{t=1}^T p_t\cdot\xi^*_t=1$. In what follows, we show the construction of $\{\beta^*_{l,t}, \xi^*_t\}$ above is feasible to $\DCkunit(\bm{p})$ and satisfy the requirements \eqref{new_condition1} and \eqref{new_condition2}. Our proof would rely on the following property of $\{\beta^*_{l,t}, \xi^*_t\}$, which we prove at the end of this proof.
\begin{claim}\label{claim:23062501}
Let $\{\beta^*_{l,t}, \xi^*_t\}$ be constructed in \eqref{eqn:23062501} and \eqref{eqn:23062502}. Then, it holds that 
\begin{equation}\label{eqn:23062503}
\sum_{\tau=t}^T p_{\tau}\cdot(\beta^*_{l,\tau}-\beta^*_{l+1,\tau})\leq \sum_{\tau=t}^T p_{\tau}\cdot(\beta^*_{l+1,\tau}-\beta^*_{l+2,\tau})
\end{equation}
for any $t=1,\dots,T$ and any $l=1,\dots, k-1$. Moreover, $\beta^*_{l,t}\geq \beta^*_{l+1,t}$ for any $t=1,\dots,T$ and any $l=1,\dots, k$.
\end{claim}

From \Cref{claim:23062501} and the construction \eqref{eqn:23062501}, we know that $\xi^*_t$ is non-negative for any $t=1,\dots,T$. Therefore, from the construction \eqref{eqn:23062502}, we know that $\{\beta^*_{l,t}, \xi^*_t\}$ above is feasible to $\DCkunit(\bm{p})$. It only remains to show that $\{\beta^*_{l,t}, \xi^*_t\}$ satisfy the requirements \eqref{new_condition1} and \eqref{new_condition2}.

For any $t=1,\dots, T$, we fix the index $l'$ such that $t_{l'}+1\leq t\leq t_{l'+1}$. Then for any index $l\geq l'$, which is equivalent to $t\leq t_{l+1}$, we have
\[
\xi^*_t=\sum_{\tau=t+1}^T p_{\tau}\cdot (\beta^*_{l',\tau}-\beta^*_{l'+1, \tau})\leq \sum_{\tau=t+1}^T p_{\tau}\cdot (\beta^*_{l,\tau}-\beta^*_{l+1, \tau})
\]
which follows from \eqref{eqn:23062503} in \Cref{claim:23062501}. Thus, we know that $\beta^*_{l,t}=0$ for any $l\geq l'$ such that $t\leq t_{l+1}$, which shows $\{\beta^*_{l,t}, \xi^*_t\}$ satisfy the requirements \eqref{new_condition1}.

For any $t=1,\dots, T$, we fix the index $l'$ such that $t_{l'}+1\leq t\leq t_{l'+1}$. Then for any index $l\leq l'$, which is equivalent to $t\geq t_{l}+1$, we have 
\[
\xi^*_t=\sum_{\tau=t+1}^T p_{\tau}\cdot (\beta^*_{l',\tau}-\beta^*_{l'+1, \tau})\geq \sum_{\tau=t+1}^T p_{\tau}\cdot (\beta^*_{l,\tau}-\beta^*_{l+1, \tau})
\]
which follows from \eqref{eqn:23062503} in \Cref{claim:23062501}. Thus, we know that 
\[
\beta^*_{l,t}=\xi^*_t-\sum_{\tau=t+1}^T p_{\tau}\cdot (\beta^*_{l,\tau}-\beta^*_{l+1, \tau})
\]
for $t\geq t_l+1$, which shows that $\{\beta^*_{l,t}, \xi^*_t\}$ satisfy the requirements \eqref{new_condition2}. Our proof of the theorem is thus completed.
\Halmos
\end{proof}

\begin{proof}{Proof of \Cref{claim:23062501}:}
We prove \eqref{eqn:23062503} by induction. Clearly, for $t=T$, since $\beta^*_{l,T}=R>0$ for any $l=1,\dots,k$ and $\beta^*_{k+1,T}=0$, \eqref{eqn:23062503} holds. We now suppose \eqref{eqn:23062503} holds for $t+1$ and we consider the situation for $t$. 

We fix the index $l'$ such that $t_{l'}+1\leq t\leq t_{l'+1}$. We then consider the following two scenarios.

Scenario (i) when $l\geq l'$. Then, from the induction hypothesis, we know that
\[
\xi^*_t=\sum_{\tau=t+1}^T p_{\tau}\cdot (\beta^*_{l',\tau}-\beta^*_{l'+1, \tau}) \leq \sum_{\tau=t+1}^T p_{\tau}\cdot (\beta^*_{l'',\tau}-\beta^*_{l''+1, \tau})
\]
for $l''=l, l+1, l+2$. We thus have $\beta^*_{l,t}=\beta^*_{l+1,t}=\beta^*_{l+2,t}=0$ and we directly prove \eqref{eqn:23062503} from the induction hypothesis. 

Scenario (ii) when $l< l'$. From the induction hypothesis, it is clear to see that
\[
\xi^*_t=\sum_{\tau=t+1}^T p_{\tau}\cdot (\beta^*_{l',\tau}-\beta^*_{l'+1, \tau}) \geq \sum_{\tau=t+1}^T p_{\tau}\cdot (\beta^*_{l,\tau}-\beta^*_{l+1, \tau}).
\]
Then, we have
\[
\beta^*_{l,t}=\xi^*_t-\sum_{\tau=t+1}^T p_{\tau}\cdot (\beta^*_{l,\tau}-\beta^*_{l+1, \tau}).
\]
On the other hand, from the construction \eqref{eqn:23062502}, we know
\[
-\beta^*_{l+1,t}\leq -\xi^*_t+\sum_{\tau=t+1}^T p_{\tau}\cdot (\beta^*_{l+1,\tau}-\beta^*_{l+2, \tau}).
\]
Therefore, it holds that
\[\begin{aligned}
\sum_{\tau=t}^T p_{\tau}\cdot(\beta^*_{l,\tau}-\beta^*_{l+1,\tau})&=\sum_{\tau=t+1}^T p_{\tau}\cdot(\beta^*_{l,\tau}-\beta^*_{l+1,\tau})+ p_{t}\cdot(\beta^*_{l,t}-\beta^*_{l+1,t})\\
&\leq (1-p_t)\cdot \sum_{\tau=t+1}^T p_{\tau}\cdot(\beta^*_{l,\tau}-\beta^*_{l+1,\tau})+p_t\cdot \sum_{\tau=t+1}^T p_{\tau}\cdot (\beta^*_{l+1,\tau}-\beta^*_{l+2, \tau}).
\end{aligned}\]
From the above inequality, we have
\[\begin{aligned}
&\sum_{\tau=t}^T p_{\tau}\cdot(\beta^*_{l+1,\tau}-\beta^*_{l+2,\tau})-\sum_{\tau=t}^T p_{\tau}\cdot(\beta^*_{l,\tau}-\beta^*_{l+1,\tau})\\
\geq& \sum_{\tau=t}^T p_{\tau}\cdot(\beta^*_{l+1,\tau}-\beta^*_{l+2,\tau})-(1-p_t)\cdot \sum_{\tau=t+1}^T p_{\tau}\cdot(\beta^*_{l,\tau}-\beta^*_{l+1,\tau})-p_t\cdot \sum_{\tau=t+1}^T p_{\tau}\cdot (\beta^*_{l+1,\tau}-\beta^*_{l+2, \tau})\\
=& ~p_t\cdot (\beta^*_{l+1,t}-\beta^*_{l+2,t})+(1-p_t)\cdot \left( \sum_{\tau=t+1}^T p_{\tau}\cdot (\beta^*_{l+1,\tau}-\beta^*_{l+2, \tau})-\sum_{\tau=t+1}^T p_{\tau}\cdot (\beta^*_{l,\tau}-\beta^*_{l+1, \tau})
 \right)\\
\geq&~p_t\cdot (\beta^*_{l+1,t}-\beta^*_{l+2,t})
\end{aligned}\]
where the last inequality follows from $p_t\leq1$ and the induction hypothesis. Therefore, it only remains to show that $\beta^*_{l+1,t}\geq\beta^*_{l+2,t}$ under the induction hypothesis, which would prove our whole claim. From the induction hypothesis, we clearly have
\[
\xi^*_t-\sum_{\tau=t+1}^T p_{\tau}\cdot (\beta^*_{l+1,\tau}-\beta^*_{l+2, \tau}) \geq \xi^*_t-\sum_{\tau=t+1}^T p_{\tau}\cdot (\beta^*_{l+2,\tau}-\beta^*_{l+3, \tau})
\]
which implies that
\begin{equation}\label{eqn:23062504}
\beta^*_{l+1,t}=\max\left\{ 0, \xi^*_t-\sum_{\tau=t+1}^T p_{\tau}\cdot (\beta^*_{l+1,\tau}-\beta^*_{l+2, \tau}) \right\} \geq \beta^*_{l+2,t}=\max\left\{ 0, \xi^*_t-\sum_{\tau=t+1}^T p_{\tau}\cdot (\beta^*_{l+2,\tau}-\beta^*_{l+3, \tau}) \right\}.
\end{equation}
Therefore, we know that 
\[
\sum_{\tau=t}^T p_{\tau}\cdot(\beta^*_{l+1,\tau}-\beta^*_{l+2,\tau})-\sum_{\tau=t}^T p_{\tau}\cdot(\beta^*_{l,\tau}-\beta^*_{l+1,\tau})
\geq p_t\cdot (\beta^*_{l+1,t}-\beta^*_{l+2,t})\geq0
\]
which completes our induction. Thus, we prove \eqref{eqn:23062503} for any $t=1,\dots, T$ and any $l=1,\dots,k-1$. Note that following the step in \eqref{eqn:23062504}, we can directly verify that $\beta^*_{l,t}\geq\beta^*_{l+1,t}$ given \eqref{eqn:23062503} has been proved, for any $t=1,\dots, T$ and any $l=1,\dots,k-1$. Our proof of the claim is thus completed. 
\Halmos
\end{proof}

\subsection{Construction of $\{\beta^*_{l,t}, \xi^*_t\}$ and a Constructive Proof of Theorem \ref{Prophetoptimaltheorem}}\label{Prooftheorem3}
Given \Cref{feasiproposition}, in order to prove Theorem \ref{Prophetoptimaltheorem}, it is enough for us to construct a feasible solution $\{\beta^*_{l,t}, \xi^*_t\}$ to $\DCkunit(\bm{p})$ in \eqref{dual} such that the primal-dual pair $\{\theta^*, x_{l,t}(\theta^*)\}$ and $\{\beta^*_{l,t}, \xi^*_t\}$ satisfies the complementary slackness conditions. Specifically, we will construct a feasible solution $\{\beta^*_{l,t}, \xi^*_t\}$ to $\DCkunit(\bm{p})$ satisfying the following conditions:
\begin{align}
&\beta^*_{1,t}\cdot\left(x_{1,t}(\theta^*)- \p_t\cdot (1-\sum_{\tau<t}x_{1,\tau}(\theta^*))\right)=0, ~~\forall t=1,\dots,T \label{complementaryslackness}\\
&\beta^*_{l,t}\cdot \left(x_{l,t}(\theta^*)- \p_t\cdot \sum_{\tau<t}(x_{l-1,\tau}(\theta^*)-x_{l,\tau}(\theta^*))\right)=0,~~\forall t=1,\dots, T, \forall l=2,\dots,k \nonumber\\
&x_{l,t}(\theta^*)\cdot \left(\beta^*_{l,t}+\sum_{\tau>t}\p_\tau\cdot (\beta^*_{l,\tau}-\beta^*_{l+1,\tau})-\xi^*_t\right)=0, ~~\forall t=1,\dots, T, \forall l=2,\dots,k \nonumber\\
&x_{k,t}(\theta^*)\cdot\left( \beta^*_{k,t}+\sum_{\tau>t}\p_\tau\cdot\beta^*_{k,\tau}-\xi^*_t\right)=0, ~~\forall t=1,\dots,T\nonumber
\end{align}
Note that from definitions, $\{x_{l,t}(\theta^*)\}$ satisfies the following conditions:
\[\begin{aligned}
&x_{l,t}(\theta^*)=0\leq \p_t\cdot \sum_{\tau=1}^{t-1}(x_{l-1,\tau}(\theta^*)-x_{l,\tau}(\theta^*)),~~\forall t\leq t_l\\
&x_{l,t}(\theta^*)=\theta^*\cdot \p_t-\sum_{v=1}^{l-1}x_{v,t}(\theta^*)\leq \p_t\cdot \sum_{\tau=1}^{t-1}(x_{l-1,\tau}(\theta^*)-x_{l,\tau}(\theta^*))~~~\text{for~}t_l+1\leq t\leq t_{l+1}
\end{aligned}\]
where $\{t_l\}$ are the time indexes associated with the definition of $\{x_{l,t}(\theta^*)\}$ and we define $t_1=0$, $t_{k+1}=T-1$. For simplicity, we also denote $\sum_{\tau=1}^{t-1}x_{0,\tau}(\theta^*)=1$ for any $t$. Thus, in order for $\{\beta^*_{l,t}, \xi^*_t\}$ to satisfy the conditions in \eqref{complementaryslackness}, it is enough for $\{\beta^*_{l,t}, \xi^*_t\}$ to be feasible to $\DCkunit(\bm{p})$ and satisfy the following conditions:
\begin{align}
&\beta^*_{l,t}=0~~~\text{for~}t\leq t_{l+1}\label{condition1}\\
&\beta^*_{l,t}+\sum_{\tau=t+1}^T p_\tau\cdot(\beta^*_{l,\tau}-\beta^*_{l+1,\tau} )=\xi^*_t~~~\text{for~}t\geq t_l+1 \label{condition2}
\end{align}
where we denote $\beta^*_{k+1,t}=0$ for notation simplicity. We now show the construction of the solution $\{\beta^*_{l,t}, \xi^*_t\}$ to $\DCkunit(\bm{p})$. Define the following constants for each $l, q\in\{1,2,\dots,k\}$:
\[
B_{l, q}=\sum_{t_l+1\leq j_1< j_2<\dots< j_q\leq t_{l+1}}\frac{\p_{j_1}\p_{j_2}\dots \p_{j_q}}{(1-\p_{j_1})(1-\p_{j_2})\dots(1-\p_{j_q})}\cdot\prod_{w=t_l+1}^{t_{l+1}}(1-\p_w)
\]
and we set $B_{l,0}=\prod_{w=t_l+1}^{t_{l+1}}(1-\p_w)$. We also define the following terms for each $l, q\in\{1,2,\dots,k\}$ and each $t\in\{t_{l}+1,\dots,t_{l+1}\}$, where $\{t_l\}$ are the time indexes defined in the construction of $\{\theta^*, x_{l,t}(\theta^*)\}$ and we define $t_1=0$, $t_{k+1}=T-1$:
\[
A_{l,q}(t)=\sum_{t+1\leq j_1< j_2<\dots< j_q\leq t_{l+1}}\frac{\p_{j_1}\p_{j_2}\dots \p_{j_q}}{(1-\p_{j_1})(1-\p_{j_2})\dots(1-\p_{j_q})}\cdot\prod_{w=t+1}^{t_{l+1}}(1-\p_w)
\]
and we set $A_{l,0}(t)=\prod_{w=t+1}^{t_{l+1}}(1-\p_w)$. Then our construction of the solution $\{\beta^*_{l,t}, \xi^*_t\}$ can be fully described as follows:
\begin{equation}\label{constructiondual}
\begin{aligned}
&\xi^*_T=\beta^*_{l_1,T}=R,~~~\forall l_1=1,2,\dots,k\\
&\beta^*_{l_1,t}=0,~~~\forall l_1=1,2,\dots,k, \forall t\leq t_{l_1+1}\\
&\xi^*_t=\phi_{l}\cdot \p_TR,~~~\forall l=1,2,\dots,k, \forall t_{l}+1\leq t\leq t_{l+1}\\
&\beta^*_{l_1,t}=\p_TR\cdot \sum_{w=l_1}^{l_2-1}\delta_{w,l_2}\cdot A_{l_2,w-l_1}(t),~~~\forall l_1=1,2,\dots,k, \forall l_2=l_1+1,\dots,k, \forall t_{l_2}+1\leq t\leq t_{l_2+1}\\
\end{aligned}
\end{equation}
where the parameters $\{\phi_{l}, \delta_{l_1,l_2}, R\}$ are defined as:
\[\begin{aligned}\label{parameter}
 &\delta_{l,k}=1~~~~~~~~\forall l=1,2,\dots,k-1 \\
 & \delta_{l,l}=0 ~~~~~~~~\forall l=1,2,\dots,k \\
 & \delta_{l_1,l_2}=\sum_{w_0=l_1+1}^{l_2}\sum_{w_1=w_0}^{l_2+1}\sum_{w_2=w_1}^{l_2+2}\dots\sum_{w_{k-1-l_2}=w_{k-2-l_2}}^{k-1}B_{l_2+1,w_1-w_0}\cdot B_{l_2+2,w_2-w_1}\dots B_{k-1,w_{k-1-l_2}-w_{k-2-l_2}}\cdot B_{k,k-w_{k-1-l_2}},\nonumber\\
 &~~~~~~~~~~~~~~~~~~~~~~~~~~~~~~~~~~~~~~~~~~~~~~~~~~~~~~~~~\forall l_2=1,2,\dots,k-1\text{~and~}l_1=1,2,\dots,l_2-1 \\
 &\phi_k=1\\
 &\phi_{l}=\sum_{q=l+1}^{k}\sum_{w=l+1}^{q}(\delta_{w-1,q}-\delta_{w,q})\cdot(1-\sum_{v=0}^{w-l-1}B_{q,v})~~~\forall l=1,2,\dots,k-1
\end{aligned}\]
and $R$ is a positive constant such that $\sum_{t=1}^{T}\p_t\cdot \xi^*_t=1$.  We then prove the feasibility of $\{\beta^*_{l,t}, \xi^*_t\}$ and the conditions \eqref{condition1}, \eqref{condition2} are satisfied. Obviously, from definition, $\beta^*_{l,t}$ is nonnegative for each $l$ and each $t$. We first prove that $\xi^*_t$ is also nonnegative for each $t$.
\begin{lemma}\label{duallemma1}
For each $l_2=1,2,\dots,k$ and each $l_1=1,2,\dots,l_2-1$, we have that $\delta_{l_1,l_2}\geq\delta_{l_1+1,l_2}$.
\end{lemma}
\begin{proof}{Proof:}
Note that when $l_2=k$, we have that $\delta_{l,k}=1$ for each $l=1,2,\dots,k-1$, thus it holds that $\delta_{l,k}\geq\delta_{l+1,k}$. When $l_2\leq k-1$, from definitions, we have that for each $l_1=1,2,\dots,l_2-1$
\[
\delta_{l_1,l_2}-\delta_{l_1+1,l_2}=\sum_{w_1=l_1+1}^{l_2+1}\sum_{w_2=w_1}^{l_2+2}\dots\sum_{w_{k-1-l_2}=w_{k-2-l_2}}^{k-1}B_{l_2+1,w_1-l_1-1}\cdot B_{l_2+2,w_2-w_1}\dots B_{k-1,w_{k-1-l_2}-w_{k-2-l_2}}\cdot B_{k,k-w_{k-1-l_2}}
\]
which completes our proof.
\Halmos
\end{proof}
We then show that the term $1-\sum_{w=0}^{q}B_{l,w}$ is nonnegative for each $l$ and each $q$. Note that the following lemma essentially implies that $\sum_{t=t_l+1}^{t_{l+1}}\p_t\cdot A_{l,q}(t)=1-\sum_{w=0}^{q}B_{l,w}$, by replacing $i_1$ with $t_l$ and $i_2$ with $t_{l+1}$ in \eqref{ABrelation}, which establishes the nonnegativity of the term $1-\sum_{w=0}^{q}B_{l,w}$.
\begin{lemma}\label{duallemma2}
For each $q\in\{1,2,\dots,k\}$ and any $1\leq i_1+1\leq i_2\leq T$, it holds that
\begin{equation}\label{ABrelation}
\begin{aligned}
&\sum_{t=i_1+1}^{i_2}\p_t\cdot\sum_{t+1\leq j_1< j_2<\dots< j_q\leq i_2}\frac{\p_{j_1}\p_{j_2}\dots \p_{j_q}}{(1-\p_{j_1})(1-\p_{j_2})\dots(1-\p_{j_q})}\cdot\prod_{v=t+1}^{i_2}(1-\p_v)\\
&=1-\sum_{w=0}^{q}\sum_{i_1+1\leq j_1< j_2<\dots< j_w\leq i_2}\frac{\p_{j_1}\p_{j_2}\dots \p_{j_w}}{(1-\p_{j_1})(1-\p_{j_2})\dots(1-\p_{j_w})}\cdot\prod_{v=i_1+1}^{i_2}(1-\p_v),
\end{aligned}
\end{equation}
\end{lemma}
\begin{proof}{Proof:}
We will do induction on $q$ from $q=0$ to $q=k$ to prove \eqref{ABrelation}.
When $q=0$, we have that
\[\begin{aligned}
\sum_{t=i_1+1}^{i_2}\p_t\cdot\prod_{v=t+1}^{i_2}(1-\p_v)&=\sum_{t=i_1+1}^{i_2}(1-(1-\p_t))\cdot\prod_{v=t+1}^{i_2}(1-\p_v)=\sum_{t=i_1+1}^{i_2}\left( \prod_{v=t+1}^{i_2}(1-\p_v)-\prod_{v=t}^{i_2}(1-\p_v) \right)\\
&=1-\prod_{v=i_1+1}^{i_2}(1-\p_v)
\end{aligned}\]
Thus, we have \eqref{ABrelation} holds for $q=0$. Suppose \eqref{ABrelation} holds for $1,2,\dots,q-1$, we consider the case for $q$. For any $1\leq i_1+1\leq i_2\leq T$, we have that
\[\begin{aligned}
&\sum_{t=i_1+1}^{i_2}\p_t\cdot\sum_{t+1\leq j_1< j_2<\dots< j_q\leq i_2}\frac{\p_{j_1}\p_{j_2}\dots \p_{j_q}}{(1-\p_{j_1})(1-\p_{j_2})\dots(1-\p_{j_q})}\cdot\prod_{v=t+1}^{i_2}(1-\p_v)\\
&=\sum_{t=i_1+1}^{i_2}\p_t\cdot\sum_{j_1= t+1}^{i_2}\sum_{j_1< j_2<\dots< j_q\leq i_2}\frac{\p_{j_1}\p_{j_2}\dots \p_{j_q}}{(1-\p_{j_1})(1-\p_{j_2})\dots(1-\p_{j_q})}\cdot\prod_{v=t+1}^{i_2}(1-\p_v)\\
&=\sum_{j_1=i_1+2}^{i_2}\sum_{t=i_1+1}^{j_1-1}\p_t\cdot \frac{ \p_{j_1}}{1-\p_{j_1}}\cdot\sum_{j_1< j_2<\dots< j_q\leq i_2}\frac{\p_{j_2}\dots \p_{j_q}}{(1-\p_{j_2})\dots(1-\p_{j_q})}\cdot\prod_{v=t+1}^{i_2}(1-\p_v)\\
&=\sum_{j_1=i_1+2}^{i_2}\frac{ \p_{j_1}}{1-\p_{j_1}}\cdot\sum_{j_1< j_2<\dots< j_q\leq i_2}\frac{\p_{j_2}\dots \p_{j_q}}{(1-\p_{j_2})\dots(1-\p_{j_q})}
\cdot\prod_{v=j_1}^{i_2}(1-\p_v)\cdot \sum_{t=i_1+1}^{j_1-1}\p_t\cdot\prod_{v=t+1}^{j_1-1}(1-\p_v)
\end{aligned}\]
where the second equality holds by exchanging the order of summation. Note that for induction purpose, we assume \eqref{ABrelation} holds for $q=0$, which implies that $\sum_{t=i_1+1}^{j_1-1}\p_t\cdot\prod_{v=t+1}^{j_1-1}(1-\p_v)=1-\prod_{v=i_1+1}^{j_1-1}(1-\p_v)$. Then we have
\[\begin{aligned}
&\sum_{j_1=i_1+2}^{i_2}\frac{ \p_{j_1}}{1-\p_{j_1}}\cdot\sum_{j_1< j_2<\dots< j_q\leq i_2}\frac{\p_{j_2}\dots \p_{j_q}}{(1-\p_{j_2})\dots(1-\p_{j_q})}
\cdot\prod_{v=j_1}^{i_2}(1-\p_v)\cdot \sum_{t=i_1+1}^{j_1-1}\p_t\cdot\prod_{v=t+1}^{j_1-1}(1-\p_v)\\
&=\sum_{j_1=i_1+2}^{i_2}\p_{j_1}\cdot\sum_{j_1< j_2<\dots< j_q\leq i_2}\frac{\p_{j_2}\dots \p_{j_q}}{(1-\p_{j_2})\dots(1-\p_{j_q})}\cdot\prod_{v=j_1+1}^{i_2}(1-\p_v)
\cdot\left(1-\prod_{v=i_1+1}^{j_1-1}(1-\p_v) \right)\\
&=\sum_{j_1=i_1+1}^{i_2}\p_{j_1}\cdot\sum_{j_1< j_2<\dots< j_q\leq i_2}\frac{\p_{j_2}\dots \p_{j_q}}{(1-\p_{j_2})\dots(1-\p_{j_q})}\cdot\prod_{v=j_1+1}^{i_2}(1-\p_v)
\cdot\left(1-\prod_{v=i_1+1}^{j_1-1}(1-\p_v) \right)
\end{aligned}\]
where the second equality holds by noting that when $j_1=i_1+1$, we have $1-\prod_{v=i_1+1}^{j_1-1}(1-\p_v)=0$. Thus, it holds that
\[\begin{aligned}
&\sum_{t=i_1+1}^{i_2}\p_t\cdot\sum_{t+1\leq j_1< j_2<\dots< j_q\leq i_2}\frac{\p_{j_1}\p_{j_2}\dots \p_{j_q}}{(1-\p_{j_1})(1-\p_{j_2})\dots(1-\p_{j_q})}\cdot\prod_{v=t+1}^{i_2}(1-\p_v)\\
&=\sum_{j_1=i_1+1}^{i_2}\p_{j_1}\cdot\sum_{j_1< j_2<\dots< j_q\leq i_2}\frac{\p_{j_2}\dots \p_{j_q}}{(1-\p_{j_2})\dots(1-\p_{j_q})}\cdot\prod_{v=j_1+1}^{i_2}(1-\p_v)
\cdot\left(1-\prod_{v=i_1+1}^{j_1-1}(1-\p_v) \right)
\end{aligned}\]
Note that for the induction purpose, we assume that \eqref{ABrelation} holds for $q-1$. Then, we have that
\[\begin{aligned}
&\sum_{j_1=i_1+1}^{i_2}\p_{j_1}\cdot\sum_{j_1< j_2<\dots< j_q\leq i_2}\frac{\p_{j_2}\dots \p_{j_q}}{(1-\p_{j_2})\dots(1-\p_{j_q})}\cdot\prod_{v=j_1+1}^{i_2}(1-\p_v)\\
&=\sum_{t=i_1+1}^{i_2}\p_t\cdot\sum_{t+1\leq j_1< j_2<\dots< j_{q-1}\leq i_2}\frac{\p_{j_1}\p_{j_2}\dots \p_{j_{q-1}}}{(1-\p_{j_1})(1-\p_{j_2})\dots(1-\p_{j_{q-1}})}\cdot\prod_{v=t+1}^{i_2}(1-\p_v)\\
&=1-\sum_{w=0}^{q-1}\sum_{i_1+1\leq j_1< j_2<\dots< j_w\leq i_2}\frac{\p_{j_1}\p_{j_2}\dots \p_{j_w}}{(1-\p_{j_1})(1-\p_{j_2})\dots(1-\p_{j_w})}\cdot\prod_{v=i_1+1}^{i_2}(1-\p_v)
\end{aligned}\]
where the second equality holds from replacing the index $j_{l+1}$ with $j_l$ for $l=2,\dots,q$ and replace the index $j_1$ with $t$. Also, note that
\[\begin{aligned}
&\sum_{j_1=i_1+1}^{i_2}\p_{j_1}\cdot\sum_{j_1< j_2<\dots< j_q\leq i_2}\frac{\p_{j_2}\dots \p_{j_q}}{(1-\p_{j_2})\dots(1-\p_{j_q})}\cdot\prod_{v=j_1+1}^{i_2}(1-\p_v)
\cdot\prod_{v=i_1+1}^{j_1-1}(1-\p_v)\\
&=\sum_{i_1+1\leq j_1< j_2<\dots< j_q\leq i_2}\frac{\p_{j_1}\p_{j_2}\dots \p_{j_q}}{(1-\p_{j_1})(1-\p_{j_2})\dots(1-\p_{j_q})}\cdot\prod_{v=i_1+1}^{i_2}(1-\p_v)
\end{aligned}\]
Thus, we have that
\[\begin{aligned}
&\sum_{t=i_1+1}^{i_2}\p_t\cdot\sum_{t+1\leq j_1< j_2<\dots< j_q\leq i_2}\frac{\p_{j_1}\p_{j_2}\dots \p_{j_q}}{(1-\p_{j_1})(1-\p_{j_2})\dots(1-\p_{j_q})}\cdot\prod_{v=t+1}^{i_2}(1-\p_v)\\
&=\sum_{j_1=i_1+1}^{i_2}\p_{j_1}\cdot\sum_{j_1< j_2<\dots< j_q\leq i_2}\frac{\p_{j_2}\dots \p_{j_q}}{(1-\p_{j_2})\dots(1-\p_{j_q})}\cdot\prod_{v=j_1+1}^{i_2}(1-\p_v)
\cdot\left(1-\prod_{v=i_1+1}^{j_1-1}(1-\p_v) \right)\\
&=1-\sum_{w=0}^{q}\sum_{i_1+1\leq j_1< j_2<\dots< j_w\leq i_2}\frac{\p_{j_1}\p_{j_2}\dots \p_{j_w}}{(1-\p_{j_1})(1-\p_{j_2})\dots(1-\p_{j_w})}\cdot\prod_{v=i_1+1}^{i_2}(1-\p_v)
\end{aligned}\]
which completes our proof by induction on $q$.
\Halmos
\end{proof}
~\\
Combining Lemma \ref{duallemma1} and Lemma \ref{duallemma2}, we draw the following conclusion.
\begin{lemma}\label{duallemma3}
For each $l=1,2,\dots,k$ and each $t=1,2,\dots,T$, we have that $\beta^*_{l,t}\geq0$ and $\xi^*_t\geq0$.
\end{lemma}
\begin{proof}{Proof:}
Note that from definition, $\beta^*_{l,t}\geq0$ for each $l$ and $t$. We then show the non-negativity of $\xi^*_t$ for each $t$. Note that Lemma \ref{duallemma1} shows that $\delta_{l_1,l_2}\geq\delta_{l_1+1,l_2}$ for each $l_2=1,2,\dots,k$ and each $l_1=1,2,\dots,l_2-1$. It only remains to show the non-negativity of the term $1-\sum_{w=0}^{q}B_{l,w}$, which can be directly established by Lemma \ref{duallemma2}. Specifically, by replacing $i_1$ with $t_l$ and $i_2$ with $t_{l+1}$ in \eqref{ABrelation}, we have $1-\sum_{w=0}^{q}B_{l,w}=\sum_{t=t_l+1}^{t_{l+1}}\p_t\cdot A_{l,q}(t)\geq0$.
\Halmos
\end{proof}
~\\
From the definition of $\{\beta^*_{l,t}, \xi^*_t\}$, condition \eqref{condition1} holds obviously. We then prove that condition \eqref{condition2} is satisfied.
\begin{lemma}\label{duallemma4}
For each $l=1,2,\dots,k$ and each $t\geq t_{l}+1$, it holds that
\[
\beta^*_{l,t}+\sum_{j=t+1}^T \p_j\cdot(\beta^*_{l,j}-\beta^*_{l+1,j} )=\xi^*_t
\]
where we denote $\beta^*_{k+1,t}=0$ for notation simplicity.
\end{lemma}
\begin{proof}{Proof:}
When $l=k$, from definition, we have $\beta^*_{l,j}=0$ for each $j\leq t_{l+1}=T-1$ and $\beta^*_{l,T}=R$, thus the lemma holds directly. When $t=T$, it is also direct to show from definition that the lemma holds. We then focus on the case where $l\leq k-1$ and $t\leq T-1$.

For a fixed $l\leq k-1$ and a fixed $t_{l}+1\leq t\leq T-1$, we denote an index $l_1\geq l$ such that $t_{l_1}+1\leq t\leq t_{l_1+1}$. We then consider the following cases separately based on the value of $l_1$.

(i). When $l_1\leq k-1$, we have that
\begin{equation}\label{001}
\beta^*_{l,t}=\p_TR\cdot \sum_{w=l}^{l_1-1}\delta_{w,l_1}\cdot A_{l_1,w-l}(t)
\end{equation}
also, for any $t+1\leq j\leq t_{l_1+1}$, we have that
\[
\beta^*_{l,j}-\beta^*_{l+1,j}=\p_TR\cdot \sum_{w=l}^{l_1-1}(\delta_{w,l_1}-\delta_{w+1,l_1})\cdot A_{l_1,w-l}(j)
\]
which implies that
\[
\sum_{j=t+1}^{t_{l_1+1}}\p_j\cdot(\beta^*_{l,j}-\beta^*_{l+1,j})=\p_TR\cdot \sum_{w=l}^{l_1-1}(\delta_{w,l_1}-\delta_{w+1,l_1})\cdot\sum_{j=t+1}^{t_{l_1+1}}\p_j\cdot A_{l_1,w-l}(j)
\]
Note that from \eqref{ABrelation}, it holds that $\sum_{j=t+1}^{t_{l_1+1}}\p_j\cdot A_{l_1,w-l}(j)=1-\sum_{q=0}^{w-l}A_{l_1,q}(t)$. Thus, we have that
\begin{equation}\label{002}
\begin{aligned}
\sum_{j=t+1}^{t_{l_1+1}}\p_j\cdot(\beta^*_{l,j}-\beta^*_{l+1,j})&=\p_TR\cdot\sum_{w=l}^{l_1-1}(\delta_{w,l_1}-\delta_{w+1,l_1})\cdot\left(1-\sum_{q=0}^{w-l}A_{l_1,q}(t)\right)\\
&=\p_TR\cdot \delta_{l,l_1}-\p_TR\cdot \sum_{w=l}^{l_1-1}\delta_{w,l_1}\cdot A_{l_1,w-l}(t)
\end{aligned}
\end{equation}
where the last equality holds from $\delta_{l_1,l_1}=0$. Similarly, for any $l_2\geq l_1+1$ and any $t_{l_2}+1\leq j\leq t_{l_2+1}$, we have that
\[
\beta^*_{l,j}-\beta^*_{l+1,j}=\p_TR\cdot \sum_{w=l}^{l_2-1}(\delta_{w,l_2}-\delta_{w+1,l_2})\cdot A_{l_2,w-l}(j)
\]
which implies that
\[
\sum_{j=t_{l_2}+1}^{t_{l_2+1}}\p_j\cdot(\beta^*_{l,j}-\beta^*_{l+1,j})=\p_TR\cdot \sum_{w=l}^{l_2-1}(\delta_{w,l_2}-\delta_{w+1,l_2})\cdot\sum_{j=t_{l_2}+1}^{t_{l_2+1}}\p_j\cdot A_{l_2,w-l}(j)
\]
Note that from Lemma \ref{duallemma2}, we have that $\sum_{j=t_{l_2}+1}^{t_{l_2+1}}\p_j\cdot A_{l_2,w-l}(j)=1-\sum_{q=0}^{w-l} B_{l_2,q}$. Thus, we have that
\begin{equation}\label{003}
\begin{aligned}
\sum_{j=t_{l_1+1}+1}^{t_{k+1}}\p_j\cdot(\beta^*_{l,j}-\beta^*_{l+1,j})&=\sum_{l_2=l_1+1}^{k}\sum_{j=t_{l_2}+1}^{t_{l_2+1}}\p_j\cdot(\beta^*_{l,j}-\beta^*_{l+1,j})\\
&=\sum_{l_2=l_1+1}^{k}\p_TR\cdot \sum_{w=l}^{l_2-1}(\delta_{w,l_2}-\delta_{w+1,l_2})\cdot(1-\sum_{q=0}^{w-l} B_{l_2,q})\\
&=\p_TR\cdot\sum_{l_2=l_1+1}^{k}\sum_{w=l+1}^{l_2}(\delta_{w-1,l_2}-\delta_{w,l_2})\cdot(1-\sum_{q=0}^{w-l-1} B_{l_2,q})\\
\end{aligned}
\end{equation}
Combining \eqref{001}, \eqref{002} and \eqref{003}, we have that
\begin{equation}\label{006}
\beta^*_{l,t}+\sum_{j=t+1}^{T}\p_j\cdot(\beta^*_{l,j}-\beta^*_{l+1,j})=\p_TR\cdot \delta_{l,l_1}+\p_TR\cdot\sum_{l_2=l_1+1}^{k}\sum_{w=l+1}^{l_2}(\delta_{w-1,l_2}-\delta_{w,l_2})\cdot(1-\sum_{q=0}^{w-l-1} B_{l_2,q})
\end{equation}
Note that
\[
\xi^*_t=\p_TR\cdot\sum_{l_2=l_1+1}^{k}\sum_{w=l_1+1}^{l_2}(\delta_{w-1,l_2}-\delta_{w,l_2})\cdot(1-\sum_{q=0}^{w-l_1-1} B_{l_2,q})
\]
in order to show $\beta^*_{l,t}+\sum_{j=t+1}^{T}\p_j\cdot(\beta^*_{l,j}-\beta^*_{l+1,j})=\xi^*_t$, it is enough to prove that
\begin{equation}\label{004}
\delta_{l,l_1}+\sum_{l_2=l_1+1}^{k}\sum_{w=l+1}^{l_2}(\delta_{w-1,l_2}-\delta_{w,l_2})\cdot(1-\sum_{q=0}^{w-l-1} B_{l_2,q})=\sum_{l_2=l_1+1}^{k}\sum_{w=l_1+1}^{l_2}(\delta_{w-1,l_2}-\delta_{w,l_2})\cdot(1-\sum_{q=0}^{w-l_1-1} B_{l_2,q})
\end{equation}
Further note that
\[\begin{aligned}
&\sum_{l_2=l_1+1}^{k}\sum_{w=l+1}^{l_2}(\delta_{w-1,l_2}-\delta_{w,l_2})\cdot(1-\sum_{q=0}^{w-l-1} B_{l_2,q})\\
=&\sum_{l_2=l_1+1}^{k}\sum_{w=l+1}^{l_2}(\delta_{w-1,l_2}-\delta_{w,l_2})
-\sum_{l_2=l_1+1}^{k}\sum_{w=l+1}^{l_2}\sum_{q=0}^{w-l-1}B_{l_2,q}\cdot(\delta_{w-1,l_2}-\delta_{w,l_2})\\
=&\sum_{l_2=l_1+1}^{k}\delta_{l,l_2}-\sum_{l_2=l_1+1}^{k}\sum_{q=0}^{l_2-l-1}B_{l_2,q}\cdot\delta_{q+l,l_2}
\end{aligned}\]
and similarly, note that
\[
\sum_{l_2=l_1+1}^{k}\sum_{w=l_1+1}^{l_2}(\delta_{w-1,l_2}-\delta_{w,l_2})\cdot(1-\sum_{q=0}^{w-l_1-1} B_{l_2,q})=\sum_{l_2=l_1+1}^{k}\delta_{l_1,l_2}-\sum_{l_2=l_1+1}^{k}\sum_{q=0}^{l_2-l_1-1}B_{l_2,q}\cdot\delta_{q+l_1,l_2}
\]
in order to prove \eqref{004}, it is enough to show that
\begin{equation}\label{005}
\sum_{l_2=l_1}^{k}\delta_{l,l_2}-\sum_{l_2=l_1+1}^{k}\sum_{q=0}^{l_2-l-1}B_{l_2,q}\cdot\delta_{q+l,l_2}=
\sum_{l_2=l_1+1}^{k}\delta_{l_1,l_2}-\sum_{l_2=l_1+1}^{k}\sum_{q=0}^{l_2-l_1-1}B_{l_2,q}\cdot\delta_{q+l_1,l_2}
\end{equation}
When $l_1=l$, it is direct to check that \eqref{005} holds. The proof of \eqref{005} when $l_1\geq l+1$ is relegated to Lemma \ref{duallemma5}. Thus, we prove that when $l_1\leq k-1$, it holds that $\beta^*_{l,t}+\sum_{j=t+1}^{T}\p_j\cdot(\beta^*_{l,j}-\beta^*_{l+1,j})=\xi^*_t$.

(ii). When $l_1=k$, we have that
\[
\beta^*_{l,t}=\p_TR\cdot\sum_{w=l}^{k-1}A_{k,w-l}(t)
\]
and for each $t+1\leq j\leq T-1$, it holds that
\[
\beta^*_{l,j}-\beta^*_{l+1,j}=\p_TR\cdot\left(\sum_{w=l}^{k-1}A_{k,w-l}(j)-\sum_{w=l+1}^{k-1}A_{k,w-l-1}(j)\right)=\p_TR\cdot A_{k,k-1-l}(j)
\]
Note that $\beta^*_{l,T}=\beta^*_{l+1,T}=R$, we have
\[
\beta^*_{l,t}+\sum_{j=t+1}^{T-1}\p_j\cdot(\beta^*_{l,j}-\beta^*_{l+1,j})=\p_TR\cdot\left(\sum_{w=l}^{k-1}A_{k,w-l}(t)+\sum_{j=t+1}^{T-1}\p_j\cdot A_{k,k-1-l}(j) \right)
\]
Note that from \eqref{ABrelation}, it holds that $\sum_{j=t+1}^{T-1}\p_j\cdot A_{k,k-1-l}(j)=1-\sum_{q=0}^{k-1-l}A_{k,q}(t)$. Thus, we have that
\[
\beta^*_{l,t}+\sum_{j=t+1}^{T-1}\p_j\cdot(\beta^*_{l,j}-\beta^*_{l+1,j})=\p_TR=\xi^*_t
\]
which completes our proof.
\Halmos
\end{proof}
\begin{lemma}\label{duallemma5}
For each $l=1,2,\dots,k-1$ and each $l_1=l, l+1, \dots,k-1$, it holds that
\begin{equation}\label{007}
\sum_{l_2=l_1}^{k}\delta_{l,l_2}-\sum_{l_2=l_1+1}^{k}\sum_{q=0}^{l_2-l-1}B_{l_2,q}\cdot\delta_{q+l,l_2}=
\sum_{l_2=l_1+1}^{k}\delta_{l_1,l_2}-\sum_{l_2=l_1+1}^{k}\sum_{q=0}^{l_2-l_1-1}B_{l_2,q}\cdot\delta_{q+l_1,l_2}
\end{equation}
\end{lemma}
\begin{proof}{Proof:}
We now prove \eqref{007} by induction on $l$ from $l=k-1$ to $l=1$. When $l=k-1$, we must have $l_1=k-1=l$, then \eqref{007} holds obviously. Suppose that there exists a $1\leq l'\leq k-2$ such that for any $l$ satisfying $l'+1\leq l\leq k-1$, \eqref{007} holds for each $l_1$ such that $l\leq l_1\leq k-1$, then we consider the case when $l=l'$. For this case, we again use induction on $l_1$ from $l_1=k-1$ to $l_1=l+1=l'+1$. When $l_1=k-1$, we have that
\[
\sum_{l_2=l_1}^{k}\delta_{l,l_2}-\sum_{l_2=l_1+1}^{k}\sum_{q=0}^{l_2-l-1}B_{l_2,q}\cdot\delta_{q+l,l_2}=\delta_{l,k-1}+\delta_{l,k}-\sum_{q=0}^{k-l-1}B_{k,q}\cdot\delta_{q+l,k}
\]
and
\[
\sum_{l_2=l_1+1}^{k}\delta_{l_1,l_2}-\sum_{l_2=l_1+1}^{k}\sum_{q=0}^{l_2-l_1-1}B_{l_2,q}\cdot\delta_{q+l_1,l_2}=\delta_{k-1,k}-B_{k,0}\cdot\delta_{k-1,k}
\]
Further note that from definition, $\delta_{v,k}=1$ for each $v\leq k-1$ and $\delta_{l,k-1}=\sum_{w_0=l+1}^{k-1}B_{k,k-w_0}=\sum_{q=1}^{k-l-1}B_{k,q}$, it is obvious that \eqref{007} holds when $l_1=k-1$. Now suppose that \eqref{007} holds for $l_1+1$ (we assume $l_1\geq l+1$ since when $l_1=l$, it is direct from definition that \eqref{007} holds), we consider the case for $l_1$. Note that
\[
\text{LHS~of~}\eqref{007}=\delta_{l,l_1}-\sum_{q=0}^{l_1-l}B_{l_1+1,q}\cdot\delta_{q+l,l_1+1}+\sum_{l_2=l_1+1}^{k}\delta_{l,l_2}-\sum_{l_2=l_1+2}^{k}\sum_{q=0}^{l_2-l-1}B_{l_2,q}\cdot\delta_{q+l,l_2}
\]
and
\[
\text{RHS~of~}\eqref{007}=\delta_{l_1,l_1+1}-\sum_{l_2=l_1+1}^{k}B_{l_2,l_2-l_1-1}\cdot\delta_{l_2-1,l_2}+\sum_{l_2=l_1+2}^{k}\delta_{l_1,l_2}-\sum_{l_2=l_1+2}^{k}
\sum_{q=0}^{l_2-l_1-2}B_{l_2,q}\cdot\delta_{q+l_1,l_2}
\]
Since we suppose for induction that \eqref{007} holds for $l_1+1$, we have that
\[
\eqref{007}\text{~holds~for~} (l, l_1)\Leftrightarrow \delta_{l,l_1}-\sum_{q=0}^{l_1-l}B_{l_1+1,q}\cdot\delta_{q+l,l_1+1}=\delta_{l_1,l_1+1}-\sum_{l_2=l_1+1}^{k}B_{l_2,l_2-l_1-1}\cdot\delta_{l_2-1,l_2}
\]
Further note that we have supposed for induction that \eqref{007} holds for $(l+1, l_1)$, which implies
\[
\delta_{l+1,l_1}-\sum_{q=0}^{l_1-l-1}B_{l_1+1,q}\cdot\delta_{q+l+1,l_1+1}=\delta_{l_1,l_1+1}-\sum_{l_2=l_1+1}^{k}B_{l_2,l_2-l_1-1}\cdot\delta_{l_2-1,l_2}
\]
Thus, it holds that
\[
\eqref{007}\text{~holds~for~} (l, l_1)\Leftrightarrow \delta_{l,l_1}-\delta_{l+1,l_1}=\sum_{q=0}^{l_1-l}B_{l_1+1,q}\cdot(\delta_{q+l,l_1+1}-\delta_{q+l+1,l_1+1})
\]
Finally, from definition, we have
\[
\delta_{l,l_1}-\delta_{l+1,l_1}=\sum_{w_1=l+1}^{l_1+1}\sum_{w_2=w_1}^{l_1+2}\dots\sum_{w_{k-1-l_1}=w_{k-2-l_1}}^{k-1}B_{l_1+1,w_1-l-1}\cdot B_{l_1+2,w_2-w_1}\dots B_{k-1,w_{k-1-l_1}-w_{k-2-l_1}}\cdot B_{k,k-w_{k-1-l_1}}
\]
and
\[
\delta_{q+l,l_1+1}-\delta_{q+l+1,l_1+1}=\sum_{w_2=q+l+1}^{l_1+2}\dots\sum_{w_{k-1-l_1}=w_{k-2-l_1}}^{k-1}B_{l_1+2,w_2-q-l-1}\dots B_{k-1,w_{k-1-l_1}-w_{k-2-l_1}}\cdot B_{k,k-w_{k-1-l_1}}
\]
which implies that
\begin{equation}\label{008}
\delta_{l,l_1}-\delta_{l+1,l_1}=\sum_{q=0}^{l_1-l}B_{l_1+1,q}\cdot(\delta_{q+l,l_1+1}-\delta_{q+l+1,l_1+1})
\end{equation}
Thus, from induction, we prove that \eqref{007} holds for each $l_1\geq l+1$. Note that \eqref{007} holds obviously for $l_1=l$, \eqref{007} holds for each $l_1\geq l$. From the induction on $l$, we know that \eqref{007} holds for each $1\leq l\leq k-1$ and each $l\leq l_1\leq k-1$, which completes our proof.
\Halmos
\end{proof}
~\\
Finally, we only need to prove feasibility of $\{\beta^*_{l,t}, \xi^*_t\}$ in the following lemma.
\begin{lemma}\label{duallemma6}
For each $l=1,2,\dots,k$ and each $t=1,2,\dots,t_l$, it holds that
\[
\beta^*_{l,t}+\sum_{j=t+1}^T \p_j\cdot(\beta^*_{l,j}-\beta^*_{l+1,j} )\geq \xi^*_t
\]
where we denote $\beta^*_{k+1,t}=0$ for notation simplicity.
\end{lemma}
\begin{proof}{Proof:}
Note that from Lemma \ref{duallemma1}, we have $\delta_{w,l_2}\geq\delta_{w+1,l_2}$, which implies that $\beta^*_{l,j}\geq\beta^*_{l+1,j}$ for each $l$ and $j$. Thus, we have that for each $t=1,2,\dots,t_l$, it holds that
\[
\beta^*_{l,t}+\sum_{j=t+1}^T \p_j\cdot(\beta^*_{l,j}-\beta^*_{l+1,j} )\geq\beta^*_{l,t_l+1}+\sum_{j=t_l+2}^T \p_j\cdot(\beta^*_{l,j}-\beta^*_{l+1,j} )
\]
Further note that Lemma \ref{duallemma4} implies that
\[
\beta^*_{l,t_l+1}+\sum_{j=t_l+2}^n p_j\cdot(\beta^*_{l,j}-\beta^*_{l+1,j} )\geq \xi^*_{t_l+1}
\]
Thus, it is enough to show that $\xi^*_i\leq \xi^*_{t_l+1}$ for each $t=1,2,\dots,t_l$. From the definition of $\xi^*_i$, it is enough to show that $\phi_l\leq\phi_{l+1}$. When $l=k-1$, we have $\phi_{l+1}=\phi_k=1$ and $\phi_l=\phi_{k-1}=1-B_{k,0}$, which implies that $\phi_{k-1}\leq\phi_k$. When $l\leq k-2$, from definition, we have
\[
\phi_l-\phi_{l+1}=\sum_{q=l+1}^{k}(\delta_{l,q}-\delta_{l+1,q})\cdot(1-B_{q,0})-\sum_{q=l+2}^{k}\sum_{w=l+2}^{q}(\delta_{w-1,q}-\delta_{w,q})\cdot B_{q,w-l-1}
\]
Note that in the proof of Lemma \ref{duallemma5}, we proved \eqref{008}, then when $k-1\geq q\geq l+1$, we have
\[
\delta_{l,q}-\delta_{l+1,q}=\sum_{w=0}^{q-l}B_{q+1,w}\cdot(\delta_{w+l,q+1}-\delta_{w+l+1,q+1})=\sum_{w=l+1}^{q+1}B_{q+1,w-l-1}\cdot(\delta_{w-1,q+1}-\delta_{w,q+1})
\]
Thus, it holds that
\[\begin{aligned}
\phi_l-\phi_{l+1}&=-\sum_{q=l+1}^{k-1}(\delta_{l,q}-\delta_{l+1,q})\cdot B_{q,0}+\sum_{q=l+1}^{k-1}\sum_{w=l+1}^{q+1}B_{q+1,w-l-1}\cdot(\delta_{w-1,q+1}-\delta_{w,q+1})\\
&~~~~-\sum_{q=l+2}^{k}\sum_{w=l+2}^{q}(\delta_{w-1,q}-\delta_{w,q})\cdot B_{q,w-l-1}\\
&=-\sum_{q=l+1}^{k-1}(\delta_{l,q}-\delta_{l+1,q})\cdot B_{q,0}+\sum_{q=l+1}^{k-1}B_{q+1,0}\cdot(\delta_{l,q+1}-\delta_{l+1,q+1})\\
&=-B_{l+1,0}\cdot\delta_{l,l+1}\leq 0
\end{aligned}\]
which completes our proof.
\Halmos
\end{proof}
~\\
Together, Lemma \ref{duallemma3}, Lemma \ref{duallemma4}, and Lemma \ref{duallemma6} establish the feasibility of $\{\beta^*_{l,t}, \xi^*_t\}$. Then, from the definition of $\{\beta^*_{l,t}, \xi^*_t\}$, obviously condition \eqref{condition1} is satisfied and from Lemma \ref{duallemma4}, condition \eqref{condition2} is satisfied. Thus, we finish the proof of Theorem \ref{Prophetoptimaltheorem}.

\subsection{Proof of Lemma \ref{splititemlemma}}\label{ProofLemma8}
\begin{proof}{Proof:}
Since we have $\LCkunit(\bm{\p})=\DCkunit(\bm{\p})$, it is enough to consider the dual LP $\DCkunit(\bm{p})$ in \eqref{dual} and prove that $\DCkunit(\bm{\p})\geq \DCkunit(\tilde{\bm{p}})$. Suppose the optimal solution of $\DCkunit(\bm{\p})$ is denoted as $\{\beta^*_{l,t}, \xi^*_t\}$, as constructed in \eqref{constructiondual}, we then construct a feasible solution $\{\tilde{\beta}_{l,t}, \tilde{\xi}_t\}$ to $\DCkunit(\tilde{\bm{p}})$ as follows:
\[\begin{aligned}
&\tilde{\xi}_t=\xi^*_t~~~\forall 1\leq t<q,~~~~\tilde{\xi}_q=\tilde{\xi}_{q+1}=\xi^*_q,~~~~\tilde{\xi}_{t+1}=\xi^*_t~~~\forall  q+1\leq t\leq T\\
&\tilde{\beta}_{l,t}=\beta^*_{l,t}~~~\forall l=1,\dots,k, \forall 1\leq t<q\\
&\tilde{\beta}_{l,q}=\tilde{\beta}_{l,q+1}=\beta^*_{l,q}~~~\forall l=1,\dots,k\\
&\tilde{\beta}_{l,t+1}=\beta^*_{l,t}~~~\forall l=1,\dots,k, \forall q+1\leq t\leq T
\end{aligned}\]
Note that we have
\[
\DCkunit(\bm{\p})=\sum_{t=1}^{T}\p_t\cdot\beta^*_{1,t}=\sum_{t=1}^{T+1}\tilde{p}_t\cdot\tilde{\beta}_{1,t}
\]
it is enough to prove that $\{\tilde{\beta}_{l,t}, \tilde{\xi}_t\}$ is feasible to $\text{Primal}(\tilde{\bm{p}},k)$. Obviously, we have $\{\tilde{\beta}_{l,t}, \tilde{\xi}_t\}$ are non-negative and $\sum_{t=1}^{T+1}\tilde{p}_t\cdot\tilde{\xi}_t=\sum_{t=1}^{T}\p_t\cdot\xi^*_t=1$, then we only need to check whether the following constraint is satisfied:
\begin{equation}\label{splitcheck}
\tilde{\beta}_{l,t}+\sum_{\tau>t}\tilde{p}_\tau\cdot (\tilde{\beta}_{l,\tau}-\tilde{\beta}_{l+1,\tau})-\tilde{\xi}_t\geq0,~~~\forall l=1,\dots,k, \forall t=1,\dots,T+1
\end{equation}
where we denote $\tilde{\beta}_{k+1,t}=0$ for notation simplicity. Note that when $t\geq q+1$, we have that
\[
\tilde{\beta}_{l,t}+\sum_{\tau>t}\tilde{p}_\tau\cdot (\tilde{\beta}_{l,\tau}-\tilde{\beta}_{l+1,\tau})-\tilde{\xi}_t=\beta^*_{l,t}+\sum_{\tau>t}\p_\tau\cdot (\beta^*_{l,\tau}-\beta^*_{l+1,\tau})-\xi^*_t\geq0,~~\forall l=1,\dots,k
\]
and when $1\leq t\leq q-1$, we also have
\[
\tilde{\beta}_{l,t}+\sum_{\tau>t}\tilde{p}_\tau\cdot (\tilde{\beta}_{l,\tau}-\tilde{\beta}_{l+1,\tau})-\tilde{\xi}_t=\beta^*_{l,t}+\sum_{\tau>t}\p_\tau\cdot (\beta^*_{l,\tau}-\beta^*_{l+1,\tau})-\xi^*_t\geq0,~~\forall l=1,\dots,k
\]
by noting $\tilde{p}_q+\tilde{p}_{q+1}=\p_q$. Now we consider the case when $t=q$, then for each $l=1,\dots,k$, we have
\begin{align}
\tilde{\beta}_{l,q}+\sum_{j=q+1}^{T+1}\tilde{p}_j\cdot(\tilde{\beta}_{l-1,j}-\tilde{\beta}_{l,j})-\tilde{\xi}_q&= \tilde{\beta}_{l,q}+\sum_{j=q+2}^{T+1}\tilde{p}_j\cdot(\tilde{\beta}_{l-1,j}-\tilde{\beta}_{l,j})-\tilde{\xi}_q+\tilde{p}_{q+1}\cdot(\tilde{\beta}_{l-1,q+1}-\tilde{\beta}_{l,q+1})\nonumber\\
&=\beta^*_{l,q}+\sum_{j=q+1}^{T}\p_j\cdot(\beta^*_{l-1,j}-\beta^*_{l,j})-\xi^*_q+\p_q\cdot(1-\sigma)\cdot(\beta^*_{l-1,q}-\beta^*_{l,q})\nonumber\\
&\geq \p_q\cdot(1-\sigma)\cdot(\beta^*_{l-1,q}-\beta^*_{l,q})\nonumber
\end{align}
Thus, it is enough to show that $\beta^*_{l-1,q}\geq\beta^*_{l,q}$ to prove feasibility. Note that from Lemma \ref{duallemma1}, for each $l_2=1,2,\dots,k$ and each $l_1=1,2,\dots,l_2-1$, we have $\delta_{l_1,l_2}\geq\delta_{l_1+1,l_2}$, then, it is direct to show that $\beta^*_{l-1,q}\geq\beta^*_{l,q}$ from the construction \eqref{constructiondual}, which completes our proof.
\Halmos
\end{proof}

\subsection{Proof of Proposition \ref{Upperprophetproposition}}\label{Proofproposition3}
\begin{proof}{Proof:}
We consider the following problem instance $\mathcal{H}$. At the beginning, there are two queries arriving deterministically with a reward $1$. Then, over the time interval $[0,1]$, there are queries with reward $r_1>1$ arriving according to a Poisson process with rate $\lambda$. At last, there is one query with a reward $\frac{r_2}{\epsilon}$ arriving with a probability $\epsilon$ for some small $\epsilon>0$.

Obviously, since $r_1>1$ and $\epsilon$ is set to be small, the prophet will first serve the last query as long as it arrives, and then serve the queries with a reward $r_1$ as much as possible, and at least serve the first two queries. Then, we have that
\[
\mathbb{E}_{\bm{I}\sim\bm{F}}[V^{\text{off}}(\bm{I})]=\hat{V}:= r_2+2\cdot\exp(-\lambda)+(r_1+1)\cdot\lambda\cdot\exp(-\lambda)+2r_1\cdot(1-(\lambda+1)\cdot\exp(-\lambda)+O(\epsilon)
\]
Moreover, for any online algorithm $\pi$, we consider the following situations separately based on the number of the first two queries that $\pi$ will serve.\\
(i). If $\pi$ will always serve the first two queries, then it is obvious that $\mathbb{E}_{\pi, \bm{I}\sim\bm{F}}[V^\pi(\bm{I})]=2$.\\
(ii). If $\pi$ serves only one of the first two queries, then the optimal way for $\pi$ to serve the second query will depend on the value of $r_1$ and $r_2$. To be more specific, if $r_1\geq r_2$, then the optimal way is to serve the query with reward $r_1$ as long as it arrives, and if $r_1<r_2$, then the optimal way is to reject all the arriving queries with reward $r_1$ and only serve the last query. Thus, it holds that
\[
\mathbb{E}_{\pi, \bm{I}\sim\bm{F}}[V^\pi(\bm{I})]\leq V(1):= 1+\exp(-\lambda)\cdot r_2+(1-\exp(-\lambda))\cdot\max\{r_1, r_2\}+O(\epsilon)
\]
(iii). If $\pi$ rejects all the first two queries, then conditioning on there are more than one queries with reward $r_1$ arriving during the interval $[0,1]$, the optimal way for $\pi$ is to serve both queries with reward $r_1$ if $r_1\geq r_2$ and only serve one query with reward $r_1$ if $r_1< r_2$. Then, it holds that
\[
\mathbb{E}_{\pi, \bm{I}\sim\bm{F}}[V^\pi(\bm{I})]\leq V(2):=\exp(-\lambda)\cdot r_2+\lambda\cdot\exp(-\lambda)\cdot(r_1+r_2)+(1-(\lambda+1)\cdot\exp(-\lambda)\cdot(r_1+\max\{r_1, r_2\})
\]
Thus, we conclude that for any online algorithm $\pi$, it holds that
\[
\frac{\mathbb{E}_{\pi, \bm{I}\sim\bm{F}}[V^\pi(\bm{I})]}{\mathbb{E}_{\bm{I}\sim\bm{F}}[V^{\text{off}}(\bm{I})]}\leq g(r_1, r_2, \lambda):= \frac{\max\{V(1), V(2), 2\}}{\hat{V}}
\]
where we can neglect the $O(\epsilon)$ term by letting $\epsilon\rightarrow0$.
In this way, we can focus on the following optimization problem
\[
\inf_{r_1>1, r_2>1, \lambda} g(r_1, r_2, \lambda)
\]
to obtain the upper bound of the guarantee of any online algorithm relative to the prophet's value. We can numerically solve the above problem and show that when $r_1=r_2=1.4119$, $\lambda=1.2319$, the value of $g(r_1, r_2, \lambda)$ reaches its minimum and equals $0.6269$, which completes our proof.
\Halmos
\end{proof}

\subsection{Proof of Theorem \ref{worstcasetheorem}}\label{Prooftheorem4}
\begin{proof}{Proof:}
For each $\bm{\p}=(\p_1,\dots,\p_T)$ satisfying $\sum_{t=1}^{T}\p_t=k$, since each irrational number can be arbitrarily approximated by a rational number, we assume without loss of generality that $\p_t$ is a rational number for each $t$, i.e., $\p_t=\frac{n_t}{N}$ where $n_t$ is an integer for each $t$ and $N$ is an integer to denote the common denominator. We first split $\p_1$ into $\frac{1}{N}$ and $\frac{n_1-1}{N}$ to form a new sequence $\tilde{\bm{p}}=(\frac{1}{N},\frac{n_1-1}{N},\frac{n_2}{N},\dots,\frac{n_T}{N})$. From Lemma \ref{splititemlemma}, we know $\LCkunit(\bm{\p})\geq\LCkunit(\tilde{\bm{p}})$. We then split $\frac{n_1-1}{N}$ into $\frac{1}{N}$ and $\frac{n_1-2}{N}$ and so on. In this way, we split $\p_1$ into $n_1$ copies of $\frac{1}{N}$ to form a new sequence $\tilde{\bm{p}}=(\frac{1}{N},\dots,\frac{1}{N},\frac{n_2}{N},\dots,\frac{n_T}{N})$ and Lemma \ref{splititemlemma} guarantees that $\LCkunit(\bm{\p})\geq\LCkunit(\tilde{\bm{p}})$. We repeat the above operation for each $t$. Finally, we form a new sequence of arrival probabilities, denoted as $\bm{\p}^{N}=(\frac{1}{N},\dots,\frac{1}{N})\in\mathbb{R}^{Nk}$, and we have $\LCkunit(\bm{\p})\geq\LCkunit(\bm{\p}^N)$.

From the above argument, we know that for each $\bm{\p}=(\p_1,\dots,\p_T)$ satisfying $\sum_{t=1}^{T}\p_t=k$, there exists an integer $N$ such that $\LCkunit(\bm{\p})\geq\LCkunit(\bm{\p}^N)$, which implies that
\[
\inf_{\bm{\p}:\sum_{t}\p_t=k}~\LCkunit(\bm{\p})=\liminf_{N\rightarrow\infty} \LCkunit(\bm{\p}^N)
\]
Thus, it is enough to consider $\liminf_{N\rightarrow\infty} \LCkunit(\bm{\p}^N)$.

We denote $\tilde{\bm{y}}_{\theta}(t)=(\tilde{y}_{1,\theta}(t),\dots,\tilde{y}_{k,\theta}(t))$. We define a function $\bm{f}_{\theta}(\cdot)=(f_{1,\theta}(\cdot), \dots,f_{k,\theta}(\cdot))$, where we denote $\tilde{y}_{0,\theta}(t)=1$ and for each $l=1,\dots,k-1$
\[
f_{l,\theta}(\tilde{y}_{1,\theta},\dots,\tilde{y}_{k,\theta},t)=\left\{\begin{aligned}
&0,~~&&\text{if~}\tilde{y}_{l-1,\theta}(t)\leq1-\theta\\
&\tilde{y}_{l-1,\theta}(t)-(1-\theta),~~&&\text{if~}\tilde{y}_{l,\theta}(t)\leq1-\theta\leq\tilde{y}_{l-1,\theta}(t)\\
&\tilde{y}_{l-1,\theta}(t)-\tilde{y}_{l,\theta}(t),~~&&\text{if~}\tilde{y}_{l,\theta}(t)\geq1-\theta\\
\end{aligned}\right.
\]
and
\[
f_{k,\theta}(\tilde{y}_{1,\theta},\dots,\tilde{y}_{k,\theta},t)=\left\{\begin{aligned}
&0,~~&&\text{if~}\tilde{y}_{k-1,\theta}(t)\leq1-\theta\\
&\tilde{y}_{k-1,\theta}(t)-(1-\theta),~~&&\text{if~}\tilde{y}_{k-1,\theta}(t)\geq 1-\theta\\
\end{aligned}\right.
\]
Moreover, variable $(\tilde{y}_1,\dots,\tilde{y}_k,t)$ belongs to the feasible set of the function $f_{l,\theta}(\cdot)$ if and only if $y_{v-1}\geq y_{v}$ for $v=1,\dots,k-1$.
Then, for each $\theta\in[0,1]$, the function $\tilde{\bm{y}}_{\theta}(t)$ in Definition \ref{ODEdefinition} should be the solution to the following ordinary differential equation (ODE):
\begin{equation}\label{ODEformula}
\frac{d\tilde{\bm{y}}_{\theta}(t)}{dt}=\bm{f}_{\theta}(\tilde{\bm{y}}_{\theta}, t)\text{~~for~}t\in[0,k]\text{~with~starting~point~}\tilde{\bm{y}}_{\theta}(0)=(0,\dots,0)
\end{equation}
For each integer $N$ and $\bm{\p}^{N}=(\frac{1}{N},\dots,\frac{1}{N})$ where $\|\bm{\p}^{N}\|_1=k$, for any fixed $\theta\in[0,1]$, we denote $\{x_{l,t}(\theta,N)\}$ as the variables constructed in Definition \ref{constructprophet} under the arrival probabilities $\bm{\p}^{N}$, where $l=1,\dots,k$ and $t=1,\dots,Nk$. We further denote $y_{l,\theta,N}(\frac{t}{N})=\sum_{\tau=1}^{t}x_{l,\tau}(\theta,N)$ and denote $\bm{y}_{\theta,N}(\cdot)=(y_{1,\theta,N}(\cdot),\dots,y_{k,\theta,N}(\cdot))$. It is direct to check that for each $t=1,\dots,Nk$, it holds that
\[
(\bm{y}_{\theta,N}(\frac{t}{N})-\bm{y}_{\theta,N}(\frac{t-1}{N}))/(\frac{1}{N})=\bm{f}_{\theta}(\bm{y}_{\theta,N} ,\frac{t-1}{N})
\]
Thus, $\{\bm{y}_{\theta,N}(t)\}_{\forall t\in[0,k]}$ can be viewed as the result obtained from applying Euler's method \citep{butcher2008numerical} to solve ODE \eqref{ODEformula}, where there are $Nk$ discrete points uniformly distributed within $[0,k]$. Note that for each $\theta\in[0,1]$, the function $\bm{f}_{\theta}(\cdot)$ is Lipschitz continuous with a Lipschitz constant $2$ under infinity norm. Moreover, it is direct to note that for each $\theta\in[0,1]$ and each $t\in[0,k]$, it holds that $\|\bm{f}_{\theta}(\tilde{\bm{y}},t)\|_{\infty}\leq1$. Then, for each $\theta\in[0,1]$, each $t_1, t_2\in[0,k]$ and each $l=1,\dots,k$, we have
\[
|\frac{d\tilde{y}_{l,\theta}(t_1)}{dt}-\frac{d\tilde{y}_{l,\theta}(t_2)}{dt}|\leq2\cdot\|\tilde{\bm{y}}_{\theta}(t_1)-\tilde{\bm{y}}_{\theta}(t_2)\|_{\infty}\leq 2\cdot|t_1-t_2|
\]
Thus, we know that
\[
|\tilde{y}_{l,\theta}(t_1)-\tilde{y}_{l,\theta}(t_2)-\frac{d\tilde{y}_{l,\theta}(t_2)}{dt}\cdot(t_1-t_2)|\leq 2\cdot (t_1-t_2)^2
\]
We can apply the global truncation error of Euler's method (Theorem 212A \citep{butcher2008numerical}) to show that $\bm{y}_{\theta,N}(k)$ converges to $\tilde{\bm{y}}_{\theta}(k)$ when $N\rightarrow\infty$. Specifically, we have
\begin{equation}\label{proofeulerconvergence}
\|\bm{y}_{\theta,N}(k)-\tilde{\bm{y}}_{\theta}(k)\|_{\infty}\leq(\exp(2k)-1)\cdot\frac{1}{N},~~\forall \theta\in[0,1]
\end{equation}
Now we define $Y(\theta)=\tilde{y}_{k,\theta}(k)$ as a function of $\theta\in[0,1]$ and for each $N$, we define $Y_N(\theta)=y_{k,\theta,N}(k)$ as a function of $\theta\in[0,1]$. \eqref{proofeulerconvergence} implies that the function sequence $\{Y_N\}_{\forall N}$ converges uniformly to the function $Y$ when $N\rightarrow\infty$. Note that for each $N$, the function $Y_N(\theta)$ is continuously monotone increasing with $\theta$ due to Lemma \ref{feasilemma3}, then from uniform limit theorem, $Y(\theta)$ must be a continuously monotone increasing function over $\theta$. Thus, the equation $Y(\theta)=1-\theta$ has a unique solution, denoted as $\sgk$. For each $N$, we denote $\theta^*_N$ as the unique solution to the equation $Y_N(\theta)=1-\theta$, where we have that $\theta^*_N=\LCkunit(\bm{\p}^N)$. Since $\{Y_N\}_{\forall N}$ converges uniformly to the function $Y$, it must hold that $\sgk=\lim_{N\rightarrow\infty}\theta^*_N$, which completes our proof.
\Halmos
\end{proof}

\section{Proofs in \Cref{randomsizesection}}

\subsection{Proof of \Cref{Largesmallproposition}}\label{pfprop:largesmall}
\begin{proof}{Proof:}
The proof is the same as the proof of Proposition 3.1 in \citet{jiang2022tight}.

Consider a problem setup ${\mathcal{H}}$ with 4 queries and
\[
({r}_1, \p_1, d_1)=(r,1,\epsilon),~~({r}_2, \p_2, d_2)=({r}_3, \p_3, d_3)=(r, \frac{1-2\epsilon}{1+2\epsilon}, \frac{1}{2}+\epsilon),~~({r}_4, \p_4, d_4)=(r/\epsilon,\epsilon,1)
\]
for $r>0$ and some $\epsilon>0$. Obviously, if the policy $\pi$ only serves queries with a size greater than 1/2, then the expected total reward is $V^{\pi}_L=r$. If the policy $\pi$ only serves queries with a size no greater than 1/2, then the expected total reward is $V^{\pi}_S=r$. Thus, the expected total reward of the policy $\pi$ is
\[
V^{\pi}=\max\{V^{\pi}_L, V^{\pi}_S\}=r+O(\epsilon)
\]
Moreover, it is direct to see that $\sum_{t=1}^{4}\p_t\cdot d_t=1$, then, we have $\text{UP}({\mathcal{H}})=4r$.
Thus, the guarantee of $\pi$ is upper bounded by $1/4+O(\epsilon)$, which converges to $1/4$ as $\epsilon\rightarrow0$.
\Halmos
\end{proof}

\subsection{Proof of \Cref{knapsackratiotheorem}}\label{Prooftheorem2}
\begin{proof}{Proof:}
It is enough to prove that the Best-fit Magician policy $\pi_{\gamma}$ in \Cref{definedistributionrandomsize} is feasible when $\gamma=\frac{1}{3+e^{-2}}$. In the remaining proof, we set $\gamma=\frac{1}{3+e^{-2}}$. For a fixed $t$, and any $a$ and $b$, denote $\mu_{t,\gamma}(a,b]=P(a<\tilde{X}_{t,\gamma}\leq b)$ assuming $\tilde{X}_{t,\gamma}$ is well-defined, it is enough to prove that $\mu_{t,\gamma}(0,1]\leq1-\gamma$ thus the random variable $\tilde{X}_{t+1,\gamma}$ is well-defined.

We define $U_t(s)=\mu_{t,\gamma}(0,s]$ for any $s\in(0,1]$. Note that by definition, we have $\mathbb{E}[\tilde{X}_{t,\gamma}]=\gamma\cdot\sum_{\tau=1}^{t}\p_\tau\cdot d_\tau\leq\gamma$. From integration by parts, we have that
\begin{equation}\label{integrationbypart}
\gamma\geq\mathbb{E}[\tilde{X}_{t,\gamma}]=\int_{s=0}^{1}sdU_t(s)=U_t(1)-\int_{s=0}^{1}U_t(s)ds
\end{equation}
We then bound the term $\int_{s=0}^{1}U_t(s)ds$. Now suppose $U_t(1)>\gamma$, otherwise $U_t(1)\leq\gamma$ immediately implies that $U_t(1)\leq1-\gamma$, which proves our result. Then there must exist a constant $u^*\in (0,1)$ such that $\gamma\cdot u^*-\gamma\cdot\ln(u^*)=U_t(1)$. We further define
\[
s^*=\left\{\begin{aligned}
&\min\{ s\in (0,1/2]: U_t(s)\geq \gamma\cdot u^* \},~~&\text{if~}U_t(\frac{1}{2})\geq \gamma\cdot u^*\\
&\frac{1}{2},&~~\text{if~}U_t(\frac{1}{2})<\gamma\cdot u^*
\end{aligned}\right.\]
Denote $U_t(s^*-)=\lim_{s\rightarrow s^*-}U_t(s)$, it holds that
\[\begin{aligned}
\int_{s=0}^{1}U_t(s)ds&=\int_{s=0}^{s^*-}U_t(s)ds+\int_{s=s^*}^{1/2}U_t(s)ds+\int_{s=1/2}^{1-s^*}U_t(s)ds+\int_{s=1-s^*}^{1}U_t(s)ds\\
&\leq s^*\cdot (U_t(s^*-)+U_t(1))+\int_{s=s^*}^{1/2}U_t(s)ds+\int_{s=1/2}^{1-s^*}U_t(s)ds\\
&\leq s^*\cdot(2\gamma u^*-\gamma\cdot\ln(u^*))+\int_{s=s^*}^{1/2}(2U_t(s)-\gamma\cdot\ln(\frac{U_t(s)}{\gamma}))ds
\end{aligned}\]
where the last inequality holds by noting that $U_t(s^*-)\leq\gamma u^*$ and for any $s\in[s^*,1/2]$, from Lemma \ref{generalsamplepathlemma}, we have $\frac{U_t(s)}{\gamma}\leq\exp(-\frac{U_t(1-s)-U_t(s)}{\gamma})$, which implies that $\frac{U_t(1-s)}{\gamma}\leq\frac{U_t(s)}{\gamma}-\ln(\frac{U_t(s)}{\gamma})$. Note that for any $s\in[s^*,1/2]$, we have that $\gamma\cdot u^*\leq U_t(s^*)\leq U_t(s)\leq U_t(1/2)\leq\gamma$, where $U_t(1/2)\leq\gamma$ holds directly from Lemma \ref{generalsamplepathlemma}. Further note that the function $2x-\gamma\cdot\ln(x/\gamma)$ is a convex function, thus is quasi convex. Then, for any $s\in[s^*,1/2]$, it holds that
\[
2U_t(s)-\gamma\cdot\ln(\frac{U_t(s)}{\gamma})\leq\max\{ 2\gamma u^*-\gamma\cdot\ln(u^*), 2\gamma \}
\]
Thus, we have that
\[\begin{aligned}
\int_{s=0}^{1}U_t(s)ds&\leq s^*\cdot(2\gamma u^*-\gamma\cdot\ln(u^*))+(1/2-s^*)\cdot \max\{ 2\gamma u^*-\gamma\cdot\ln(u^*), 2\gamma \}\\
\end{aligned}\]
If $2\gamma u^*-\gamma\cdot\ln(u^*)\leq 2\gamma$, we have $\int_{s=0}^{1}U_t(s)ds\leq 2s^*\gamma+\gamma-2s^*\gamma=\gamma$. From \eqref{integrationbypart}, we have that $U_t(1)\leq2\gamma<1-\gamma$.\\
If $2\gamma u^*-\gamma\cdot\ln(u^*)> 2\gamma$, we have $\int_{s=0}^{1}U_t(s)ds\leq\gamma u^*-\frac{\gamma}{2}\cdot\ln(u^*)$. From \eqref{integrationbypart} and the definition of $u^*$, we have that
\[
U_t(1)=\gamma u^*-\gamma\cdot\ln(u^*)\leq \gamma+\gamma u^*-\frac{\gamma}{2}\cdot\ln(u^*)
\]
which implies that $u^*\geq \exp(-2)$. Note that the function $x-\ln(x)$ is non-increasing on $(0,1)$, we have $U_t(1)\leq \gamma\cdot\exp(-2)+2\gamma=1-\gamma$, which completes our proof.
\Halmos
\end{proof}

\subsection{Proof of \Cref{uppertheorem}}\label{pf:uppertheorem}
\begin{proof}{Proof:}
We denote by $d_t$ the size of query $t$. Then, we have $\bm{p}=(p_1,\dots,p_t)$ and $\bm{D}=(d_1,\dots,d_T)$.
For each $\epsilon>0$, we consider the following $\bm{\p}$ and $\bm{D}$:
\[\begin{aligned}
(\p_1, d_1)=(1, \epsilon),~~(\p_t, d_t)=(\frac{1-2\epsilon}{(T-2)(\frac{1}{2}+\epsilon)}, \frac{1}{2}+\epsilon)\text{~~for~all~}2\leq t\leq T-1\text{~~and~~}(\p_T,d_T)=(\epsilon, 1)
\end{aligned}\]
It is direct to check that $\sum_{t=1}^{T}\p_t\cdot d_t=1$.

We denote by $\theta^*$ the maximum ratio in the knapsack OCRS problem. We now focus on the last query $(\p_T,d_T)=(\epsilon, 1)$. Note that in order to accept this last query with probability $\theta^*$, we must not accept any query during the period $1$ to period $T-1$, with probability at least $\theta^*$. Therefore, it holds that
\[\begin{aligned}
\theta^*&\leq 1-P(\text{accept some query $t\leq T-1$})\\
&\leq 1-P(\text{accept query }1 \text{ and all queries }2\leq t\leq T-1 \text{ are inactive})-P(\text{accept some query $2\leq t\leq T-1$}).
\end{aligned}\]
Note that we can bound 
\[\begin{aligned}
P(\text{accept query }1 \text{ and all queries }2\leq t\leq T-1 \text{ are inactive})&=\theta^*\cdot \left(1-\frac{1-2\epsilon}{(T-2)(\frac{1}{2}+\epsilon)} \right)^{T-2}\\
&=\theta^*\cdot e^{-2}+O(\epsilon).
\end{aligned}\]
On the other hand, we know that
\[
P(\text{accept some query $2\leq t\leq T-1$})=\theta^*\cdot\sum_{t=2}^{T-1}p_t=2\theta^*+O(\epsilon).
\]
Therefore, when $\epsilon\rightarrow0$, the optimal value $\theta^*$ must satisfy the inequality
\[
\theta^*\leq 1-\theta^*\cdot e^{-2}-2\theta^*
\]
which implies that $\theta^*\leq\frac{1}{3+e^{-2}}$. Our proof is thus completed.

\Halmos
\end{proof}

\section{Proofs in \cref{sec:extensions}}

\subsection{Proof of Lemma \ref{UDsamplepathlemma}}\label{proofUDlemma}
\begin{proof}{Proof:}
We prove \eqref{UDsamplepathinequality} by induction on $t$. When $t=0$, since $\mu_{0,\gamma}(0,b]=0$ for any $0<b\leq1/2$, \eqref{UDsamplepathinequality} holds trivially. Now suppose that \eqref{UDsamplepathinequality} holds for $t-1$, we consider the case for $t$. Denote $\mathcal{F}_t$ as the support of $\tilde{d}_t$ and for each $ d_t\in\mathcal{F}_t$, we denote $\eta_{t,\boldsymbol{\gamma}}(d_t)$ as the threshold defined in \eqref{defineetaUD}. Then we define the following division of $\mathcal{F}_t$:
\[\begin{aligned}
&\mathcal{F}_{t,1}:=\{d_t\in\mathcal{F}_{t}: \eta_{t,\boldsymbol{\gamma}}(d_t)=0 \text{~and~} d_t\leq b\}\\
&\mathcal{F}_{t,2}=\{d_t\in\mathcal{F}_{t}:\eta_{t,\boldsymbol{\gamma}}(d_t)=0 \text{~and~}  b< d_t\leq 1-b\}\\
&\mathcal{F}_{t,3}=\{d_t\in\mathcal{F}_{t}:\eta_{t,\boldsymbol{\gamma}}(d_t)>0 \text{~and~} d_t\leq 1- b\}\\
&\mathcal{F}_{t,4}=\{d_t\in\mathcal{F}_{t}:\eta_{t,\boldsymbol{\gamma}}(d_t)=0 \text{~and~} 1- b<d_t\}\\
&\mathcal{F}_{t,5}=\{d_t\in\mathcal{F}_{t}:\eta_{t,\boldsymbol{\gamma}}(d_t)>0 \text{~and~}  1- b<d_t\}
\end{aligned}\]
Note that for each $d_t\in\mathcal{F}_t$, $\eta_{t,\boldsymbol{\gamma}}(d_t)=0$ implies that a measure $p_t(d_t)\cdot (\gamma_t-\mu_{t-1,\boldsymbol{\gamma}}(0,1-d_t])$ of empty sample paths will be moved to $d_t$ due to the inclusion of realization $d_t$ when defining $\tilde{X}_{t,\boldsymbol{\gamma}}$. More specifically, the movement of sample paths due to the inclusion of each realization $d_t\in\mathcal{F}_t$ can be described as follows:\\
(i). For each $d_t\in\mathcal{F}_{t,1}$, obviously, $p_t(d_t)\cdot (\gamma_t-\mu_{t-1,\boldsymbol{\gamma}}(0,1-d_t])$ measure of sample paths, which is upper bounded by $p_t( d_t)\cdot (\gamma_t-\mu_{t-1,\boldsymbol{\gamma}}(0,1-b])$, will be moved from $0$ to the range $(0,b]$, while a quantity $a_1(d_t)\leq p_t(d_t)\cdot\mu_{t-1,\boldsymbol{\gamma}}(0,b]$ measure of sample paths will be moved out of the range $(0,b]$. Moreover, at most $p_t(d_t)\cdot\mu_{t-1,\boldsymbol{\gamma}}(0,b]$ measure of sample paths will be moved into the range $(b,1-b]$. \\
(ii). For each $d_t\in\mathcal{F}_{t,2}$, $p_t(d_t)\cdot\mu_{t-1,\boldsymbol{\gamma}}(0,b]$ measure of sample paths will be moved out of the range $(0,b]$. Moreover, $p_t(d_t)\cdot (\gamma_t-\mu_{t-1,\boldsymbol{\gamma}}(0,1-d_t])$ measure of sample paths, which is upper bounded by $p_t(d_t)\cdot (\gamma_t-\mu_{t-1,\boldsymbol{\gamma}}(0,b])$, will be moved from $0$ into the range $(b,1-b]$, while at most $p_t(d_t)\cdot\mu_{t-1,\boldsymbol{\gamma}}(0,b]$ measure of sample paths will be moved from $(0,b]$ into $(b, 1-b]$. Thus, the measure of new sample path that is moved into the range $(b,1-b]$ is upper bounded by $\gamma_t\cdot p_t(d_t)$.\\
(iii). For each $d_t\in\mathcal{F}_{t,3}$, then a quantity $a_3(d_t)\leq p_t(d_t)\cdot\mu_{t-1,\boldsymbol{\gamma}}(0,b]$ measure of sample paths is moved out of the range $(0,b]$, and at most $p_t(d_t)\cdot\mu_{t-1,\boldsymbol{\gamma}}(0,b]$ measure of sample paths is moved into the range $(b,1-b]$.\\
(iv). For each $d_t\in\mathcal{F}_{t,4}$ or $d_t\in\mathcal{F}_{t,5}$, since $d_t>1-b$, obviously, no new sample path will be added to the range $(b,1-b]$ due to the inclusion of such realization $d_t$ when defining $\tilde{X}_{t,\boldsymbol{\gamma}}$, while the measure of the sample paths within the range $(0,b]$ can only become smaller.\\
To conclude, denoting
\[
a_1=\sum_{d_t\in\mathcal{F}_{t,1}} a_1(d_t) \text{~and~}
\hat{p}_1= \sum_{d_t\in\mathcal{F}_{t,1}}p_t(d_t)\text{~and~}\hat{p}_2=\sum_{d_t\in\mathcal{F}_{t,2}}p_t(d_t)\text{~and~}a_3=\sum_{d_t\in\mathcal{F}_{t,3}} a_3(d_t) \text{~and~}\hat{p}_3=\sum_{d_t\in\mathcal{F}_{t,3}}p_t(d_t)
\]
we have that
\begin{equation}\label{UDrandmovemu1}
\begin{aligned}
\mu_{t,\boldsymbol{\gamma}}(0,b]\leq&\mu_{t-1,\boldsymbol{\gamma}}(0,b]+(\gamma_t-\mu_{t-1,\boldsymbol{\gamma}}(0,1-b])\cdot\hat{p}_1-a_1-\mu_{t-1,\boldsymbol{\gamma}}(0,b]\cdot \hat{p}_2-a_3
\end{aligned}
\end{equation}
and
\begin{equation}\label{UDrandmovemu2}
\mu_{t,\boldsymbol{\gamma}}(b,1-b]\leq \mu_{t-1,\boldsymbol{\gamma}}(b,1-b]+a_1+\gamma_t\cdot\hat{p}_2+a_3
\end{equation}
Moreover, it holds that $\hat{p}_1+\hat{p}_2+\hat{p}_3\leq1$. We now consider the following two cases separately.\\
\textit{Case 1:} If $\hat{p}_1>0$, then we must have $\gamma_t\geq\mu_{t-1,\boldsymbol{\gamma}}(0,1-b]$. Notice that $\hat{p}_1\leq1-\hat{p}_2$, from \eqref{UDrandmovemu1},  we have
\begin{align}
\mu_{t,\boldsymbol{\gamma}}(0,b]&\leq\mu_{t-1,\boldsymbol{\gamma}}(0,b]+(\gamma_t-\mu_{t-1,\boldsymbol{\gamma}}(0,1-b])\cdot\hat{p}_1
-a_1-\mu_{t-1,\boldsymbol{\gamma}}(0,b]\cdot \hat{p}_2-a_3\nonumber\\
&\leq \mu_{t-1,\boldsymbol{\gamma}}(0,b]+(\gamma_t-\mu_{t-1,\boldsymbol{\gamma}}(0,1-b])\cdot(1-\hat{p}_2)-a_1-\mu_{t-1,\boldsymbol{\gamma}}(0,b]\cdot \hat{p}_2-a_3\nonumber\\
&= (\gamma_t-\mu_{t-1,\boldsymbol{\gamma}}(b,1-b])\cdot(1-\hat{p}_2)-a_1-a_3\nonumber\\
&\leq (\gamma_1-\mu_{t-1,\boldsymbol{\gamma}}(b,1-b])\cdot(1-\hat{p}_2)-a_1-a_3\label{UDrandmovemu6}
\end{align}
where the last inequality holds from $\gamma_1\geq\gamma_t$.
Moreover, from \eqref{UDrandmovemu2}, we have that
\begin{align}
\exp(-\frac{1}{\gamma_1}\cdot\mu_{t,\boldsymbol{\gamma}}(b,1-b])&\geq \exp(-\frac{1}{\gamma_1}\cdot\mu_{t-1,\boldsymbol{\gamma}}(b,1-b]-\frac{\gamma_t\hat{p}_2}{\gamma_1})\cdot\exp(-\frac{1}{\gamma_1}\cdot a_1-\frac{1}{\gamma_1}\cdot a_3)\nonumber\\
&\geq \exp(-\frac{1}{\gamma_1}\cdot\mu_{t-1,\boldsymbol{\gamma}}(b,1-b]-\hat{p}_2)\cdot\exp(-\frac{1}{\gamma_1}\cdot a_1-\frac{1}{\gamma_1}\cdot a_3)\nonumber\\
&\geq \exp(-\frac{1}{\gamma_1}\cdot\mu_{t-1,\boldsymbol{\gamma}}(b,1-b]-\hat{p}_2)\cdot (1-\frac{1}{\gamma_1}\cdot a_1-\frac{1}{\gamma_1}\cdot a_3)\nonumber\\
&=\exp(-\frac{1}{\gamma_1}\cdot\mu_{t-1,\boldsymbol{\gamma}}(b,1-b]-\hat{p}_2)\nonumber\\
&~~~-\exp(-\frac{1}{\gamma_1}\cdot\mu_{t-1,\boldsymbol{\gamma}}(b,1-b]-\hat{p}_2)\cdot (\frac{1}{\gamma_1}\cdot a_1+\frac{1}{\gamma_1}\cdot a_3)\nonumber\\
&\geq \exp(-\frac{1}{\gamma_1}\cdot\mu_{t-1,\boldsymbol{\gamma}}(b,1-b]-\hat{p}_2)-\frac{1}{\gamma_1}\cdot a_1-\frac{1}{\gamma_1}\cdot a_3\label{UDrandmovemu7}
\end{align}
where the second inequality holds from $\gamma_1\geq\gamma_t$, the third inequality holds from $\exp(-x)\geq1-x$ for any $x\geq0$ and the last inequality holds from $\exp(-x)\leq1$ for any $x\geq0$.
Further note that
\[
\exp(-\frac{1}{\gamma_1}\cdot\mu_{t-1,\boldsymbol{\gamma}}(b,1-b]-\hat{p}_2)=\exp(-\frac{1}{\gamma_1}\cdot\mu_{t-1,\boldsymbol{\gamma}}(b,1-b])
\cdot\exp(-\hat{p}_2)\geq(1-\frac{1}{\gamma_1}\cdot\mu_{t-1,\boldsymbol{\gamma}}(b,1-b])\cdot(1-\hat{p}_2)
\]
From \eqref{UDrandmovemu6} and \eqref{UDrandmovemu7}, we have
\[\begin{aligned}
\exp(-\frac{1}{\gamma_1}\cdot\mu_{t,\boldsymbol{\gamma}}(b,1-b])&\geq (1-\frac{1}{\gamma_1}\cdot\mu_{t-1,\boldsymbol{\gamma}}(b,1-b])\cdot(1-\hat{p}_2)-\frac{1}{\gamma_1}\cdot a_1-\frac{1}{\gamma_1}\cdot a_3\\
&\geq \frac{1}{\gamma_1}\cdot\mu_{t,\boldsymbol{\gamma}}(0,b]
\end{aligned}\]
\textit{Case 2:} If $\hat{p}_1=0$ which also implies $a_1=0$, then we have
\begin{equation}\label{UDrandmovemu10}
\mu_{t,\boldsymbol{\gamma}}(0,b]\leq\mu_{t-1,\boldsymbol{\gamma}}(0,b]-\mu_{t-1,\boldsymbol{\gamma}}(0,b]\cdot \hat{p}_2-a_3
\end{equation}
and
\[
\mu_{t,\boldsymbol{\gamma}}(b,1-b]\leq \mu_{t-1,\boldsymbol{\gamma}}(b,1-b]+\gamma_t\cdot\hat{p}_2+a_3\leq \mu_{t-1,\boldsymbol{\gamma}}(b,1-b]+\gamma_1\cdot\hat{p}_2+a_3
\]
Thus, it holds that
\begin{align}
\exp(-\frac{1}{\gamma_1}\cdot \mu_{t,\boldsymbol{\gamma}}(b,1-b])&\geq \exp(-\frac{1}{\gamma_1}\cdot\mu_{t-1,\boldsymbol{\gamma}}(b,1-b]-\hat{p}_2 )\cdot \exp(-\frac{1}{\gamma_1}\cdot a_3)\nonumber\\
&\geq \exp(-\frac{1}{\gamma_1}\cdot\mu_{t-1,\boldsymbol{\gamma}}(b,1-b]-\hat{p}_2 )\cdot(1-\frac{1}{\gamma_1}\cdot a_3)\nonumber\\
&\geq \exp(-\frac{1}{\gamma_1}\cdot\mu_{t-1,\boldsymbol{\gamma}}(b,1-b]-\hat{p}_2 )-\frac{1}{\gamma_1}\cdot a_3\nonumber\\
&\geq \exp(-\frac{1}{\gamma_1}\cdot\mu_{t-1,\boldsymbol{\gamma}}(b,1-b])\cdot(1-\hat{p}_2)-\frac{1}{\gamma_1}\cdot a_3\nonumber\\
&\geq \frac{1}{\gamma_1}\cdot\mu_{t-1,\boldsymbol{\gamma}}(0,b]\cdot(1-\hat{p}_2)-\frac{1}{\gamma_1}\cdot a_3\label{UDrandmovemu11}
\end{align}
where the third inequality holds from $\exp(-a)\leq1$ for any $a\geq0$ and the last inequality holds from induction hypothesis. Our proof is completed immediately by combining \eqref{UDrandmovemu10} and \eqref{UDrandmovemu11}.
\Halmos
\end{proof}

\subsection{Proof of Theorem \ref{UDemptyboundtheorem}}\label{proofUDemptytheorem}
\begin{proof}{Proof:}
For each fixed $t$, we define $U_t(s)=\mu_{t,\boldsymbol{\gamma}}(0,s]=P(0<\tilde{X}_{t,\boldsymbol{\gamma}}\leq s)$ for any $s\in(0,1]$. Note that by \Cref{UDpolicy}, we have $\mathbb{E}[\tilde{X}_{t,\boldsymbol{\gamma}}]=\sum_{\tau=1}^{t}\gamma_{\tau}\cdot \psi_\tau$. From integration by parts, we have that
\begin{equation}\label{UDintegrationbypart}
\sum_{\tau=1}^{t}\gamma_{\tau}\cdot \psi_\tau=\mathbb{E}[\tilde{X}_{t,\boldsymbol{\gamma}}]=\int_{s=0}^{1}sdU_t(s)=U_t(1)-\int_{s=0}^{1}U_t(s)ds
\end{equation}
We then bound the term $\int_{s=0}^{1}U_t(s)ds$. If $U_t(1)\leq\gamma_1$, then we immediately have
\[
P(\tilde{X}_{t,\boldsymbol{\gamma}}=0)\geq1-\gamma_1\geq 1-\gamma_1-\sum_{\tau=1}^{t}\gamma_{\tau}\cdot\psi_{\tau}
\]
which proves \eqref{UDbound1}.
Thus, in the remaining part of the proof, it is enough for us to only focus on the case $U_t(1)>\gamma_1$.

If $U_t(1)>\gamma_1$,
then there must exists a constant $u^*\in (0,1)$ such that
\[
\gamma_1\cdot u^*-\gamma_1\cdot\ln(u^*)=U_t(1).
\]
We further define
\[
s^*=\left\{\begin{aligned}
&\min\{ s\in (0,1/2]: U_t(s)\geq \gamma_1\cdot u^* \},~~&\text{if~}U_t(\frac{1}{2})\geq \gamma_1\cdot u^*\\
&\frac{1}{2},&~~\text{if~}U_t(\frac{1}{2})<\gamma_1\cdot u^*
\end{aligned}\right.\]
Following the proof of \Cref{knapsackratiotheorem}, we can show that
\[\begin{aligned}
\int_{s=0}^{1}U_t(s)ds&\leq s^*\cdot(2\gamma_1\cdot u^*-\gamma_1\cdot\ln(u^*))+(1/2-s^*)\cdot \max\{ 2\gamma_1\cdot u^*-\gamma_1\cdot\ln(u^*), 2\gamma_1 \}\\
\end{aligned}\]
We further simplify the above expression separately by comparing the value of $2\gamma_1\cdot u^*-\gamma_1\cdot\ln(u^*)$ and $2\gamma_1$.\\
\textit{Case 1}: If $2\gamma_1\cdot u^*-\gamma_1\cdot\ln(u^*)\leq 2\gamma_1$, we have $\int_{s=0}^{1}U_t(s)ds\leq 2s^*\gamma_1+\gamma_1-2s^*\gamma_1=\gamma_1$. From \eqref{UDintegrationbypart}, we have that
\[
U_t(1)\leq\gamma_1+\sum_{\tau=1}^{t}\gamma_{\tau}\cdot \psi_\tau
\]
\textit{Case 2}: If $2\gamma_1\cdot u^*-\gamma_1\cdot\ln(u^*)> 2\gamma_1$, we have $\int_{s=0}^{1}U_t(s)ds\leq\gamma_1\cdot u^*-\frac{\gamma_1}{2}\cdot\ln(u^*)$. From \eqref{UDintegrationbypart} and the definition of $u^*$, we have that
\[
U_t(1)=\gamma_1\cdot u^*-\gamma_1\cdot\ln(u^*)\leq \sum_{\tau=1}^{t}\gamma_{\tau}\cdot \psi_\tau+\gamma_1\cdot u^*-\frac{\gamma_1}{2}\cdot\ln(u^*)
\]
The above inequality implies that
\[
u^*\geq\exp(-\frac{2}{\gamma_1}\cdot\sum_{\tau=1}^{t}\gamma_{\tau}\cdot \psi_\tau)
\]
Note that the function $x-\ln(x)$ is non-increasing on $(0,1)$, hence we have
\[
U_t(1)\leq 2\cdot\sum_{\tau=1}^{t}\gamma_{\tau}\cdot \psi_\tau+\gamma_1\cdot \exp(-\frac{2}{\gamma_1}\cdot\sum_{\tau=1}^{t}\gamma_{\tau}\cdot \psi_\tau)
\]
Combing the above two cases, we conclude that
\[
U_t(1)\leq\max\{ \gamma_1+\sum_{\tau=1}^{t}\gamma_{\tau}\cdot \psi_\tau,~~~ 2\cdot\sum_{\tau=1}^{t}\gamma_{\tau}\cdot \psi_\tau+\gamma_1\cdot \exp(-\frac{2}{\gamma_1}\cdot\sum_{\tau=1}^{t}\gamma_{\tau}\cdot \psi_\tau) \}
\]
Note that $P(\tilde{X}_{t,\boldsymbol{\gamma}}=0)=1-U_t(1)$, we conclude that
\[
P(\tilde{X}_{t,\boldsymbol{\gamma}}=0)\geq\min\{1-\gamma_1-\sum_{\tau=1}^{t}\gamma_{\tau}\cdot\psi_{\tau},~~1-2\cdot\sum_{\tau=1}^{t}\gamma_{\tau}\cdot \psi_\tau-\gamma_1\cdot \exp(-\frac{2}{\gamma_1}\cdot\sum_{\tau=1}^{t}\gamma_{\tau}\cdot \psi_\tau) \}
\]
which completes our proof.
\Halmos
\end{proof}

\subsection{Proof of Lemma \ref{UDoconstgammalemma}}\label{UDproofconstlemma}
\begin{proof}{Proof:}
Since the function $h_{\gamma_0}(\cdot)$ is non-increasing and non-negative over $[0,1]$, it is direct to see that
\[
1\geq\hat{\gamma}_1\geq\dots\geq\hat{\gamma}_T\geq0
\]
Note that for each $t=1,\dots,T$, we have
\[
\int_{\tau=0}^{k_t}h_{\gamma_0}(\tau)d\tau=\sum_{\tau=1}^{t}\hat{\gamma}_{\tau}\cdot\psi_{\tau}
\]
and $\gamma_0\geq\hat{\gamma}_1$. Then, for each $t=1,\dots,T-1$ and each $\tau\in [k_t, k_{t+1}]$, it holds that
\[
h_{\gamma_0}(\tau)\leq 1-\gamma_0-\int_{\tau'=0}^{\tau}h_{\gamma_0}(\tau')d\tau'\leq 1-\gamma_0-\int_{\tau'=0}^{k_t}h_{\gamma_0}(\tau')d\tau'\leq1-\hat{\gamma}_1-\sum_{\tau'=1}^{t}\hat{\gamma}_{\tau'}\cdot\psi_{\tau'}
\]
which implies that
\[
\hat{\gamma}_{t+1}\leq 1-\hat{\gamma}_1-\sum_{\tau'=1}^{t}\hat{\gamma}_{\tau'}\cdot\psi_{\tau'}
\]
since $\hat{\gamma}_{t+1}$ is defined as the average of function $h_{\gamma_0}(\cdot)$ over $[k_t, k_{t+1}]$ in \eqref{UDfeasiblegamma}.

Similarly, note that the function $2x+\gamma_0\cdot\exp(-\frac{2}{\gamma_0}\cdot x)$ is monotone increasing when $x\geq0$. Then, for each $t=1,\dots,T-1$ and each $\tau\in [k_t, k_{t+1}]$, we have
\[\begin{aligned}
h_{\gamma_0}(\tau)&\leq 1-2\cdot\int_{\tau'=0}^{\tau}h_{\gamma_0}(\tau')d\tau'-\gamma_0\cdot\exp(-\frac{2}{\gamma_0}\cdot\int_{\tau'=0}^{\tau}h_{\gamma_0}(\tau')d\tau'  ) \\
&\leq 1-2\cdot\int_{\tau'=0}^{k_t}h_{\gamma_0}(\tau')d\tau'-\gamma_0\cdot\exp(-\frac{2}{\gamma_0}\cdot\int_{\tau'=0}^{k_t}h_{\gamma_0}(\tau')d\tau'  )\\
&=1-2\cdot \sum_{\tau'=1}^{t}\hat{\gamma}_{\tau'}\cdot\psi_{\tau'}-\gamma_0\cdot\exp(-\frac{2}{\gamma_0}\cdot\sum_{\tau'=1}^{t}\hat{\gamma}_{\tau'}\cdot\psi_{\tau'})
\end{aligned}\]
which implies that
\[
\hat{\gamma}_{t+1}\leq1-2\cdot \sum_{\tau'=1}^{t}\hat{\gamma}_{\tau'}\cdot\psi_{\tau'}-\gamma_0\cdot\exp(-\frac{2}{\gamma_0}\cdot\sum_{\tau'=1}^{t}\hat{\gamma}_{\tau'}\cdot\psi_{\tau'})
\]
since $\hat{\gamma}_{t+1}$ is defined as the average of function $h_{\gamma_0}(\cdot)$ over $[k_t, k_{t+1}]$ in \eqref{UDfeasiblegamma}. Thus, we conclude that $\{\hat{\gamma}_t\}_{t=1}^T$ is a feasible solution to $\text{OP}(\boldsymbol{\psi})$.
\Halmos
\end{proof}

\subsection{Proof of \Cref{Upperunitdensity}}\label{proofUpperunitdensity}

It is enough for us to consider a problem setup $\mathcal{H}$ with $T$ queries, where each query has a deterministic size $\frac{1}{2}+\frac{1}{T}$ and is active with probability $\frac{2}{T}$. It is clear that $\textbf{UP}(\mathcal{H})=1$. However, any online algorithm $\pi$ can serve at most one query, given at least one query has arrived. Then, the expected capacity utilization of any online algorithm $\pi$ is upper bound by
\[
(\frac{1}{2}+\frac{1}{T})\cdot (1-(1-\frac{2}{T})^T)=\frac{1-e^{-2}}{2}+O(\frac{1}{T})
\]
This implies an upper bound $\frac{1-e^{-2}}{2}$ as $T\rightarrow\infty$.

\end{APPENDIX}
\end{document}